%% file: main.tex
\documentclass[11pt,letterpaper]{article}
\usepackage{etoolbox}

\newtoggle{onlineappendix}
\newtoggle{doublespacing}
\togglefalse{onlineappendix}
\togglefalse{doublespacing}

\usepackage{epsfig}
\usepackage{graphicx}
\usepackage{color}
\usepackage{figcaps}
\usepackage{hyperref}
\usepackage{titling}
\usepackage{url}
\usepackage{natbib}
\usepackage{amsmath,amsfonts,amsthm}
\usepackage{setspace}
\usepackage{rotating}
\usepackage{pstricks}
\usepackage{caption,subcaption}
\usepackage{array}
\usepackage{tikz}
\usepackage{apptools}
\usepackage{color-edits}
\usepackage{pgfplots}
\pgfplotsset{compat=1.18}
\AtAppendix{\counterwithin{condition}{section}}
\AtAppendix{\counterwithin{claim}{section}}
\AtAppendix{\counterwithin{theorem}{section}}
\AtAppendix{\counterwithin{lemma}{section}}
\usepackage{thm-restate}

\newtheorem{definition}{Definition}
\newtheorem{assumption}{Assumption}
\newtheorem{lemma}{Lemma}
\newtheorem{claim}{Claim}
\newtheorem{theorem}{Theorem}
\newtheorem{proposition}{Proposition}
\newtheorem{condition}{Condition}
\newtheorem{corollary}{Corollary}

\DeclareMathOperator*{\E}{\mathbb{E}}
\newcommand{\var}{\mathrm{Var}}
\DeclareMathOperator*{\reals}{\mathbb{R}}

\DeclareMathOperator*{\argmin}{arg\,min}
\newcommand{\one}{\mathbf{1}}

\input{Sections/macros}

\printfigures
\figcapsoff

\begin{document}
\onehalfspacing

\title{Communicating with Anecdotes\thanks{This work is partially supported by the National Science Foundation under grant CCF-2145898, the European Research Council (ERC) under the European Union's Horizon 2020 research and innovation program (grant agreement No. 866132), the Office of Naval Research under grant N00014-24-1-2159, an Alfred P. Sloan fellowship, a Schmidt Sciences AI2050 fellowship, and a Google Research Scholar award.
Some parts of this work were done when the authors visited the
Simons Institute for the Theory of Computing and when Haghtalab and Mohan were employed at Microsoft Research, New England.
We are grateful to Nageeb Ali, Lukas Bolte, Vincent Crawford, Ben Golub, Matt Jackson, Pooya Molavi, Harry Pei, Joel Sobel, Leeat Yariv and seminar participants at University of Chicago, Northwestern, Penn State, Stanford Institute for Theoretical Economics (SITE), Stony Brook Game Theory Festival, and Virtual Seminars in Economic Theory, among others, for helpful comments.
An earlier version of this paper appeared as an NBER working paper under the title “Persuading with
Anecdotes”~\citep{haghtalab2021persuading}.
}}
\author{Nika  Haghtalab\\ UC Berkeley \and Nicole Immorlica \\Microsoft Research \and Brendan Lucier \\Microsoft Research\and Markus Mobius \\Microsoft Research \and Divyarthi Mohan \\Tel Aviv University}
\date{\today}
\maketitle \thispagestyle{empty}

\begin{abstract}
We study a communication game between a sender and receiver. The sender chooses one of her signals about the state of the world (i.e., an \textit{anecdote}) and communicates it to the receiver who takes an action affecting both players. The sender and receiver both care about the state of the world but are also influenced by personal preferences, so their ideal actions can differ. We characterize perfect Bayesian equilibria. The sender faces a \textit{temptation to persuade}: she wants to select a biased anecdote to influence the receiver’s action. Anecdotes are still informative to the receiver (who will debias at equilibrium) but the attempt to persuade comes at a cost to precision. This gives rise to \textit{informational homophily} where the receiver prefers to listen to like-minded senders because they provide higher-precision signals. Communication becomes \textit{polarized} when the sender is an expert with access to many signals, with the sender choosing extreme outlier anecdotes at equilibrium (unless preferences are perfectly aligned). This polarization dissipates all the gains from communication with an increasingly well-informed sender when the anecdote distribution is heavy-tailed. Experts can therefore face a \textit{curse of informedness}: receivers will prefer to listen to less-informed senders who cannot pick biased signals as easily.
\end{abstract}

\begin{tabbing}
\end{tabbing}

\newpage

\iftoggle{doublespacing}{
\doublespacing
}

\begin{quote}
Half the truth is often a great lie.\\
\emph{-- Benjamin Franklin}
\end{quote}

\input{Sections/introduction}

\input{Sections/model}

\input{Sections/persuasion_temptation}

\input{Sections/equilibrium}

\input{Sections/experts}

\input{Sections/commitment}

\input{Sections/conclusion}

\bibliographystyle{aer}
\bibliography{Sections/anecdotes}

\newpage

\begin{appendix}

\input{Sections/appendix_diffuse_prior}

\input{Sections/apppendix_model}

\input{Sections/appendix_persuasion}

\input{Sections/appendix_equilibrium}

\input{Sections/appendix_experts}

\input{Sections/appendix_commitment}

\iftoggle{onlineappendix}{
\end{appendix}
\input{Sections/online_appendix}

}
{
\input{Sections/appendix_commitment_no_foresight}
\end{appendix}
}

\end{document}

%% file: Sections/introduction.tex
\section{Introduction}\label{sec:introduction}

As an economist, explaining the beginning of a recession  to a friend who isn't familiar with economic jargon can be difficult. While the economist usually analyzes trends using metrics like GDP growth and investment levels, her friend might have a different approach, such as noticing whether specific companies are growing or reducing their workforce. Both approaches can be valid representations of the same underlying state of the world, but the mismatched representations can create a communication challenge.

One effective approach involves both the sender and receiver reaching a consensus on a shared representation. Crafting such a shared ``language'' is a practice commonly adopted by experts within a given field, enabling them to communicate efficiently. A different method entails the sender transmitting raw data (or ``anecdotes") to the receiver, allowing them to interpret and map it to their own internal representation. Such \textit{anecdotal communication} is common when the sender or receiver are non-experts. Newspapers, for example, assemble articles from factual snippets such as reporting the quarterly job creation figures\footnote{E.g., WSJ article ``November Employment Report Shows U.S. Economy Added 263,000 Jobs''~\citep{wsjexample}}, highlighting a recent high-profile layoff at a major tech company\footnote{E.g., NYTimes article ``Meta Lays Off More Than 11,000 Employees''~\citep{facebookexmaple}}, or conducting a series of interviews with individuals who have recently resigned from their jobs\footnote{E.g., NYTimes article ``We Revisited People Who Left Their Jobs Last Year. Are They (Still) Happy?''~\cite{nytquit}}. Similarly, politicians often try to appeal to a broad spectrum of voters by incorporating stories into their speeches, such as initiatives that helped their districts in the past.

Anecdotal communication is versatile but less efficient than communication between experts because the sender cannot simply summarize diverse data points into a single number. The sender must select among the many anecdotes she could potentially report. For example, the New York Times must sift through a vast array of possible anecdotes from the work of its journalists and wire services to determine what to include in a 500-word article on its homepage. This need for selection creates a \textit{temptation} for the sender to influence the receiver. For instance, a current officeholder may predominantly incorporate favorable anecdotes about the economy into their speeches instead of the most representative ones, motivated by their bid for re-election. Conversely, their opponent may utilize anecdotes to portray a less optimistic outlook, aiming to cast doubt on the incumbent's abilities. This \textit{persuasion temptation} makes anecdotal communication more noisy: a rational receiver can correct for the sender's bias but not for the greater noise that is embedded when listening to less representative anecdotes. 

\begin{figure}
\caption{Illustration of translation invariant equilibrium}  \label{fig:intuition}
\begin{center}
\input{Figures/intuition_equilibrium}
\end{center}
\end{figure}

In this paper, we build a simple model of anecdotal communication that explores the tension between the sender choosing representative anecdotes that are most informative and tail anecdotes that are meant to persuade the receiver to take a more preferred action.

\subsection{Summary of Results}

We summarize the key findings using the most straightforward version of our model as shown in Figure \ref{fig:intuition}.  There are two agents - a sender $S$ and a receiver $R$. The receiver takes an action that affects both agents but the state of the world $\theta$ is only observed by the sender. The sender would like the receiver to take action $a_S^*=\theta+M_S$ (shown in blue) while the receiver's preferred action is $a_R^*=\theta+M_R$ (shown in red). The offsets $M_S$ and $M_R$ indicate the agents' \textit{personal preferences} and in the above example we have $M_S<M_R$ such that the sender would always like the receiver to take an action slightly more to the left. The sender has access to a set of signals which are drawn from a distribution that is centered around the state $\theta$ (shown as gray circles). We call these signals \textit{anecdotes} because the sender is constrained in her communication with the receiver and can only send one such anecdote to the receiver. 

\paragraph{Targeting Equilibrium.}
We will focus on \textit{translation invariant} strategies where the receiver takes action $x+\sigma$ after observing the sender's anecdote $x$. Intuitively, the receiver acts as if he knows the sender's bias and corrects for it; for example, the receiver knows that a news publisher will tend to select more left-leaning anecdotes and correct for this bias by choosing an action further to the right ($\sigma>0$). In equilibrium, the sender would therefore ideally want to send an anecdote exactly equal to $\theta+r$ where $r = M_S-\sigma$, as this would ensure that the receiver takes the sender's preferred action. However, the sender will generally not have this exact anecdote available.  
Our first main result is to show that she will select the anecdote closest to $\theta+r$ as the second-best alternative (see Theorem \ref{thm:sender.br}). Next, we observe that the anecdote closest to $\theta+r$ (shown as a green circle here) is more likely to lie to the right of the offset than to the left because the density of the anecdote distribution is higher there. Therefore, the expected value of the sent anecdote is equal to $\theta+\beta(r)$ where $\beta(r)>r$ is the \textit{bias} of the sender's communication. In equilibrium, the receiver will choose his shift $\sigma$ to undo this bias and then take her preferred action such that $\sigma=-\beta(r)+M_R$. Combining the sender's best response condition with the receiver's, we obtain:
\begin{eqnarray}
    r&=&M_S-\underbrace{(M_R-\beta(r))}_{\sigma} \nonumber \\
    \beta(r)-r &=& \underbrace{M_R-M_S}_{\Delta}
\end{eqnarray}
This is the key equilibrium condition in our model and all our main results follow from it: for a given difference $\Delta$ in the personal preferences of receiver and sender there exists a \textit{targeting communication scheme} with offset $r$ satisfying the above equation such that the sender will choose the anecdote closest to $\theta+r$ (see Theorem \ref{thm:bne.exist}).

\paragraph{Informational Homophily.} 
When the sender and receiver have common interests, i.e., $M_R = M_S$, the above condition boils down to $\beta(r)=r$. For a single-peaked and symmetric distribution around the state $\theta$, as shown in Figure \ref{fig:intuition}, this implies that the sender will choose offset $r=0$ and therefore select the anecdote closest to the $\theta$. Intuitively, this anecdote is the most representative and conveys the most information to the receiver. When the sender is more left-leaning than the receiver ($\Delta>0$) we instead obtain an equilibrium offset $r<0$ and the sender will select biased anecdotes with a left-leaning bias. While the receiver is able to correct for this bias, he is worse off because these anecdotes have higher variance. In Section \ref{sec:equilibrium}, we explore how this gives rise to \textit{informational homophily} where both the sender and receiver prefer to talk to like-minded partners because lack of persuasion temptation makes the sender minimize the information loss from anecdotal communication.

\paragraph{Polarization.} 
Unlike models in the cheap talk and rational persuasion literature, the sender in our model is constrained in her message space to select an existing anecdote \citep{crawford1982strategic, kamenica2011bayesian}. This enables us to define the degree of polarized communication by the offset $r$ of the sender's targeting communication scheme: the further to the left (or right) this offset lies, the more polarized is the sender's communication. As explained in the example above, polarization increases with the amount of misalignment between the sender and receiver (as measured by $|\Delta|$). However, in Section \ref{sec:experts} we show that polarization also increases with the number of anecdotes available to the sender: the more informed the sender is, the more polarized her communication becomes. Intuitively, the information loss decreases with the number of anecdotes at any given offset $r$ (formally, $\beta(r)-r$ decreases). This makes it less costly for the sender to choose a more extreme anecdote and the equilibrium targeting scheme will become more polarized. Interestingly, while any misalignment between the sender and receiver will generate some degree of polarization this phenomenon is more pronounced for more informed experts (i.e., senders with access to more anecdotes). 

\paragraph{Curse of Informedness.} 
Section \ref{subsec:informedness} explores whether a receiver prefers to listen to a less informed but more like-minded sender or a better-informed expert. While communication with the expert is more polarized (which lowers utility) the expert is also better informed (which lowers the information loss). We show that for heavy-tailed distributions (with tail densities declining at a less than exponential rate) receivers prefer to seek out non-experts. Intuitively, experts in such an environment have access to too many outlier anecdotes such that the variance in the sent anecdote increases with the informedness of the expert. We say that such experts suffer from the \textit{curse of informedness} - they have an incentive to collect fewer anecdotes in order to be listened to more by less aligned receivers. 

\paragraph{Commitment.} 
Section \ref{sec:commitment} analyzes a version of our model where the sender can commit to a targeting communication scheme. Commitment eliminates persuasion temptation and therefore gives rise to most informative communication schemes that maximize welfare. For example, in the example above the sender will send the most representative anecdote which is the anecdote closest to the posterior mean.

\subsection{Related Work}
Our model is related to the economic literature on strategic communication which includes both cheap talk games and verifiable disclosure games (also called persuasion games). In the canonical paper in the literature on cheap talk games, \citet{crawford1982strategic} consider a setting similar to ours in which the sender and receiver have state-dependent but misaligned preferences.  The major difference in their model is that the sender is unrestricted in what signal she can send and hence can pool states arbitrarily.  The authors characterize the set of perfect Bayesian equilibria and show the most informative of these equilibria pools only nearby states. The resulting precision of the receiver's posterior belief, interpreted as a coarseness of the message space or vagueness of the chosen language in equilibrium, depends on the degree of misalignment between the sender and receiver's preferences.  We observe a similar phenomenon in our setting: more aligned preferences induce more information transmission in equilibrium.  Our results however point to a different driver of this phenomenon.  Namely, the friction that prevents communication from devolving to the least informative signal is the exogenous limitation of communicating anecdotes rather than the endogenous equilibrium choice of language coarseness.

Our exogenous limitation of communicating anecdotes is reminiscent of the literature on verifiable disclosure games or persuasion games, introduced by \citet{grossman1980disclosure}, \citet{grossman1981informational}, and \citet{milgrom1981good} (see \citet{milgrom2008seller} for a survey of this literature).  These papers consider the setting of a seller who can choose whether to disclose information about a product to a buyer and wishes to maximize the buyer's posterior belief about the value or quality of the product. Similar to our model of communication via anecdotes, the seller in these papers cannot arbitrarily distort information about the product. The classic results show that in every perfect Bayesian equilibrium, the seller fully discloses her information or, if she's limited in how much information she can disclose, reveals the most favorable information.  In our model, this sort of unraveling to extreme signals is tempered by the fact that the sender's preferences are state dependent and so she does not simply wish to maximize the receiver's belief.  Other mechanisms that limit unraveling include the presence of naive receivers (\citet{kartik2007credulity}), exogenous costs to misrepresentation (\citet{kartik2009strategic}), uncertainty in the informedness of the sender (\citet{dye1985disclosure,jung1988disclosure,dziuda2011strategic}), multi-dimensionality of the state (\citet{martini2018multidimensional}), or the introduction of commitment power (discussed below), among others.

Much of the literature on strategic communication studies the impact of commitment power.  
\citet{kamenica2011bayesian} characterize the optimal signaling scheme of a sender with commitment power in cheap talk games.  Their characterization implies that a sender with quadratic loss (as in our model) who is not restricted to communicating anecdotes would communicate her belief about the state of the world thereby completely avoiding the uninformative babbling equilibrium of \citet{crawford1982strategic}.  In Section~\ref{sec:commitment}, we show similarly that, subject to the restriction of sending an anecdote, the sender wishes to communicate the most informative anecdote.  This result suggests that in our model, commitment power can move the equilibrium from a partially informative one to the most informative one. Other recent work demonstrates investigates commitment power in verifiable disclosure games, studying the impact of partial commitment (\citet{lipnowski2022persuasion,min2021bayesian,nguyen2021bayesian,lin2022credible}) or characterizing conditions under which commitment has limited impact on achievable outcomes in equilibrium (\citet{glazer2008study,hart2017evidence,sher2011credibility,zhang2022withholding}).
Importantly, in all of these models the sender's incentives are purely to persuade, with utility that depends on the receiver's action but not the state of the world.  In our model the sender has an incentive to inform as well as to persuade. Thus in light of the aforementioned literature, our results can be interpreted as showing that a desire to persuade can lead to a significant loss of communication fidelity (and hence welfare) in the absence of commitment.  Our perfect Bayesian equilibrium characterization in Section~\ref{sec:equilibrium} shows that this loss can happen and be significant with even with a small desire to persuade (i.e., even with a small preference misalignment).

%% file: Figures/intuition_equilibrium.tex
\begin{tikzpicture}[scale = 0.85]

\fill [shading = axis,left color=green!30!white, right color=white,color=white, domain=2:4, variable=\x]
  (2, 0)
  -- plot ({\x}, {4*exp(-(\x-7.5)*(\x-7.5)/2/pi/2/2)})
  -- (4, 0)
  -- cycle;

\fill [shading = axis,right color=green!30!white, left color=white,color=white, domain=2:1, variable=\x]
  (2, 0)
  -- plot ({\x}, {4*exp(-(\x-7.5)*(\x-7.5)/2/pi/2/2)})
  -- (1, 0)
  -- cycle;

  \draw[domain=1:14, smooth, variable=\x, black] plot ({\x}, {4*exp(-(\x-7.5)*(\x-7.5)/2/pi/2/2)});

  \draw[-] (-1, 0) -- (15, 0);
  
  \draw[-stealth,ultra thick,black]    (7.5,-1)   -- (7.5,-0.2) node[below=25] {$\theta$};
	\draw[-stealth,ultra thick,red]    (3.7,-1)   -- (3.7,-0.2) node[below=25] {$\theta+\beta(r)$};
	\draw[-stealth,ultra thick,blue]    (2,-1)   -- (2,-0.2) node[below=25] {$\theta+r$};

  	\draw[dashed, thick, blue]    (2,0)   -- (2,4.3);
  	\draw[dashed, thick, red]    (3.7,0)   -- (3.7,4.3);

      \draw[dashed, thick, blue]    (6.65,0)   -- (6.65,4.3);
  	\draw[dashed, thick, red]    (8.35,0)   -- (8.35,4.3);

    \draw [stealth-stealth, thick, purple](2,4.2) -- (3.7,4.2) node[above=5,midway] {$\Delta$};

    \draw [stealth-stealth, thick, purple](6.65,4.2) -- (8.35,4.2) node[above=5,midway] {$\Delta=M_R-M_S$};

    \draw [stealth-stealth, thick, blue](2,3) -- (6.65,3) node[above=5,midway] {$\sigma$};

    \draw [stealth-stealth, thick, red](3.7,2) -- (8.35,2) node[above=5,midway] {$\sigma$};

  	\draw[dashed, thick, black]    (7.5,0)   -- (7.5,{4*exp(-(7.5-7.5)*(7.5-7.5)/2/pi/2/2)});

	\draw[-stealth,ultra thick,blue]    (6.65,-1)   -- (6.65,-0.2) node[below=25] {$a_S^*$};

 \draw[-stealth,ultra thick,red]    (8.35,-1)   -- (8.35,-0.2) node[below=25] {$a_R^*$};
   
  \node[red] at (8.35,-2.5) {$\theta+M_R$};
  \node[blue] at (6.65,-2.5) {$\theta+M_S$};

   \foreach \c in {1,2.55,2.95,3.4,4,5,8,8.5,9.5,12}{ %
	\fill[black!20] (\c,0) circle[radius=0.15]; %
   }
	\fill[black!30!green] (2.55,0) circle[radius=0.15]; %

 \draw[-stealth,ultra thick,black!30!green]    (2.55,-2.3)   -- (2.55,-0.2) node[below=50] {
 \begin{minipage}{2.7 cm}
 closest anecdote to $\theta+r$
 \end{minipage}
 };

 \node[-stealth,,black!40!green] at (2.4,0.6) {
 \begin{minipage}{1.9 cm}
 \tiny
 density of closest anecdote to $\theta+r$
 \end{minipage}
 };

\end{tikzpicture}

%% file: Sections/model.tex
\section{A Model of Anecdotal Communication}\label{sec:model}

In this section we formally introduce our model and characterize the sender's and receiver's best response strategies which will allow us to characterize the equilibria of the game in Section \ref{sec:equilibrium}.

\subsection{Model Setup}

We consider a communication game played by two players, a sender (``she'') and a receiver (``he'').  The sender has information about a payoff-relevant state of the world $\theta\in\reals$ drawn from a common prior. The receiver, in turn, chooses a payoff-relevant action $a \in\reals$ 
(for example, how much to invest in the current market).

\paragraph{Preferences.} Players' preferences over their actions depend on the state of the world.  However, their preferences can differ.

We model this by introducing personal preferences $M_R \in \reals$ and $M_S \in \reals$ for the receiver and sender respectively, which are shifts of the ideal action relative to the state of the world.  More formally, the receiver's utility is $$u_R(a,\theta) = -(a-(\theta+M_R))^2,$$ and the sender's utility is $$u_S(a,\theta) = -(a-(\theta+M_S))^2.$$  

We assume personal preferences are publicly known and write $\Delta = M_R - M_S$ for the known difference in personal preferences. Intuitively, $\Delta$ captures the preference misalignment between sender and receiver.

\paragraph{Sender's knowledge.}  The sender has access to noisy signals about the state of the world that she can potentially share with the receiver. Given a  distribution $F$ over the reals, we model these shareable signals as a set of $n$ samples $x_1, \dotsc, x_n$ where each $x_i=\theta+\epsilon_i$ for $\epsilon_i\sim F$ drawn independently. We will write $\vec{x} = (x_1, \dotsc, x_n)$ for the profile of samples which we will refer to as \textit{anecdotes} from now on. We think of these anecdotes as immutable facts about the world which the sender can decide to share, but which she cannot otherwise manipulate. For example, the receiver might not know about the survey or research paper until the sender chooses to reveal it but he can subsequently look up the survey or paper and fact-check it. While $\vec{x}$ is known only to the sender, we assume the anecdote distribution $F$ as well as the number of anecdotes, $n$, is common knowledge.\footnote{This shuts down a common pathway for partial information transmission: in our model, there is no uncertainty about \textit{how much} information the sender has.}

The sender might have additional information that cannot be easily shared or fact-checked at low cost. For example, the sender's knowledge about the state of the world might be informed by her own detailed research and modeling efforts.  We model such side information as an additional signal $y\in \reals \cup \left\{\emptyset\right\}$. Given a  distribution $G$ over $\reals \cup \left\{\emptyset\right\}$ and $\gamma\sim G$ , the sender has either access to no additional information (if $\gamma=\emptyset$) or to a signal $y=\theta+\gamma$ (if $\gamma \in \reals$).  Most of our intermediate results hold for general distributions $G$. However, we pay special attention to two cases: the \textit{ foresight} setting where $\gamma=0$ with probability $1$, and the sender has full information about the state of the world; second, the setting with \textit{no foresight}, where $\gamma=\emptyset$ and the sender has no information beyond the set of anecdotes. As with $\vec{x}$, we assume that $y$ is private knowledge of the sender but that the distribution $G$ is common knowledge. The sender uses her anecdotes $\vec x$ and side information $y$ to form a posterior belief over the state of the world $\theta$. We will denote by $\theta_S(\vec{x},y)$ the posterior mean of $\theta$ given $(\vec{x},y)$.

\paragraph{Communication.} The sender communicates exactly one anecdote in $\vec{x}$ to the receiver. While the anecdote is communicated honestly, the sender can cherry-pick from the set of anecdotes she has access to. Note, that the sender's side information $y$ cannot be communicated -- only an anecdote can be shared. A strategy for the sender in our game is a communication scheme $\pi \colon \reals^n \times \reals \to \reals$ that maps every realization of $n$ anecdotes $\vec x$ and side information $y$ to a choice of one of the $n$ anecdotes.  In particular, for all  $\vec{x}$ and $y$ we have $\pi(\vec{x}, y) = x_i$ for some $i \in [n]$. 

\begin{figure}
\caption{Timing of communication game}  \label{fig:timing}
\begin{center}
\begin{tikzpicture}[scale=1,every node/.style={scale=1}]
\draw[->,color=black] (0,0) -- (14.5,0);

\draw[color=black] (0,-0.25) -- (0,0.25);
\draw[color=black] (0,1) node {Round 0};
\draw[color=black,anchor=north] (0,-0.5) node {\begin{minipage}{4cm}
                                 Nature chooses state $\theta$, anecdotes $\vec x$ and $y$.
                               \end{minipage}};

\draw[color=black] (5.5,-0.25) -- (5.5,0.25);
\draw[color=black] (5.5,1) node {Round 1};
\draw[color=black,anchor=north] (5.5,-0.5) node {\begin{minipage}{4cm}
                                 Sender sends anecdote $\tilde{x}=\pi(\vec x,y)$.
                               \end{minipage}};

\draw[color=black] (11,-0.25) -- (11,0.25);
\draw[color=black] (11,1) node {Round 2};
\draw[color=black,anchor=north] (11,-0.5) node {\begin{minipage}{5.5cm}
                                 Receiver takes action $a=\alpha(\tilde{x})$.\\
                                 $u_S(a,\theta) = -(a-(\theta+M_S))^2$\\
                                 $u_R(a,\theta) = -(a-(\theta+M_R))^2$
                               \end{minipage}};

\end{tikzpicture}
\end{center}
\end{figure}

\paragraph{Equilibrium.}  Since the receiver does not observe the choice of nature, a strategy for the receiver is an action rule $\alpha \colon\reals \to\reals$ that maps the sender's chosen anecdote to a choice of action. The timing of our game is shown in Figure \ref{fig:timing}.  In round $0$, nature chooses state $\theta$, anecdotes $x_1, \dotsc, x_n$, and signal $y$. In round $1$, the sender selects anecdote $\tilde{x}=\pi(\vec x,y)$; this choice is observed by the receiver.  In round $2$, the receiver selects action $\alpha(\tilde{x})$.  Payoffs are then realized as described above.   Given a communication scheme $\pi$, we will write $D_{\pi, x}$ for the posterior distribution of $\theta$ given that $\pi(\vec{x}, y) = x$.  We are interested in the perfect Bayesian equilibrium of the game, that is, strategies for the sender and receiver that maximize payoffs under their (consistent) beliefs.\footnote{The perfect Bayesian equilibria we consider will not involve zero-probability events and will therefore also be sequential equilibria.}

\begin{definition} \label{def:bne-def}
A pair of strategies $(\pi^*, \alpha^*)$, together with a belief function $B \colon \reals \to \Delta(\reals)$ for the receiver mapping every observation to a distribution over the state of the world, form a perfect Bayesian equilibrium if: 
\begin{enumerate}
    \item For each observed anecdote $x$, action $\alpha^*(x)$ maximizes expected receiver utility given distribution $B(x)$ over $\theta$, i.e.,
    $$\alpha^*(x) \in \arg\max_a \left\{ \E_{\theta \sim B(x)}[ u_R( a, \theta ) ] \right\}.$$ \label{item:receiver1}
    \item {$B$ is the \emph{rational belief} with respect to $\pi^*$. That is, for each $x$, $B(x) = D_{\pi^*,x}$ is} the posterior distribution of $\theta$ given that $\pi^*(\vec{x},y) = x$.  \label{item:receiver2}
    \item For each $\vec{x}$ and $y$, $\pi^*(\vec{x}, y)$ maximizes sender utility given $\alpha^*$, i.e., $$\pi^*(\vec{x}, y) \in \arg\max_{x_i \in \vec{x}} \left\{ \E_{\theta}[u_S( \alpha^*(x_i),\theta )\ |\ (\vec{x},y)] \right\}.$$ \label{item:senderbr}
\end{enumerate}
\end{definition}
\noindent Given communication scheme $\pi$, we denote by $\alpha_\pi$ the action rule that satisfies requirements~(\ref{item:receiver1}) and (\ref{item:receiver2}) of Definition~\ref{def:bne-def} and call it the \emph{best response} to $\pi$.

\paragraph{Diffuse Prior.} Expected payoffs in our game are driven by the posterior beliefs of the receiver and sender regarding the state of the world $\theta$.  These posterior beliefs are formally defined with respect to a prior distribution over $\theta$.  As we are interested in the setting where agents are initially uninformed about the state of the world, we will define posterior beliefs with respect to a \emph{diffuse prior} over $\theta$.  That is, the common prior over $\theta$ reveals no information about it.
\begin{assumption} \label{ass:diffuse}
The prior over $\theta$ is diffuse, i.e., it is $N(0, \infty)$.
\end{assumption}
We note that the exact form of the prior over $\theta$ is not important so long as it is a diffuse prior that reveals no information about $\theta$, i.e., its density is almost uniform everywhere.  

Of course, the diffuse prior is not a valid probability distribution and is not formally defined.  We emphasize that it is being used to define posterior beliefs given the information available to the sender and receiver.\footnote{One could equivalently view the realization of $\theta$ as a (non-stochastic) choice of nature, and interpret Assumption~\ref{ass:diffuse} as a behavioral assumption that the agents form posterior beliefs consistent with having no prior knowledge of $\theta$.}
As has been noted by \cite{ambrus2021defining}, one must take care when defining ex-ante payoffs in any game involving a diffuse prior.  They describe sufficient conditions for ex ante payoffs to be well-defined and consistent with any appropriate limiting sequence of proper prior distributions.  In Appendix~\ref{app:diffuse-prior} we verify that ex ante payoffs are likewise well-defined in our model and discuss further implications of Assumption~\ref{ass:diffuse}, including the precise form of the posterior beliefs.

\subsection{Characterizing Best Response Strategies} \label{sec:translation_invariance}

Our analysis focuses on a 
class of \textit{translation invariant} strategies for both the sender and the receiver. Intuitively, a translation invariant communication scheme encodes a sender's \textit{communication posture} such as ``always send the left-most anecdote'' or ``always send the right-most anecdote'' that does not depend on the particular realization of the state of the world. This is an appealing property in the context of our diffuse prior assumption where specific numbers have no special meaning.

Similarly, a translation invariant action rule describes the receiver's action as a fixed offset from the received anecdote. Such an action rule describes a receiver  who simply believes the anecdote he hears is representative of the state, albeit potentially with a shift. We will show in this section that translation invariance is internally consistent in the sense that the best response of the sender to a translation invariant action rule is a translation invariant communication scheme and vice versa.

We formally define translation invariant action rules as follows.
\begin{definition}
An action rule $\alpha$ is a translation if $\alpha(x+\delta) = \alpha(x)+\delta$ for all $x$ and all $\delta \in \reals$.
\end{definition}
\noindent Note that if action rule $\alpha$ is a translation, then there is a value $\sigma \in \reals$ such that $\alpha(x) = x + \sigma$ for all $x$.  We refer to $\sigma$ as the \emph{shift} of $\alpha$, written $\sigma(\alpha)$. These receivers act as if they know the typical bias of a sender, e.g.,
a receiver who thinks the New York Times, being a slightly left-of-center paper, sends anecdotes
shifted slightly left.

We next formally define translation invariant communication schemes. Given a profile of anecdotes $\vec{x}$ and a constant $\delta \in \reals$ we will write $\vec{x}+\delta$ for the shifted profile of anecdotes $(x_1 + \delta, x_2 + \delta, \dotsc, x_n + \delta)$.
\begin{definition}\label{def:trans_invariance}
A communication scheme $\pi$ is translation invariant if $\pi(\vec{x}+\delta,y+\delta) = \pi(\vec{x},y)+\delta$ 
for all $\vec{x},y$ and all $\delta \in \reals$.\footnote{We assume that $y+\delta = \emptyset$ for $y=\emptyset$.}
\end{definition}
\noindent Not all communication schemes are translation invariant.  For example, the communication scheme that sends the anecdote closest to zero is not translation invariant, nor is the one that sends the minimum anecdote if that anecdote is an even number and the maximum otherwise.  However, many natural communication schemes are translation invariant. The top panel of Figure \ref{fig:invariant_schemes} shows the simplest example: the \textit{minimum scheme} $\pi_{\min}(\vec{x},y)$ selects the minimum anecdote (analogously, the \textit{maximum scheme} $\pi_{\max}(\vec{x},y)$ selects the maximum anecdote). 

\begin{figure}[ht]
\caption{Examples of translation invariant communication schemes}  \label{fig:invariant_schemes}
\begin{center}
\input{Figures/targeting_schemes}
\end{center}
\end{figure}

A particularly important class of translation invariant communication schemes are those where the sender selects a signal that is closest to a shift $r$ from her posterior mean. We call these \textit{targeting schemes}. 
\begin{definition}
The targeting scheme with offset $r \in \reals$ is a communication scheme that always returns the anecdote from $\vec{x}$ that is closest to $\theta_S(\vec{x},y) + r$.
\end{definition}
\noindent Note that since $\theta$ is drawn from a diffuse prior, we have $\theta_S(\vec{x}+\delta,y+\delta) = \theta_S(\vec{x},y)+\delta$ (formalized in Appendix \ref{app:diffuse-prior}). Hence a targeting scheme is translation invariant. The middle panel of Figure \ref{fig:invariant_schemes} shows the \textit{mean targeting scheme} for $r=0$ where the sender selects the signal closest to the posterior mean. The bottom panel illustrates a targeting scheme with $r<0$. Note that as $r\rightarrow -\infty$ the targeting scheme approaches the minimum scheme.

It will be helpful to define the \textit{bias} of a translation  invariant communication scheme.  Given a translation invariant communication scheme $\pi$, we will say the bias of $\pi$, $\beta(\pi)$, is equal to $\E_{\theta,\vec{x},y}[\pi(\vec{x},y) - \theta]$. If the anecdote distribution $F$ is symmetric around $0$ then the minimum scheme has a left bias ($\beta(\pi_{\min})<0$) and a targeting scheme $\pi_{r}$ has a left bias for $r<0$ and a right bias for $r>0$ (and will be unbiased for $r=0$). 

We can now state the main result: the sender's best response to a translation invariant action rule is not just translation invariant but also a targeting scheme.
\begin{theorem}\label{thm:sender.br}
If action rule $\alpha$ is a translation, then the best response of the sender is translation invariant.  More specifically, it is the targeting scheme with offset $M_S - \sigma(\alpha)$.
\end{theorem}
\noindent Note that Theorem \ref{thm:sender.br} characterizes the sender's best response among all (not necessarily translation invariant) communication schemes. All the proofs in this section are relegated to Appendix \ref{app:model}. The intuition for the result is as follows: assume that the sender had exactly an anecdote at distance $r$ from the posterior mean. The recipient will then take action $a=\theta_S(\vec{x},y)+\sigma(\alpha)$. If $r=M_S - \sigma(\alpha)$ we get $a=\theta_S(\vec{x},y)+M_S$ which is exactly the loss-minimizing action from the sender's perspective. In general, the sender will not have exactly an anecdote at the offset available and so she chooses the closest anecdote as a second-best -- for this reason, the sender's best response is a targeting scheme.

Also, note that the sender's best response is generally unique \textit{except} when there is foresight and the anecdote distribution $F$ is bounded. For example, if $\underline{\epsilon}$ is the lower bound of the domain of $F$ and $\sigma(\alpha)\leq \underline{x}$ then the sender's optimal targeting scheme is the minimum scheme and any $r<\underline{x}$ will be optimal (including the one prescribed by Theorem \ref{thm:sender.br}).

We close out our analysis of translation invariant strategies by showing that the best response to a translation invariant communication scheme is a translation invariant action rule.
\begin{theorem}\label{thm:receiver.br}
For any translation invariant communication scheme $\pi$, the best response of the receiver to $\pi$ is a translation with shift $M_R-\beta(\pi)$. 
\end{theorem}
\noindent The intuition is as follows. The receiver knows that the sent anecdote is on average a distance $\beta(\pi)$ away from the sender's posterior mean. Hence, he will try to undo this bias by subtracting it from the sent anecdote which provides her with an unbiased estimate of the state of the world. He will then minimize her loss by taking an equal to this estimate plus her personal preference $M_R$.

We have now established that sender and receiver translation invariant strategies are self-consistent in the sense that they are best responses to each other. We have not yet shown that equilibria with translation invariant strategies exist. We will show their existence in Section \ref{sec:equilibrium}. 

\subsection{Broadcasting}\label{subsec:broadcasting}

Senders often communicate anecdotes to several receivers simultaneously. Experts or politicians, for instance, can appear on TV or are cited in newspapers, reaching many viewers or readers. Our model can easily be extended to allow for \textit{broadcasting} to multiple receivers whose personal preference $M_R$ is drawn from a distribution $G$ over the real numbers. Each type $M_R$ is taking an action $a(\tilde{x}|M_R)$ after observing anecdote $\tilde{x}$ which is ``broadcast'' to all receiver types. The sender cares about the \textit{mean} action of the receivers:
\begin{equation}
    \overline{a}=\int a(\tilde{x}|M_R) dG
\end{equation}
We denote the mean personal preferences of the receivers with $\overline{M}_R$:
\begin{equation}
    \overline{M}_R=\int M_R dG
\end{equation}
Our notion of translation invariant communication schemes readily generalizes to this setting.
\begin{proposition}\label{prop:broadcast}
In any translation invariant communication scheme of the broadcast model, the best response of the receiver type $M_R$ is a translation with shift $\sigma(M_R)=M_R-\beta(\pi)$. The sender's best response is the targeting scheme with offset $M_S-\overline{\sigma}$ where $\overline{\sigma}=\int \sigma(M_R)dG$ is the mean shift across receiver types.
\end{proposition}
\noindent Therefore, the sender in the broadcast model behaves as if she is communicating with a single ``representative'' receiver, whose personal preference matches the average personal preferences of the entire audience.

\subsection{Discussion of Modeling Assumptions}\label{sub:modeling assumptions}

\noindent\textbf{Crawford-Sobel (1982).} Our model follows \citet{crawford1982strategic} except that we restrict communication to anecdotes. This restriction is most appropriate in situations where the sender or receiver are non-experts -- for example, newspapers communicating news to readers or politicians talking to voters. Experts often share a \textit{common language} that allows them to communicate beliefs efficiently. For example, experts on infectious diseases use the reproduction factor $R_0$ to describe how quickly a disease such as measles can spread between people. An expert sender can therefore summarize her knowledge about measles, for example, by communicating this factor to the receiver. A parent, however, might find it easier to understand infectiousness by being told that a neighboring school had to cancel classes because too many teachers called in sick. A journalist might therefore relate the expert's message by listing examples of school districts who experienced measles outbreaks and interviews with local school principals and teachers.

\noindent\textbf{Common Knowledge.} The assumption that sender's preference is publicly known is justifiable in settings where the sender is a known entity, say a politician or newspaper.  In such settings, the sender is often communicating with a known distribution of receiver types -- the general public for instance.  Our results would follow largely unchanged if the receiver's preference is drawn from a known distribution.

\noindent\textbf{No ``Fake News''.} We assume that anecdotes cannot be manipulated or falsified but can only be selected in a possibly biased manner. The non-falsification assumption is reasonable in public discourse: politicians and newspapers by-and-large report facts, or else risk being caught by fact-checkers. However, they have editorial control over the selection of those facts and can influence the listener this way. Unlike lying about and making up facts, biased selection of facts is typically accepted in public discourse. The non-manipulation assumption is stronger: certainly two senders can sometimes frame two facts differently. A full analysis of framing exceeds the scope of our paper but we note that our framework can accommodate some natural forms of framing: for example, if the sender can add a ``spin'' to the sent anecdote $\tilde{x}$ by shifting its value by a fixed constant $s$ then the receiver observes $\tilde{x}+s$ and can undo this spin as long as $s$ is common knowledge (which is natural if the sender's personal preference is also common knowledge).

\noindent\textbf{Commitment.} In Section \ref{sec:commitment} we will study a variant of the game where the sender can \textit{commit} to a communication scheme. A sender may be able to commit in settings where her reputation precedes her. In this case an equilibrium does not need to {meet requirement~(\ref{item:senderbr}) of Definition \ref{def:bne-def}, i.e.,  $\pi^*(\vec{x}, y)$ does not need to  maximize expected sender utility given $\alpha^*(x)$.}  Rather, for every $\pi$, we fix a best response $\alpha_\pi$ of the receiver. We then require that
\begin{enumerate}
\item[\textit{3'.}] {$\pi^*$ is such that}
\begin{equation}\label{def:commitment}
\pi^* \in \arg\max_\pi \left\{ \E_{\theta, \vec{x}, y}[ u_S( \alpha_\pi(\pi(\vec{x}, y)),\theta ) ] \right\}.    
\end{equation}
\end{enumerate}

\noindent\textbf{Model of Memory.} We can reinterpret our framework as a model of \textit{memory} where the current self (sender) decides which data point to commit to her long-term memory that is available to her future self (receiver). This interpretation is particularly intuitive for the case of commitment or where the current and future self share the same personal preferences ($\Delta=0$).

%% file: Figures/targeting_schemes.tex
\begin{tikzpicture}[scale=0.85]
  \draw[-] (-1, 0) -- (15, 0);
  \draw[domain=1:14, smooth, variable=\x, black] plot ({\x}, {2*exp(-(\x-7.5)*(\x-7.5)/2/pi/2/2)});
   \foreach \c in {0,1,3,3.5,5,8,8.5,9,9.5,12}{ %
	\fill[black!20] (\c,0) circle[radius=0.15]; %
   }
  \draw[-stealth,ultra thick]    (7.5,-1)   -- (7.5,-0.2) node[below=25] {$\theta$};
	\fill[black!30!green] (0,0) circle[radius=0.15];
	\draw[-stealth,ultra thick,black!30!green]    (0,-1)   -- (0,-0.2) node[below=25] {$\pi_{\min}(\vec{x})$};
  \node at (7.5,1) {Minimum scheme};
\end{tikzpicture}
\begin{tikzpicture}[scale=0.85]
  \draw[-] (-1, 0) -- (15, 0);
  \draw[domain=1:14, smooth, variable=\x, black] plot ({\x}, {2*exp(-(\x-7.5)*(\x-7.5)/2/pi/2/2)});
   \foreach \c in {0,1,3,3.5,5,8,8.5,9,9.5,12}{ %
	\fill[black!20] (\c,0) circle[radius=0.15]; %
   }
  \draw[-stealth,ultra thick]    (7.5,-1)   -- (7.5,-0.2) node[below=25] {$\theta$};
	\fill[black!30!green] (8,0) circle[radius=0.15];
	\draw[-stealth,ultra thick,black!30!green]    (8,-.8)   -- (8,-0.2) node[below=15] {$\pi_{0}(\vec{x})$};
  \node at (7.5,1) {Mean scheme};
\end{tikzpicture}
\begin{tikzpicture}[scale=0.85]
  \draw[-] (-1, 0) -- (15, 0);
  \draw[domain=1:14, smooth, variable=\x, black] plot ({\x}, {2*exp(-(\x-7.5)*(\x-7.5)/2/pi/2/2)});
   \foreach \c in {0,1,3,3.5,5,8,8.5,9,9.5,12}{ %
	\fill[black!20] (\c,0) circle[radius=0.15]; %
   }
  \draw[-stealth,ultra thick]    (7.5,-1)   -- (7.5,-0.2) node[below=25] {$\theta$};
	\fill[black!30!green] (3,0) circle[radius=0.15];
	\draw[-stealth,ultra thick,black!30!green]    (3,-.8)   -- (3,-0.2) node[below=15] {$\pi_{r}(\vec{x})$};
	\draw[-stealth,ultra thick,black!30!green]    (2.1,-1)   -- (2.1,-0.2) node[below=25] {$\theta+r$};
  \node at (7.5,1) {Targeting scheme};
\end{tikzpicture}

%% file: Sections/persuasion_temptation.tex
\section{Most Informative Communication and Persuasion Temptation}\label{sec:persuasion}

Before characterizing the equilibria of the game it is helpful to highlight a fundamental tension in the sender's objective: on the one hand she wants to send an informative anecdote so that the receiver takes an action that reflects the state of the world. On the other hand, if the anecdote is too precise then the receiver will take an action close to his ideal action $\theta + M_R$, which differs from the sender's preferred action unless their preferences are aligned.

We call this tension the sender's \textit{persuasion temptation}. To understand it better, it is helpful to decompose the sender's utility into two components.
\begin{proposition}\label{prop:sender_summary}
Suppose the receiver's action rule $\alpha$ is a translation with shift $\sigma(\alpha)$.
Then the sender's expected utility from any translation invariant communication scheme $\pi_S$ is
\begin{equation}
-\underbrace{\E\left[(\pi_S(\vec{x},y)-(\theta_S(\vec{x},y)+\beta(\pi_S)))^2\right]}_{\mbox{information loss}}-\underbrace{\left[\sigma(\alpha) - (M_S - \beta(\pi_S)) \right]^2}_{\mbox{disagreement loss}}.    
\end{equation}
\end{proposition}
\noindent The first component of this decomposition, $\E[(\pi_s(\vec{x},y)-(\theta_S(\vec{x},y)+\beta(\pi_S)))^2]$, is the variance of the communicated anecdote as $\theta_S(\vec{x},y)+\beta(\pi_S)=\E[\pi_s(\vec{x},y)]$ by the definition of bias. We interpret this variance as the inherent {\em information loss} of the communication scheme because this loss would be incurred even if the sender sends this anecdote to herself. The second term captures the disagreement between the receiver's shift $\sigma(\alpha)$ and the sender's preferred shift for the receiver which equals $M_S-\beta(\pi_S)$.\footnote{The sender would like the receiver to choose an action at offset $M_S$ after compensating for the bias $\beta(\pi_S)$ in the communication scheme.}

\subsection{Most Informative Communication Schemes}\label{sub:most_informative}

We can better understand the tension in minimizing both the information loss and the disagreement loss by first characterizing communication schemes that minimize \textit{only} the information loss. This is the most informative communication scheme for a sender who is choosing the action herself based on a single recalled anecdote. For example, a sender who can only keep a single anecdote in memory for her future self solves this problem.

It turns out that most informative communication schemes are always targeting schemes that have to satisfy a simple necessary condition, namely that the offset $r^*$ of the targeting scheme has to equal its bias $\beta(r^*)$ which we call the ``balanced-offset'' condition.
\begin{theorem}\label{thm:most_informative}
    Any translation invariant communication scheme that minimizes the sender's information loss is a targeting scheme that satisfies the balanced-offset condition $r^*=\beta(r^*)$. Moreover, such a targeting scheme exists.
\end{theorem}
\noindent The proofs of this section are relegated to Appendix \ref{app:persuasion}. We can show that the balanced-offset condition implies that the most informative targeting scheme will have the sender choose an offset close to the global maximum of the density distribution. We now develop this intuition for the special case of foresight (where the sender knows the state of the world). Section \ref{sec:commitment} extends it to the no-foresight case.
\begin{lemma}\label{lem:symmetric}
    For symmetric anecdote distribution $F$ with a peak at $0$ the most informative targeting scheme for any number of anecdotes $n$ under foresight is the mean scheme.
\end{lemma}
\noindent Figure \ref{fig:most_informative_normal_0} shows the normal anecdote distribution $F\sim N(0,1)$. Because of symmetry, the closest anecdote is equally likely to lie to the left or right of the offset $r=0$ (indicated by green shading) and therefore the targeting scheme has bias $\beta(0)=0$. Moreover, no offset $r<0$ can implement a most informative targeting scheme because the closest anecdote is more likely to the right of the offset than to the left and therefore $\beta(r)>r$. Analogously, no offset $r>0$ can implement such a communication scheme either.

\begin{figure}
\caption{Examples of most informative offset $r^*$ under foresight}\label{fig:most_informative_normal}
  \centering
  \begin{subcaptiongroup}
    \centering
    \parbox[t]{.46\textwidth}{
    \centering
    \input{Figures/most_informative_r_normal_0}
    \caption{$F\sim N(0,1)$}\label{fig:most_informative_normal_0}}
    \parbox[t]{.46\textwidth}{
    \centering
    \input{Figures/most_informative_r_normal_2}
    \caption{$F\sim N(2,1)$}\label{fig:most_informative_normal_2}}
  \end{subcaptiongroup}
  \begin{center}
    \small Green shading shows symmetric density of closest anecdote to offset $r^*$.
\end{center}
\end{figure}

Lemma \ref{lem:symmetric} is intuitive because it places the most informative offset at the point where the density of anecdote distribution is the highest. Anecdotes are most densely distributed around the peak and the closest anecdote to $r^*=0$ is therefore a better predictor of the state of the world compared to the closest anecdote to any other offset $r \neq 0$. This argument carries over to the case of symmetric single-peaked distributions that are not centered at $0$. For example, Figure \ref{fig:most_informative_normal_2} shows an normal anecdote distribution centered around $2$ which shifts the most informative offset to the new peak $r^*=2$.

The general intuition that the most informative offset is close to the maxima of the density distribution even holds when the density has more than one mode and is not symmetric.
\begin{lemma}\label{lem:most_informative_limit}
Assume the sender has foresight, the anecdote density $f$ is uniformly continuous on its domain and fix any $\epsilon>0$. There exists an $\overline{n}$ such that for any number of anecdotes $n>\overline{n}$ the most informative offset $r^*$ is at most $\epsilon$ away from one of the global maxima of the anecdote density.
\end{lemma}
Figure \ref{fig:most_informative_bimodal_5} shows an example of a bimodal distribution where the anecdote is with equal probability drawn from either the normal distributions $N(-2,1)$ or $N(2,1)$. The density of this bimodal distribution has two global maxima at $-2$ and $2$. When the sender has foresight and $5$ anecdotes to choose from there are three offsets that satisfy the balanced-offset condition: $r^*=\pm 1.73$ and $r^*=0$. The latter is a local minimum but the former offsets implement the two possible most informative targeting schemes. For $n=10$ the two most informative targeting schemes at $r^*=\pm 1.98$ are even closer to the peaks of the density as predicted by Lemma \ref{lem:most_informative_limit} (see Figure \ref{fig:most_informative_bimodal_10}).

\begin{figure}
\caption{Most informative offset $r^*$ under foresight for bimodal normal mixture distribution}\label{fig:most_informative_bimodal}
  \centering
  \begin{subcaptiongroup}
    \centering
    \parbox[t]{.46\textwidth}{
    \centering
    \input{Figures/most_informative_r_bimodal_5}
    \caption{$n=5$}\label{fig:most_informative_bimodal_5}}
    \parbox[t]{.46\textwidth}{
    \centering
    \input{Figures/most_informative_r_bimodal_10}
    \caption{$n=10$}\label{fig:most_informative_bimodal_10}}
  \end{subcaptiongroup}
  \begin{center}
  \small $F \sim B\cdot N(-2,1)+(1-B)\cdot N(2,1)$ where $B$ is Bernoulli with $p=\frac{1}{2}$.
  \end{center}
\end{figure}

\subsection{Persuasion Temptation}

We now return to the original problem and allow the sender to take the disagreement loss into account. Note, that if the receiver chooses his best response $\sigma(\alpha)=M_R-\beta(r^*)$ to the sender's strategy (Theorem \ref{thm:receiver.br}) then the disagreement loss equals $(M_R-M_S)^2$. In the special case where the sender and receiver share the same personal preferences (e.g. $M_R=M_S$) the disagreement loss is zero and therefore a balanced-offset targeting scheme is also an equilibrium of the sender-receiver game as we will show in Section \ref{sec:equilibrium}. This is intuitive because the most informative targeting scheme solves exactly the problem where the sender is constrained to send her future self a single anecdote.

However, the balanced-offset targeting scheme is no longer an equilibrium when sender and receiver are misaligned. In that case, the sender can reduce her disagreement loss by slightly deviating to a targeting scheme with a bias of $\beta(r^*)+\delta$ ($|\delta|$ is small), which moves the receiver's action towards the sender's preferred action. For example, in the motivating example of Figure \ref{fig:intuition} she will choose $\delta<0$. This reduces the disagreement loss to $(|M_R-M_S|-|\delta|)^2$ which is an $O(|\delta|)$ improvement to the sender's utility. The information loss increases but this effect is only $O(\delta^2)$ because we deviated from the targeting scheme that minimized the information loss.\footnote{If we write the information loss as a function $L(\beta)$ of the bias $\beta$ and assume it is twice differentiable then we can use the Taylor expansion around $\beta^*$ and write $L(\beta^*+\delta)= L(\beta^*)+\frac{1}{2}L''(\beta^*+\zeta\delta)\delta^2$ for some $0\leq\zeta\leq 1$ because $L'(\beta^*)=0$ (as $L$ is minimized at $\beta^*$). Therefore, if $L''$ is bounded the change in the information loss is $O(\delta^2)$.} Consequently, a small $\delta$-shift in the sender's communication bias is a strictly profitable deviation. Of course, the receiver will respond by adjusting the shift in her action rule by $-\delta$. In return, the sender will adjust her bias by another $\delta$ leading to \textit{partial} unraveling. There is no full unraveling towards always reporting the minimum (or maximum) anecdote because at some point the increase in the information loss term is no longer second-order. 

Note, that in any equilibrium, the disagreement loss is always $(M_R-M_S)^2$ due to Theorem \ref{thm:receiver.br} despite the sender's best efforts to reduce it. The rational receiver can always undo any bias that is added by the sender. However, the sender's \textit{persuasion temptation} increases her information loss. Communicating with a non-aligned receiver versus an aligned receiver therefore has two costs to the sender. First of all, there is always the disagreement loss $(M_R-M_S)^2$. Second, there is also an \textit{additional} information loss because the temptation to persuade the receiver induces the sender to choose a suboptimal communication scheme. This is an example of \textit{informational homophily} from the sender's perspective: given a choice between a non-aligned and an aligned receiver the sender always prefers to communicate with the aligned receiver.

The next result shows that this type of homophily is mutual in our model.
\begin{proposition}\label{prop:receiver-loss}
Let $\pi$ be a translation invariant communication scheme and let $\alpha_\pi$ be the best response of the receiver. Then the utility of the receiver is equal to the negative variance of the anecdote $\pi(\vec{x},y)$, 
\begin{equation}
  -\E[(\pi(\vec{x},y) - (\theta  +\beta(\pi)))^2].  
\end{equation}
\end{proposition}
\noindent
Hence, the receiver's loss is exactly equal to the sender's (expected) information loss. The receiver suffers no disagreement loss because she can in equilibrium undo any bias imposed by the sender. Therefore, the receiver always prefers to listen to an aligned sender (without persuasion temptation) rather than a non-aligned sender. Moreover, sender and receiver utility differ only by the constant disagreement loss $(M_R-M_S)^2$.

%% file: Figures/most_informative_r_normal_0.tex
\begin{tikzpicture}

\fill [shading = axis,left color=white, right color=green!30!white, ,color=white, domain=-1:0, variable=\x]
  (-1, 0)
  -- plot ({\x}, {2.5*exp(-(\x-0)*(\x-0)/2)})
  -- (0, 0)
  -- cycle;
\fill [shading = axis,right color=white, left color=green!30!white, ,color=white, domain=0:1, variable=\x]
  (0, 0)
  -- plot ({\x}, {2.5*exp(-(\x-0)*(\x-0)/2)})
  -- (1, 0)
  -- cycle;

\fill [shading = axis,left color=white, right color=black!30!white, ,color=white, domain=-2:-1.7, variable=\x]
  (-2, 0)
  -- plot ({\x}, {2.5*exp(-(\x-0)*(\x-0)/2)})
  -- (-1.7, 0)
  -- cycle;
\fill [shading = axis,right color=white, left color=black!30!white, ,color=white, domain=-1.7:-1.2, variable=\x]
  (-1.7, 0)
  -- plot ({\x}, {2.5*exp(-(\x-0)*(\x-0)/2)})
  -- (-1.2, 0)
  -- cycle;
  \draw[dashed,black]    (-1.7,0)   -- (-1.7,3);
\draw[-stealth,ultra thick,black!30]    (-1.7,-1)   -- (-1.7,-0.2) node[below=25] {$r<0$} node[below=35] {$\beta(r)>r$};

\fill [shading = axis,right color=white, left color=black!30!white, ,color=white, domain=1.7:2, variable=\x]
  (1.7, 0)
  -- plot ({\x}, {2.5*exp(-(\x-0)*(\x-0)/2)})
  -- (2, 0)
  -- cycle;
\fill [shading = axis,left color=white, right color=black!30!white, ,color=white, domain=1.2:1.7, variable=\x]
  (1.2, 0)
  -- plot ({\x}, {2.5*exp(-(\x-0)*(\x-0)/2)})
  -- (1.7, 0)
  -- cycle;
  \draw[dashed,black]    (1.7,0)   -- (1.7,3);
\draw[-stealth,ultra thick,black!30]    (1.7,-1)   -- (1.7,-0.2) node[below=25] {$r>0$} node[below=35] {$\beta(r)<r$};

  \draw[domain=-2.5:2.5, smooth, variable=\x, black] plot ({\x}, {2.5*exp(-(\x-0)*(\x-0)/2)});

  \draw[->] (-3, 0) -- (3, 0);
  \draw[->] (0,0) -- (0, 3);

\draw[-stealth,ultra thick,black!30!green]    (0,-1)   -- (0,-0.2) node[below=25] {$r^*=0$} node[below=35] {$\beta(r^*)=0$};

\end{tikzpicture}

%% file: Figures/most_informative_r_normal_2.tex
\begin{tikzpicture}

\fill [shading = axis,left color=white, right color=green!30!white, ,color=white, domain=-1:0, variable=\x]
  (-1, 0)
  -- plot ({\x}, {2.5*exp(-(\x-0)*(\x-0)/2)})
  -- (0, 0)
  -- cycle;
\fill [shading = axis,right color=white, left color=green!30!white, ,color=white, domain=0:1, variable=\x]
  (0, 0)
  -- plot ({\x}, {2.5*exp(-(\x-0)*(\x-0)/2)})
  -- (1, 0)
  -- cycle;

  \draw[domain=-2.5:2.5, smooth, variable=\x, black] plot ({\x}, {2.5*exp(-(\x-0)*(\x-0)/2)});

  \draw[->] (-3, 0) -- (3, 0);
  \draw[->] (0-2,0) -- (-2, 3);

\draw[-stealth,ultra thick,black!30!green]    (0,-1)   -- (0,-0.2) node[below=25] {$r^*=2$} node[below=35] {$\beta(r^*)=2$};

\end{tikzpicture}

%% file: Figures/most_informative_r_bimodal_5.tex
\begin{tikzpicture}

  \draw[domain=-3:3, smooth, variable=\x, black] plot ({\x}, {0.5*5*exp(-(\x-2)*(\x-2)/2)+0.5*5*exp(-(\x+2)*(\x+2)/2)});

  \draw[->] (-3, 0) -- (3, 0);
  \draw[->] (0,0) -- (0, 3);

\draw[-stealth,ultra thick,black!30!green]    (-1.73,-1)   -- (-1.73,-0.2) node[below=25] {$r^*=-1.73$};

\draw[-stealth,ultra thick,black!30!green]    (1.73,-1)   -- (1.73,-0.2) node[below=25] {$r^*=1.73$} ;

\draw[-stealth,ultra thick,black!30]    (0,-1)   -- (0,-0.2) node[below=25] {$r^*=0$};

\draw[dashed,black!30!green]    (-1.73,0)   -- (-1.73,3);
\draw[dashed,black!30!green]    (1.73,0)   -- (1.73,3);

\end{tikzpicture}

%% file: Figures/most_informative_r_bimodal_10.tex
\begin{tikzpicture}

  \draw[domain=-3:3, smooth, variable=\x, black] plot ({\x}, {0.5*5*exp(-(\x-2)*(\x-2)/2)+0.5*5*exp(-(\x+2)*(\x+2)/2)});

  \draw[->] (-3, 0) -- (3, 0);
  \draw[->] (0,0) -- (0, 3);

\draw[-stealth,ultra thick,black!30!green]    (-1.98,-1)   -- (-1.98,-0.2) node[below=25] {$r^*=-1.98$};

\draw[-stealth,ultra thick,black!30!green]    (1.98,-1)   -- (1.98,-0.2) node[below=25] {$r^*=1.98$} ;

\draw[-stealth,ultra thick,black!30]    (0,-1)   -- (0,-0.2) node[below=25] {$r^*=0$};

\draw[dashed,black!30!green]    (-1.98,0)   -- (-1.98,3);
\draw[dashed,black!30!green]    (1.98,0)   -- (1.98,3);

\end{tikzpicture}

%% file: Sections/equilibrium.tex
\section{Communication Equilibria, Polarized Communication and Informational Homophily}\label{sec:equilibrium}

In this section we show that translation invariant equilibria in our sender-receiver game always exist and are characterized by a simple condition that generalizes the ``balanced-offset'' condition of most informative targeting schemes from Section \ref{sec:persuasion}. We then show that these equilibria generically give rise to \textit{polarized communication} in the sense that whenever the sender is more left-leaning (right-leaning) than the receiver she will select anecdotes with a left (right) offset. Polarization in return decreases utility for both the sender and the receiver and therefore implies \textit{informational homophily}.

\subsection{Equilibrium Characterization and Existence}

We start by defining a \textit{translation invariant equilibrium} which is simply a PBE where both the sender and the receiver use translation invariant strategies. 
\begin{definition}
A perfect Bayesian equilibrium $(\pi, \alpha)$ is translation invariant if $\pi$ is a translation invariant communication scheme and $\alpha$ is a translation.
\end{definition}
Theorem \ref{thm:sender.br} tells us that in such translation invariant PBE the sender will use a targeting scheme with offset $r$ that satisfies
\begin{equation}
    r=M_S-\sigma(\alpha).
\end{equation}
This is also shown graphically in our motivating example from Figure \ref{fig:intuition}. Moreover, theorem \ref{thm:receiver.br} pins down the shift $\sigma(\alpha)$:
\begin{equation}
    \sigma(\alpha)=M_R-\beta(r).
\end{equation}
Here, $\beta(r)$ denotes the bias $\beta(\pi)$ of targeting scheme $\pi$ with offset $r$.

Combining these two conditions gives us a necessary condition for a translation invariant PBE:
\begin{eqnarray}
r&=&M_S - \underbrace{(M_r-\beta(r))}_{\sigma(\alpha)} \nonumber \\
\beta(r) -r &=& \underbrace{M_R-M_S}_{\Delta}.
\end{eqnarray}
Note that this condition reduces to the balanced-offset condition when the preferences of sender and receiver are aligned ($\Delta=0$).

The above condition is not only necessary but also a sufficient condition for PBE.
\begin{theorem}
\label{thm:bne}
A pair $(\pi,\alpha)$ 
is a translation invariant equilibrium if and only if there exists some value $r \in \reals$ such that $\pi$ is the targeting scheme with offset $r$, $\alpha$ is a translation with shift $M_R-\beta(r)$, and the bias satisfies
\begin{equation*}
    \beta(r) - r = M_R-M_S = \Delta.
\end{equation*}
\end{theorem}
All the proofs in this section can be found in Appendix \ref{app:equilibrium}.

\begin{figure}
\caption{Equilibrium offset for normal distribution $N(0,1)$ and $n=10$ anecdotes}\label{fig:normal_equilibrium_n10}
  \centering
  \begin{subcaptiongroup}
    \centering
    \parbox[t]{.46\textwidth}{
    \centering
    \input{PythonFigures/beta_r_normal_10}
    \caption{Equilibrium offset for $\Delta=M_R-M_S=2$}\label{fig:normal_equilibrium_n10_Delta2}}
    \parbox[t]{.46\textwidth}{
    \centering
    \input{PythonFigures/Delta_rstar_normal_10}
    \caption{Equilibrium offset $r(\Delta)$}\label{fig:normal_equilibrium_n10_anyDelta}}
  \end{subcaptiongroup}
\end{figure}

The distance between the bias $\beta(r)$ of the targeting scheme and its offset $r$ is therefore the key relationship to understand the emergence of communication equilibria in our game. Figure \ref{fig:normal_equilibrium_n10_Delta2} shows how this distance varies with $r$ in the simple motivating example of Figure \ref{fig:intuition} where the sender has foresight and access to $n=10$ normally distributed anecdotes around the state of the world $\theta$ ($F \sim N(0,1)$). For negative offset $r=0$ we observe that $\beta(r)>r$ because the closest anecdote to $r$ is more likely to be on the right of $r$ (see Figure \ref{fig:intuition}).  Moreover, the distance $\beta(r)-r$ increases as we decrease $r$ because the targeting scheme is approaching the minimum scheme and the bias $\beta(r)$ converges to the expected offset of the minimum anecdote (in this example, the minimum of $n=10$ anecdotes). This ensures that for any $\Delta>0$ there will be some $r^*<0$ that satisfies $\beta(r^*)-r^*=\Delta$. For example, for $M_R-M_S=2$ we obtain $r^*=-3.56$.

We can show that this intuition generalizes for a broad set of anecdote distributions.
\begin{theorem}
\label{thm:bne.exist}
For any $n$, $M_S$ and $M_R$, if 
$\E_{\epsilon \sim F}|\epsilon|$ is bounded then a translation invariant PBE exists.
\end{theorem}
\noindent We restrict attention from now on to anecdote distributions that satisfies this property.

We highlight the main technical ideas of the formal proof provided in Appendix \ref{app:equilibrium}. Given a communication scheme $\pi_r$ with offset $r$, let $z = \pi_r(\vec x,y) - (\theta_S(\vec x, y) + r)$ denote the distance between the target $\theta_S(\vec{x},y) + r$ and the closest anecdote (out of $n$ total anecdotes), where $z$ is positive if the closest anecdote is larger and negative if the closest anecdote is smaller.  Let $H(r)$ be the expected value of $z$, given offset $r$. By definition, we have that $\beta(r) = r + H(r)$. Theorem~\ref{thm:bne} now implies that to show that a PBE exists, it suffices to show that there exists a value of $r$ such that
$H(r) = M_R-M_S$, i.e., $H(r)$ is an onto function in its codomain $(-\infty, +\infty)$.

We show this in two steps. In the first step, we establish that $H(r)$ tends to $+\infty$ and $-\infty$, respectively, as $r\rightarrow -\infty$ and $r\rightarrow +\infty$. 
At a high level, this is due to that fact that boundedness of $\E[|x-\theta|]$ also implies that $\E[\max_i x_i - \theta_S(\vec x,y)]$ and $\E[\min_i x_i - \theta_S(\vec x,y)]$ are also bounded (for a fixed $n$). Therefore, $\beta(r)$ remains bounded even as $|r|\rightarrow \infty$. This proves that $H(r)$ tends to $+\infty$ and $-\infty$, respectively, as $r\rightarrow -\infty$ and $r\rightarrow +\infty$. 

Had $H(r)$ been a continuous function, the first step would suffice to prove that $H(r)$ is an onto function taking all values in $(-\infty, +\infty)$. 
The second step handles discontinuous $H(r)$ by establishing that at any point that $H(r)$ is discontinuous, the left limit of $H(r)$ is smaller than the right limit.  To see why, suppose $H$ is not continuous at $r_0$.  For each possible realization of $(\theta,\vec{x},y)$, either $z$ is continuous at $r_0$ or it is not.  If not, this means that $\theta_S(\vec{x},y) + r_0$ is precisely halfway between two anecdotes in $\vec{x}$, say with absolute distance $d > 0$ to each, in which case the limit of $z$ from below is $-d$ (distance to the anecdote to the left) and the limit of $z$ from above is $d$ (distance to the anecdote to the right).  Integrating over all realizations, we conclude that the one-sided limits of $H$ exist and $\lim_{r \to r_0^-}H(r) < \lim_{r \to r_0^+}H(r)$.

Together these two steps show that $H(r)$ can take any value in the range $(-\infty, +\infty)$ and therefore a translation invariant PBE always exists

We can use the same logic to characterize the translation invariant equilibrium of the broadcasting extension of our model that we introduced in Section \ref{subsec:broadcasting}.
\begin{lemma}
A translation invariant equilibrium in the broadcast model satisfies $\beta(r)-r=\overline{M}_R-M_S$ where $\overline{M}_r=\int M_F dG$ is the mean personal preference across receiver types. Moreover, for any $n$ such an equilibrium exists if $\E_{\epsilon \sim F}|\epsilon|$ is bounded.
\end{lemma}\label{lemma:broadcast_equilibrium}
\noindent Consequently, the sender's incentive to persuade is dictated by the misalignment $\Delta=\overline{M}_R-M_S$ between this representative receiver and the sender.

\subsection{Polarized Communication with Non-aligned Preferences}

We now show that communication becomes more \textit{polarized} the more misaligned the preferences of sender and receiver are. Before formally defining polarization, we examine the case where sender and receiver are perfectly aligned ($\Delta=0$).
\begin{corollary}
    When sender and receiver preferences are aligned ($\Delta=0$) the balanced-offset condition holds and the equilibrium targeting scheme minimizes both the sender's and receiver's information loss.
\end{corollary}
\noindent  This follows immediately from Theorems \ref{thm:bne} and \ref{thm:most_informative}.
In our example of a normal anecdote distribution $N(0,1)$ this implies that the sender chooses the mean scheme ($r^*=0$) where she sends an anecdote closest to her posterior mean.

In principle there can be many balanced-offset equilibria: we will focus on the minimal and maximal balanced-offset equilibria, $\underline{r}$ and $\overline{r}$:\footnote{We know that $\underline{r}$ (resp., $\overline{r}$) exists because $H(r)\rightarrow \infty$ as $r\rightarrow -\infty$ (resp., $H(r)\rightarrow -\infty$ as $r \rightarrow \infty$), {and at any point of discontinuity $r_0$ we have $\lim_{r \to r_0^-}H(r) \leq H(r_0) \leq \lim_{r \to r_0^+}H(r)$.  To see why this implies existence of $\underline{r}$, define $r_0 = \inf\{r : H(r) \geq 0\}$; then it suffices to show that $H(r_0) = 0$.  If $H$ is continuous at $r_0$ then this is immediate, but if $H$ is discontinuous at $r_0$ then $\lim_{r \to r_0^+}H(r) > \lim_{r \to r_0^-}H(r) \geq 0$ in violation of the definition of $r_0$.}} 
\begin{eqnarray}
    \underline{r} &=& \min \left\{r|H(r)=0\right\} \nonumber \\
    \overline{r} &=& \max \left\{r|H(r)=0\right\}
\end{eqnarray}
We then define the \textit{least polarized} communication equilibrium $r(\Delta)$:\footnote{{Likewise, the least polarized equilibrium exists because $H(r)\rightarrow \infty$ as $r\rightarrow -\infty$ and $H(r)\rightarrow -\infty$ as $r\rightarrow \infty$, and and at any point of discontinuity $r_0$ we have $\lim_{r \to r_0^-}H(r) \leq H(r_0) \leq \lim_{r \to r_0^+}H(r)$}.}
\begin{equation}
    r(\Delta) = \begin{cases}\begin{array}{ll}
    \max \left\{r|H(r)=\Delta, r<\underline{r} \right\} & \quad \mbox{for $\Delta>0$} \\
    \min \left\{r|H(r)=\Delta, r>\overline{r} \right\} & \quad \mbox{for $\Delta<0$} 
    \end{array}\end{cases}
    \end{equation}
For example, if the sender is more left-leaning than the receiver ($\Delta>0$) then $r(\Delta)$ is the equilibrium with offset $r(\Delta)$ that is closest to \textit{the left of} the minimal balanced-offset equilibrium $\underline{r}$. For sufficiently large $\Delta$ all possible equilibria must be to the left of $\underline{r}$. However, for small $\Delta$ it is possible that there are equilibria in the interval $[\underline{r},\overline{r}]$. For symmetric single-peaked distributions such as the normal distribution we know that $\underline{r}=\overline{r}=0$ and therefore the least polarized equilibrium for $\Delta>0$ are the ones with minimal absolute offset. 

Figure \ref{fig:normal_equilibrium_n10_anyDelta} shows the unique equilibrium offset for the normal distribution as a function of $\Delta$. The graph suggests that communication becomes increasingly polarized the more misaligned the preferences of sender and receiver are. The next result shows this observation holds generally.
\begin{proposition}\label{prop:r_increases_with _Delta} Assume some fixed anecdote distribution $F$ and a fixed number of anecdotes $n$. The least polarized equilibrium
$r(\Delta)$ decreases with $\Delta$. Therefore, the sender's communication becomes more left-polarized (right-polarized) the more her preferences lie to the left of (right of) the receiver's preferences.
\end{proposition}
\noindent Figure \ref{fig:least_polarized_intuition} illustrates the intuition behind the proof for $\Delta>0$. It shows the graph of $H(r)=\beta(r)-r$ and the minimal balanced-offset equilibrium at $\underline{r}$ for some non-single-peaked distribution. When the difference, $\Delta$, in personal preferences equals $\Delta_1$, there are three possible equilibria with the least polarized one at $r(\Delta_1)$. For $\Delta=\Delta_2$ there are two equilibria with the least polarized one at $r(\Delta_2)<r(\Delta_1)$. For any $\Delta>\Delta_2$ (such as $\Delta_3$) there is only a single equilibrium with $r(\Delta)<r(\Delta_2)$. Therefore, the function $r(\Delta)$ has a discontinuity at $\Delta_2$ but is decreasing for any $\Delta>0$.

\begin{figure}
\caption{Intuition for Proposition \ref{prop:r_increases_with _Delta} ($r(\Delta)$ decreases)}\label{fig:least_polarized_intuition}
  \centering
    \input{Figures/least_polarized_intuition}
\end{figure}

As the misalignment between sender and receiver becomes extreme, the sender reports essentially only the minimum or maximum anecdote.
\begin{corollary}\label{corollary:minmax}
Assume some fixed anecdote distribution $F$ and a fixed number of anecdotes $n$. As $\Delta \rightarrow \infty$ ($\Delta \rightarrow -\infty$) the sender's communication scheme at any translation invariant equilibrium converges to the minimum (maximum) scheme.
\end{corollary}

\subsection{Informational Homophily}

We know from Theorem \ref{thm:most_informative} that all targeting schemes that minimize the sender's information loss (as well as the receiver's utility) have to satisfy the balanced-offset condition. Therefore, any equilibrium with $\Delta \neq 0$ has to have strictly lower utility for the sender and weakly lower utility for the receiver compared to the most informative communication scheme.\footnote{The utility is strictly lower for the sender because the information loss is weakly lower and the disagreement loss $\Delta^2$ is strictly negative.}

Intuitively, one might expect that more polarized communication in response to greater misalignment in sender and receiver preferences ($|\Delta| \uparrow$) decreases utility. We show that this holds under foresight.
\begin{proposition}\label{prop:utility_decreases_with _Delta} Assume some fixed anecdote distribution $F$, a fixed number of anecdotes $n$ and foresight. Both the sender and receiver utility of the least polarized equilibrium decreases with $|\Delta|$.
\end{proposition}
\noindent The proof establishes that the variance of the sent anecdote in the least polarized equilibrium increases with $|\Delta|$. The result then follows immediately from Propositions \ref{prop:sender_summary} and \ref{prop:receiver-loss}.

Proposition \ref{prop:utility_decreases_with _Delta} implies that anecdotal communication gives rise to \textit{informational homophily}. Both the sender and the receiver prefer to communicate with more like-minded partners because communication is less polarized and therefore more precise. Greater polarization does not ``fool'' a rational receiver who corrects for the larger bias in the sender's anecdote, but it comes at the cost of precision which hurts both the sender and receiver.

%% file: PythonFigures/beta_r_normal_10.tex
\begin{tikzpicture}[scale = 0.85]
\begin{axis}[
    domain=-5:5,
    xmin=-5, xmax=5,
    axis lines=center,
    ytick distance=1,
    xlabel=$r$    
    ]
    \addplot[smooth,color=black!30!green] table[x index=1,y index=2,col sep=comma] {PythonFigures/r_beta_normal_10_3.csv} node[below,pos=0.8] {$\beta(r)$};
    \addplot[dashed,color=purple] coordinates {(-5, -3) (3,5)} node[right,pos=0.9] {$r+\Delta$};
    \addplot[dashed,color=black] coordinates {(-5, -5) (4,4)} node[below,pos=0.9] {$r$};

    \draw[dotted, color=black]    (-3.56,0)   -- (-3.56,-5);
    \draw[-stealth,ultra thick,black!30!green]    (-3.56,1)   -- (-3.56,0.2) node[above=15] {$r^*=-3.56$};

\end{axis}
\end{tikzpicture}

%% file: PythonFigures/Delta_rstar_normal_10.tex
\begin{tikzpicture} [scale = 0.85]
\begin{axis}[
    domain=-3:3,
    xmin=-3, xmax=3,
    axis lines=center,
    ytick distance=1,
    xlabel=$\Delta$,
    ylabel=$r^*$
    ]
    \addplot[smooth,color=purple] table[x index=1,y index=3,col sep=comma] {PythonFigures/delta_on_r_10.csv} ;

    \draw[dotted, color=black!30!green]    (2,0)   -- (2,-3.56);
    \draw[dotted, color=black!30!green]    (2,-3.56) -- (0,-3.56);
    \draw[-stealth,ultra thick,black!30!green]    (-1,-3.56)   -- (-0.2,-3.56) node[left=25] {$r(2)=-3.56$};

\end{axis}
\end{tikzpicture}

%% file: Figures/least_polarized_intuition.tex
\begin{tikzpicture}[scale = 0.85]

\draw[->] (2, 0) -- (-6.5,0) ;
\draw[->] (2,-1) -- (2, 4.5);

\draw[-stealth,ultra thick,purple]    (0,-0.5)   -- (0,-0.1) node[below=10] {$\underline{r}$};

\draw[domain=0:-6.5, smooth, variable=\x, purple] plot ({\x}, {(-\x^3+9.5*\x^2-24*\x)/6}) node[right=10] {$H(r)$};

\draw[dotted, color=black]    (2,2) -- (-6.5,2);
\draw[-stealth,ultra thick,black]    (2.5,2)   -- (2.2,2) node[right=10] {$\Delta_1$};
\draw[dotted, color=black]    (-0.68,0) -- (-0.68,2);
\draw[-stealth,ultra thick,black]    (-0.68,-0.5)   -- (-0.68,-0.2) node[below=10] {$r(\Delta_1)$};

\draw[dashed, color=black]    (2,3.044) -- (-6.5,3.044);
\draw[-stealth,ultra thick,black]    (2.5,3.044)   -- (2.2,3.044) node[right=10] {$\Delta_2$};
\draw[dashed, color=black]    (-1.743,0) -- (-1.743,3.044);
\draw[-stealth,ultra thick,black]    (-1.743,-0.5)   -- (-1.743,-0.2) node[below=10] {$r(\Delta_2)$};

\draw[dash dot, color=black]    (2,4) -- (-6.5,4);
\draw[-stealth,ultra thick,black]    (2.5,4)   -- (2.2,4) node[right=10] {$\Delta_3$};
\draw[dash dot, color=black]    (-6.29,0) -- (-6.29,4);
\draw[-stealth,ultra thick,black]    (-6.29,-0.5)   -- (-6.29,-0.2) node[below=10] {$r(\Delta_3)$};

\end{tikzpicture}

%% file: Sections/experts.tex
\section{Experts and the Curse of Informedness}\label{sec:experts}

In the previous section we characterized communication and welfare when varying the misalignment $\Delta$ between sender and receiver preferences but keeping the anecdote distribution and the number of anecdotes constant. We showed that greater misalignment induces greater polarization which in return reduces welfare. Therefore, both sender and receiver will tend to seek out like-minded communication partners (keeping everything else equal).

Another natural question is whether communication becomes more or less efficient when the sender is better informed (more of an ``expert''). Formally, we are fixing both the anecdote distribution and the preference misalignment $\Delta$ in this section and vary the number of anecdotes $n$. Somewhat surprisingly, increasing $n$ can make communication \textit{less} efficient. We first show that experts use more extreme anecdotes to communicate \textit{unless} the preferences are completely aligned ($\Delta=0$). Intuitively, the information loss for any fixed offset $r$ decreases with $n$ and therefore the sender has a greater incentive to choose an even more biased anecdote. However, the impact on welfare is ambiguous because extreme anecdotes can also become more informative about the state of the world as $n$ increases. We show that for an (appropriately defined) ``light-tailed'' distributions (such as the normal distribution) experts provide more precise anecdotes despite the polarization. On the other hand, for ``heavy-tailed'' distributions the precision of anecdotes sent by experts decreases with $n$. Intuitively, such experts have access to too many ``outlier anecdotes'' which makes it difficult for the rational receiver to learn about the state of the world. We call this phenomenon the \textit{curse of informedness}.

\subsection{Polarized Experts}

Throughout this section we assume that the sender has foresight. We denote the bias of a targeting scheme with offset $r$ and $n$ anecdotes with $\beta_n(r)$.

Figure \ref{fig:normal_equilibrium_n3_n10_Delta2} illustrates the construction of the (unique) targeting equilibrium when $\Delta=2$ and the anecdote distribution is normally distributed ($F \sim N(0,1)$) for $n=3$ and $n=10$. We observe in this example that increasing the number of anecdotes makes the sender's communication more polarized as the sender's target moves from $r_3(2)=-2.83$ to $r_{10}(2)=-3.56$. This observation holds generally. For the next result, we assume that the anedote distribution $F$ is continuous and strictly positive over its domain $[r_{\min},r_{\max}]$ (where we allow $r_{\min}=-\infty$ and/or $r_{\max}=\infty$).

\begin{figure}
\caption{Equilibrium offset $r_n(\Delta)$ for normal anecdote distribution}\label{fig:polarizing_experts}
  \centering
  \begin{subcaptiongroup}
    \centering
    \parbox[t]{.46\textwidth}{
    \centering
    \input{PythonFigures/r_beta_normal_10_3}
    \caption{$r_{10}(2)<r_3(2)$}\label{fig:normal_equilibrium_n3_n10_Delta2}}
    \parbox[t]{.46\textwidth}{
    \centering
    \input{PythonFigures/Delta_rstar_normal_10_3}
    \caption{$r_{10}(0.51)=r_3(1.14)=-2$}\label{fig:normal_equilibrium_n10_3_anyDelta}}    
  \end{subcaptiongroup}
\end{figure}

\begin{proposition} \label{prop:polarization_n} Assume that the sender has foresight and that the anecdote distribution $F$ is fixed, continuous and strictly positive over its domain $[r_{\min},r_{\max}]$ (where $r_{\min}$ can be $-\infty$ and/or $r_{\max}$ can be $\infty$). Suppose also that the personal preferences of receiver and sender differ by a fixed $\Delta \neq 0$.  Then the communication at equilibrium becomes more polarized as $n$ increases, in the following sense: there exist sequences $R_{\min}(n)$ and $R_{\max}(n)$ with $R_{\min}(n) \to r_{\min}$ and $R_{\max}(n) \to r_{\max}$ such that for any equilibrium offset $r^*_n$ (i.e., $\beta_n(r^*_n) - r^*_n = \Delta$) we have that either $r^*_n \le R_{\min}(n)$ or $r^*_n \ge R_{\max}(n)$ for all $n$.
\end{proposition}
\noindent All the proofs for this section are in Appendix \ref{app:experts}. The intuition for this proposition is shown in Figure \ref{fig:normal_equilibrium_n3_n10_Delta2}. The bias $\beta_n(r)$ of the targeting scheme gets closer to the diagonal as $n$ increases: the more anecdotes the sender has access to, the easier it is to find an anecdote close to the offset $r$. Therefore, the information loss for the sender from communicating in a more polarized fashion becomes smaller as $n$ increases and hence the sender's persuasion temptation will induce her to choose a more extreme targeting scheme. 

Figure \ref{fig:normal_equilibrium_n10_3_anyDelta} demonstrates that even a small misalignment in personal preferences can produce a high level of polarization. When the sender has access to $n=3$ anecdotes the differences in personal preferences have to be at least $\Delta\geq 1.14$ for the sender to target an anecdote two standard deviations away from the mean anecdote. A more informed sender with access to $n=10$ anecdotes will exhibit that degree of polarization already for  $\Delta \geq 0.51$.  Proposition \ref{prop:polarization_n} therefore implies that the sender's informedness and polarization are complements. 

\subsection{The Curse of Informedness} \label{subsec:informedness}

Does polarized communication by experts decrease the sender's utility (and by Proposition \ref{prop:receiver-loss} the receiver's utility)? In Section \ref{sec:equilibrium}, we observed greater polarization and lower utility when fixing the number of anecdotes $n$ and increasing the sender/receiver alignment ($|\Delta| \uparrow$)  (see Proposition \ref{prop:utility_decreases_with _Delta}).

It is not obvious that a similar result should hold when we fix the sender/receiver alignment $\Delta$ and let the sender become more informed ($n \uparrow$). While communication becomes more polarized (which decreases utility for fixed $n$) there are also more anecdotes to choose from, so even the more extreme anecdotes may be more informative to the receiver. The net impact on both the sender's and receiver's utility therefore seems ambiguous and depends on how informative extreme anecdotes are about the state of the world.

We next show that the net utility impact depends on the tails of the anecdote distribution. We will show that if $F$ is sufficiently ``heavy-tailed'' then polarization will make the expert's signal increasingly uninformative so that the receiver would prefer to talk to \textit{any less-informed sender} (even highly misaligned ones). We call this the ``curse of informedness''. On the other hand, if $F$ is ``light-tailed'' then the expert's anecdotes become increasingly informative as $n$ increases despite the polarization. 

\begin{definition}
We say that the anecdote density has strong heavy tails if its hazard rate satisfies
\begin{equation}\label{eqn:local_fatness}
    \frac{f(x)}{1-F(x)} = u(x) x^{\alpha-1}
\end{equation}
for a positive function $u(x)$ which is bounded from above by some constant $\overline{U}$ and satisfies $x^{\beta} \frac{u'(x)}{u(x)}\rightarrow 0$ for some $\beta>0$ and $0\leq\alpha<1$ and $|x|>\underline{x}$. The anecdote distribution has a strong light tails if 
\begin{equation}
    \frac{f(x)}{1-F(x)} = u(x) x^{\alpha-1}
\end{equation}
for a positive function $u(x)$ which is bounded from below by some constant $\underline{U}$ with $x^{\beta}\frac{u'(x)}{u(x)}\rightarrow 0$ for some $\beta>0$ and $\alpha>1$, $|x|>\underline{x}$. 
\end{definition}
\noindent Distributions with strong heavy (light) tails have tails that decline at a less than (more than) exponential rate everywhere. For example, Table \ref{tab:distribution_by_tail} defines the strong heavy-tailed Pareto distribution ($\alpha=0$), the  strong light-tailed Gaussian distribution ($\alpha=2$) and the Laplace distribution which is a knife-edge case with its exponential tails ($\alpha=1$).

\begin{table}
    \caption{Examples of strong heavy-tailed, strong light-tailed and knife-edge distributions}
    \label{tab:distribution_by_tail}
    \centering
    \bgroup
    \def\arraystretch{2}
    \begin{tabular}{|l|l|l|}
    \hline
     Distribution & probability density function $f(\epsilon)$ &  Type \\
     \hline
     Pareto & $f(\epsilon) = \frac{1}{(|\epsilon|+1)^3}$ & strong heavy tails ($\alpha=0$)\\
     \hline
     Gaussian & $f(\epsilon) = \frac{1}{2\sqrt{\pi}} \exp(-\frac{1}{2}\epsilon^2)$ & strong light tails ($\alpha=2$)\\
    \hline
     Laplace & $f(\epsilon) = \frac{1}{2}\exp(-|\epsilon|)$ & knife-edge ($\alpha=1$)\\
     \hline
    \end{tabular}
    \egroup
\end{table}

We can now state the main result.
\begin{theorem}\label{thm:heavy_light_tails}
Assume that the sender has foresight, the anecdote distribution $F$ is fixed, {continuous, and strictly positive over its domain,} and the personal preferences of receiver and sender differ by a fixed $\Delta \neq 0$. If $F$ is strong-heavy-tailed then the variance of the anecdotes {sent at any equilibrium} converges to $\infty$ as $n \to \infty$.  If instead the anecdote distribution is strong-light-tailed then the variance of the anecdote {sent at any equilibrium} converges to $0$ as $n \to \infty$.
\end{theorem}
\noindent Figure \ref{fig:tails_variance_n} plots the variance of the sent anecdote for $n\in [10,1000]$ for the Pareto, Gaussian and Laplace distributions from Table \ref{tab:distribution_by_tail}. Consistent with Theorem \ref{thm:heavy_light_tails}, the variance stays constant around $0.9$ for the knife-edge Laplace distributions, converges to $0$ for the light-tailed Gaussian and diverges for the heavy-tailed Pareto distribution.

\begin{figure}[ht]
\caption{Variance of sent anecdote for Pareto, Gaussian and Laplace distributions for $\Delta=1$}\label{fig:tails_variance_n}
  \centering
\input{PythonFigures/tails_variance_n}
\end{figure}

An immediate consequence is that for heavy-tailed distributions greater expertise lowers both the sender's and receiver's utility unless both of them are perfectly aligned (this follows from Proposition \ref{prop:receiver-loss}). Therefore, a receiver might prefer to listen to more-aligned but less informed sender rather than a less-aligned expert because the expert finds it tempting to report easily accessible extreme anecdotes. With light-tailed distributions, such extreme anecdotes are still highly informative about the state of the world and therefore the receiver will generally prefer to listen to experts even if their preferences are not aligned with his own.

Importantly, receivers in our model do not distrust experts when anecdotes are heavy-tailed because they think experts tell them lies. Rather they do not trust experts with different preferences because they correctly anticipate that such experts will tell them unrepresentative data points which are too noisy to infer the state of the world.

The curse of informedness also affects the incentives of senders to \textit{acquire} information in the first place. When the anecdote distribution is heavy-tailed and senders expect to be matched with unaligned receivers they have an incentive to \textit{remain less informed}. Being too well informed makes them less credible sources of information for receivers (unless the receiver is perfectly aligned).

To connect these results to the real-world, we can look at the dramatic decline in trust of scientists between 2019 and 2024. The COVID-19 pandemic gave a lot of visibility to experts who advised the public on how to reduce infections, deal with school closings and how to manage the economy. Academic experts in the US tend to be more left-leaning than the rest of the population.\footnote{A study of voter registration of 7,243 faculty at 40 leading U.S. universities in the fields of Economics, History, Journalism/Communications, Law, and Psychology found that 3,623 were registered Democrats and 314 Republicans with an overall ratio of 11.5:1 while the ratio in the overall population is close to $1$ \citep{econjwatch2016}.} Therefore, the personal preference $M_S$ of the majority of experts will lie to the left of the average voter's personal preferences ($\Delta=\overline{M}_R-M_S>0$). The high degree of uncertainty during the pandemic might also serves as a good example of a heavy-tailed anecdote distribution. Hence, the ``curse of informedness'' (and Lemma \ref{lemma:broadcast_equilibrium}) implies that the variance in experts' messages was potentially large during the pandemic. Consistent with this story, 2023 Pew Research study documents that the share of US adults who ``have a great deal or fair amount of confidence in scientists to act in the public's best interests'' declined from 87\% to 73\% between April 2020 and October 2023 with strong level of trust falling from 39\% to 23\% \citep{2023Pew}. At the same time, distrust more than doubled from 12\% to 27\%. While the decline was more pronounced among Republicans, it decreased for both Democrats and Republicans.

%% file: PythonFigures/r_beta_normal_10_3.tex
\begin{tikzpicture}[scale = 0.85]
\begin{axis}[
    domain=-5:5,
    xmin=-5, xmax=5,
    axis lines=center,
    ytick distance=1,
    xlabel=$r$,
    ylabel=$\beta(r)$
    ]
    \addplot[smooth,color=black!30!green] table[x index=1,y index=2,col sep=comma] {PythonFigures/r_beta_normal_10_3.csv} node[above,pos=0.9] {$\beta_{10}(r)$};

    \addplot[smooth,color=blue] table[x index=1,y index=3,col sep=comma] {PythonFigures/r_beta_normal_10_3.csv} node[below,pos=0.8] {$\beta_{3}(r)$};
    
    \addplot[dashed,color=black] coordinates {(-5, -5) (4,4)} node[below,pos=0.9] {$r$};

    \addplot[dashed,color=purple] coordinates {(-5, -3) (3,5)} node[right,pos=0.9] {$r+\Delta$};
    
    \draw[dotted, color=black!30!green]    (-3.56,0)   -- (-3.56,-5);
    \draw[-stealth,ultra thick,black!30!green]    (-3.56,1)   -- (-3.56,0.2) node[above=15] {$-3.56\quad $};

    \draw[dotted, color=blue]    (-2.83,0)   -- (-2.83,-5);
    \draw[-stealth,ultra thick,blue]    (-2.83,1.5)   -- (-2.83,0.2) node[above=25] {$-2.83$};

\end{axis}
\end{tikzpicture}

%% file: PythonFigures/Delta_rstar_normal_10_3.tex
\begin{tikzpicture}[scale = 0.85]
\begin{axis}[
    domain=-3:3,
    xmin=-3, xmax=3,
    axis lines=center,
    ytick distance=1,
    xlabel=$\Delta$,
    ylabel=$r^*$
    ]
    \addplot[smooth,color=black!30!green] table[x index=1,y index=3,col sep=comma] {PythonFigures/delta_on_r.csv} node[above=10,pos=0.2] {$r_{10}(\Delta)$};
    \addplot[smooth,color=blue] table[x index=1,y index=2,col sep=comma] {PythonFigures/delta_on_r.csv} node[below=5,pos=0.1] {$r_{3}(\Delta)$};

    \draw[dotted, color=black]    (0,-2)   -- (3,-2);
    \draw[dotted, color=black!30!green]    (0.51,0)   -- (0.51,-2);
    \draw[dotted, color=blue]    (1.14,0)   -- (1.14,-2);

    \draw[-stealth,ultra thick,black!30!green]    (0.51,1)   -- (0.51,0.2) node[above=15] {$0.51$};
    \draw[-stealth,ultra thick,blue]    (1.14,1.5)   -- (1.14,0.2) node[above=25] {$1.14$};

\end{axis}
\end{tikzpicture}

%% file: PythonFigures/tails_variance_n.tex
\begin{tikzpicture}[scale = 0.85]
\begin{axis}[
    xmode=log,
    domain=10:1000,
    xmin=10, xmax=1000,
    axis lines=left,
    ytick distance=1,
    xlabel=$n$,
    ylabel=$\var\left(\pi_S\right)$
    ]
    \addplot[smooth,color=black!30!green] table[x index=1,y index=3,col sep=comma] {PythonFigures/expert_pareto.csv} node[above=50,pos=0.6] {Pareto (heavy tails)};
    
    \addplot[smooth,color=red] table[x index=1,y index=3,col sep=comma] {PythonFigures/expert_gaussian.csv} node[above=0,pos=0.7] {Gaussian (light tails)};

    \addplot[smooth,color=blue] table[x index=1,y index=3,col sep=comma] {PythonFigures/expert_laplace.csv} node[above=0,pos=0.7] {Laplace (knife-edge)};

\end{axis}
\end{tikzpicture}

%% file: Sections/commitment.tex
\section{Communication with Commitment}\label{sec:commitment}

Up to now we have assumed that the sender cannot commit to a communication scheme. However, sometimes commitment is possible. 
For example, a sender might have a reputation for a particular type of reporting or  a reputable newspaper might commit to always selecting an unbiased set of facts for their articles (and might be punished by readers if they are later found out to have deviated from this communication scheme).  Other examples of commitment include the interpretation of our model as a behavioral game played between current and future self, where current self decides which anecdotes to save to memory so that future self makes the best possible decisions.

We formally defined the commitment equilibrium as a modification of Definition~\ref{def:bne-def} in Section~\ref{sub:modeling assumptions} (see Equation \ref{def:commitment}). As before, Theorems ~\ref{thm:sender.br} and ~\ref{thm:receiver.br} imply that if the sender uses a translation invariant communication scheme $\pi$, then the receiver's best response will be a translation and vice versa. Hence, we again focus on translation invariant commitment equilibria. The following theorem characterizes the set of such equilibria.
\begin{theorem}
\label{thm:stackelberg}
The sender's optimal translation invariant commitment $\pi_S$ is the most informative communication scheme that minimizes only the information loss. In particular, it is a targeting scheme that satisfies the balanced-offset condition $r^* = \beta(r^*)$ and the receiver's response is a translation with shift $M_R - r^*$.
\end{theorem}
The proof is provided in Appendix \ref{app:commitment}. Intuitively, commitment removes the sender's persuasion temptation: since she cannot deceive the receiver, she will ignore the unavoidable disagreement loss and simply minimize the information loss. She therefore will implement the socially optimal communication scheme subject to the constraint that she can only send a single anecdote. Hence, the sender solves the same problem as in Section \ref{sec:persuasion} and behaves as if her personal preferences are aligned with the receiver's. 

\subsection{Characterizing Commitment Equilibria}

Our characterization of most informative communication schemes in Section \ref{sub:most_informative} immediately provides a characterization of the translation invariant commitment equilibria when the sender has \textit{foresight} in the following cases:
\begin{enumerate}
    \item When $n$ is large, the sender's optimal communication scheme is a targeting scheme with a target close to the global maximum density, by Lemma~\ref{lem:most_informative_limit}.
    \item When the anecdote distribution $F$ is symmetric and single-peaked at $0$ then the sender's optimal targeting scheme is the mean scheme by Lemma~\ref{lem:symmetric}.
\end{enumerate}

Intuitively, one might expect that the mean scheme is also an equilibrium \textit{without} foresight. Even though the sender does not know the state of the world, her estimate of the posterior mean will concentrate around the true state (for large $n$) and we therefore would expect the sender to select an anecdote close to the posterior mean since the density of the anecdote distribution is also maximized around the true state. We show that this intuition is correct for a class of distribution which we call \textit{well-behaved}.
\begin{restatable}{definition}{wellbehaved}\label{def:well-behaved}
We say that the anecdote distribution $F$ is well-behaved if the following holds.
\begin{enumerate}
    \item The distribution is strictly single-peaked and symmetric with finite variance. 
    \item Let $g(x) = f'(x)/f(x)$. That is, $g(x) = \frac{d \log f(x)}{d x}$. We assume that $|g'(x)| \le c_1 $ for all $x$, and some constant $c_1>0$. That is, $|g(x)| \le c_1|x|+ c_2$. \footnote{Note that, for $x < 1$ we bound $|g(x)| \le c_2$ and otherwise we can bound  $|g(x)|\le c_1$.}
    \item $F$ has exponential tails. That is, there is a constant $Q > 0$, such that for $x>Q$, we have $1-F(x) \leq c_3\exp\left(-|x| \right)$ for a constant $c_3>0$, and $x < -Q$  we have $F(x) \leq c_3\exp\left(-|x| \right)$.
\end{enumerate}
\end{restatable}
\noindent For example, the normal distribution $F \sim N(0,1)$ and the Laplace distribution with density $f(\epsilon)=\frac{1}{2}\exp(-|\epsilon|)$ are well-behaved.  Recall that $f(\epsilon)$ is the density of the signal distribution at $\theta + \epsilon$, an offset of $\epsilon$ from the true state of the world.  

\begin{restatable}{theorem}{optimalCommitment}\label{thm:unbiased_optimality}
For any well-behaved anecdote distribution, in the no-foresight setting, the unbiased communication scheme that sends the closest signal to $\theta_S(\vec x, y)$ strictly dominates any biased signaling scheme for sufficiently large $n$ and is optimal among all unbiased communication schemes.
\end{restatable}
\iftoggle{onlineappendix}{
The proof is delegated to a separate online appendix. 
}
{
The proof is delegated to Appendix \ref{app:commitment_noforesight}.} Intuitively, the sender would like to send \textit{precisely} the posterior mean to the receiver. However, since she is constrained to sending a signal she has to contend with the second best which is to send the signal closest to the posterior mean.
When we interpret our model as a model of memory where the current self communicates with her future self by storing a single anecdote in memory we can think of the anecdote closest to the posterior mean as the ``most representative anecdote''.

\subsection{Biased Commitment Equilibria under No-Foresight}\label{subsec:commitment:discussion}

We have just shown that for single-peaked and symmetric anecdote distributions the mean communication scheme is optimal under foresight as well as under no-foresight when $n$ is large. 
However, unbiased communication ($r^*=0$) is not necessarily optimal for small $n$ under no-foresight even when the anecdote distribution is single-peaked and symmetric. 

The simplest example is the case where the sender has $n=2$ anecdotes.
\begin{proposition}\label{prop:2anecdotes}    
Suppose $n=2$, signals are drawn from a symmetric distribution around $\theta$, and the sender has no foresight. Then the optimal communication scheme is either the minimum scheme or the maximum scheme.
\end{proposition}
\noindent The proofs of this subsection appear in Appendix~\ref{app:commitment}. The optimal scheme still has to satisfy the balanced-offset condition of Theorem \ref{thm:stackelberg} such that $\beta(r^*)=r^*$. However, there are three distinct candidate targeting schemes:  the mean scheme ($r^*=0$), the minimum scheme ($r^*<0$), and the maximum scheme ($r^*>0$), where the offset (and equivalently the bias) is zero, negative, and positive, respectively. 

The mean scheme incurs a higher level of information loss compared to the minimum and maximum schemes due to a distinctive aspect of the two-signal setting under no-foresight: both anecdotes have the same absolute distance to the posterior mean $\theta_S(\vec{x},y)$ or empirical mean $\frac{{x_1 + x_2}}{2}$. Consequently, the information loss of a communication scheme primarily hinges on the uncertainty regarding whether the posterior mean lies to the left or right of the communicated signal. This uncertainty is minimized by having the mean always lying to the left (or the right), i.e., consistently sending the minimum (or the maximum, respectively).

The biasedness of the optimal communication scheme is not an artifact of the two-signal setting as demonstrated in the following example.
\begin{proposition}\label{prop:3uniform}
Suppose $n=3$, signals are drawn from a uniform distribution around $\theta$, and the sender has no foresight.  Then at every commitment equilibrium the sender uses a communication scheme with non-zero bias.
\end{proposition}
\noindent The idea behind Proposition~\ref{prop:3uniform} is that, conditional on the value of the sender's posterior mean, the conditional density over signal realizations is not necessarily single-peaked. For uniform distributions and $n=3$, the correlation between the posterior mean and the minimum and maximum signals is stronger than the correlation between the posterior mean and the middle signal.  One can therefore communicate more information about the posterior mean through a biased communication scheme ($r^*\neq 0$) that sometimes returns the minimum signal (or, by symmetry, sometimes returns the maximum signal). 

Why does bias help in this example?  Recall that there is intrinsic error in the sender's posterior mean.  This variance is unavoidable.  But it can introduce correlation with particular anecdotes.  This correlation can be used to help minimize the variance between the posterior mean and the anecdote passed to the receiver.  This is why, for the uniform case, it is helpful to bias toward more extreme signals: even though they are not more informative than the moderate signals when it comes to the true state of the world, they are more informative with respect to the sender's posterior mean.  The interplay between these two sources of errors therefore introduces an incentive for the sender to systematically bias their communication.

However, for large $n$ the correlation between the posterior mean and the closest signal to the posterior mean dissipates and therefore the unbiased mean scheme re-emerges as the optimal communication scheme under no-foresight as shown in Theorem~\ref{thm:unbiased_optimality}.

%% file: Sections/conclusion.tex
\section{Conclusion}\label{sec:conclusion}

We introduce a model of strategic communication where the sender communicates raw data points (anecdotes) that are informative about an underlying state of the world. While this type of anecdotal communication is less efficient than communicating posterior beliefs about the state directly, it does not require that the sender and receiver have a shared representation of the state of the world. Therefore, our model can help us understand how newspapers convey news to their diverse readership by assembling articles from various quotes and factual statements, or how politicians rely on examples and anecdotes to connect with their voters. We have shown four main results. First of all, our model naturally gives rise to polarization as the sender balances the need to inform the receiver with the temptation to persuade him which makes her target biased anecdotes. Second, polarization increases, and welfare decreases with the misalignment between sender's and receiver's preferences. Third, polarization is amplified when experts communicate because they have access to more outlier anecdotes. Fourth, polarization will hurt both the sender and the receiver when the anecdote distribution is heavy-tailed giving rise to the curse of informedness. 

\paragraph{Stories as ``Carriers of Beliefs''.}
On a conceptual level, our model fits within a paradigm where stories or anecdotes provide the basis for both communicating as well as storing beliefs. While we mainly focus on the former interpretation in this paper, we also explained that our model can be viewed as the current self (sender) communicating with a future self (receiver). This can be expanded into a theory of \textit{optimal storage of stories} for a decision-maker who has limited memory and can only store information about the state of the world through representative data points rather than posterior beliefs. We believe that this model of storing ``beliefs'' is appropriate in many situations where we need to learn about the world without being an expert. For instance, when describing notable traits of a colleague, friend, or acquaintance -- like kindness, intelligence, or creativity -- we could illustrate them through specific anecdotes. These anecdotes might showcase instances when the individual demonstrated kindness, solved a challenging problem, or devised an innovative solution, rather than relying solely on formal metrics such as like altruism or IQ scores. 

\paragraph{Extensions.}
We hope that our model can serve as a starting point for future extensions. For example, our model assumes that the sender can only send a single anecdote. This makes our analysis particularly simple because we can focus on translation invariant equilibria. However, senders can and do communicate more than one anecdote. For example, a 500-word news article contains on average 10 to 20 sentences and 5 to 10 facts \citep{2024NewsAsLego}.  How can one adapt our analysis? We can gain some intuition by looking at the special case where sender and receiver are aligned ($\Delta=0$), the sender has foresight and the anecdote distribution is symmetric with a single peak at $0$. If we focus on some small enough interval of width $\delta$ around the state of the world where the anecdote density does not change much, we can think of $\delta n$ anecdotes that are approximately uniformly distributed over this interval. If the sender can only send one anecdote she will send the closest one to the state of the world which is an expected distance $O(1/n)$ away from the true state and the (quadratic) information loss is therefore $O(1/n^2)$. With two anecdotes, the sender could combine any of approximately $\delta n/2$ anecdotes to the left and right of the true state and choose among these $O(n^2)$ combinations the one whose average is closest to the true state. We conjecture that the expected distance of this best average is now only $O(1/n^2)$ away from the true state and the information loss is therefore only $O(1/n^4)$. This heuristic argument suggests that the information loss from anecdotal communication declines rapidly with the number of anecdotes that the sender can transmit. The analysis in Sections \ref{sec:equilibrium} and \ref{sec:experts} then suggests that polarization becomes even more extreme when sender and receiver are no longer aligned which might mitigate or even reverse any reductions in the information loss that are achievable in the optimal case.

%% file: Sections/appendix_diffuse_prior.tex
\section{Diffuse Prior }\label{app:diffuse-prior}
In this section we discuss our assumption of a diffuse prior and its implications on the posterior of the agents. Throughout, we restrict attention to translation invariant sender and receiver strategies that we introduce in Section \ref{sec:translation_invariance}. We start by showing that sender and receiver beliefs \textit{after} sending an anecdote are well defined if the prior on the state of the world $\theta$ is diffuse. We then show that players' ex-ante payoffs \textit{before} a message is sent are also well-defined.  

\subsection{Posterior Distributions of the Sender and Receiver}

In statistics, it is common to use an \emph{improper prior} as uninformative priors. The simplest way to formalize a diffuse prior that reveals no information about $\theta$ is to consider the density to be a constant $\mu(\theta) = c$ for all $\theta\in \mathbb{R}$. It is important to note that while $\mu$ is not a proper probability distribution (since $\int_\theta \mu(\theta) d\theta = \infty$), it is still possible that the posterior formed can still be proper and well defined\footnote{Also, in a more Frequentist view, $\mu$ can be thought of as the likelihood function to capture the absence of data.}. 

\begin{claim}\label{claim:diffuse-prior}
Given a diffuse prior and $\vec x\sim \theta + F$, $y \sim \theta + G$, the posterior distribution of the sender conditioned on $\vec x, y$, is a proper distribution. Moreover, when $F$ and $G$ are symmetric, the posterior mean $\theta_S(\vec x,y)$ is an unbiased estimator of $\theta$.

Similarly, given a diffuse prior, a translation invariant communication scheme $\pi$, and a signal $x= \pi(\vec x,y)$. The posterior distribution of the receiver conditioned on $\pi, x$, is a proper distribution. Further, the posterior mean of the receiver is $\theta_R(\pi, x) = x - \beta(\pi)$.
\end{claim}

We prove this claim below.

\paragraph{Sender's posterior distribution.} Recall that the anecdotes $x_1,\ldots,x_n$ are drawn independently from $\theta + F$. Thus the pdf of an anecdote given $\theta$ is $f(x - \theta)$.  Similarly, $y$ is drawn from $\theta + G$ and hence the pdf of $y$ given $\theta$ is $g(y-\theta)$.

We first observe that in the foresight case, when $G=0$, the sender's posterior is a point mass at $y$.

For any $G$ that is a proper distribution, we see that the posterior of $\theta$ given $y$ is proper.
\begin{align*}
    \mu(\theta|y) &= \frac{g(y-\theta)\mu(\theta)}{\int_{\hat\theta} g(y-\hat\theta) \mu(\hat\theta)d\hat\theta} \\
    &= \frac{g(y-\theta)}{\int_{\hat\theta} g(y-\hat\theta)d\hat\theta} \\
    &= \frac{g(y-\theta)}{\int_{\gamma} g(\gamma)d\gamma} = g(y-\theta)
\end{align*}

The first equality is just the definition of a posterior, and the second equality holds since $\mu$ is the diffuse prior with $\mu(\theta) = c$ for all $\theta$. The third equality does a change of variables to $\gamma = y - \hat\theta$. Finally, the last step follows because $G$ is a proper distribution. Hence $\mu(\theta| y)$ is a proper posterior distribution.

Recall that $\mu(\theta| \vec x, y) = \frac{\hat{f}(\vec x | y , \theta) \mu(\theta| y )}{ \int_{\hat\theta} \hat{f}(\vec x | y , \hat\theta) \mu(\hat\theta| y ) d\hat\theta}$, where $\hat f (\vec x | \theta ,y)$ is the conditional pdf of $\vec x$ given $\theta, y$. That is, we can use $\mu(\theta| y)$ as a prior. Since $\mu(\theta| y)$ is a proper distribution, the posterior $\mu(\theta| \vec x, y)$ is also proper.

For the non-foresight case, when $G$ is diffuse, we can use a similar argument as above to first compute the posterior given $x_1$ and $y$. We get,

\begin{align*}
    \mu(\theta| x_1, y) & =  \frac{f(x_1 -\theta) g(y-\theta)\mu(\theta)}{\int_{\hat\theta} f(x_1 - \hat\theta)g(y-\hat\theta) \mu(\hat\theta)d\hat\theta} \\
    &= \frac{f(x_1-\theta)}{\int_{\epsilon} f(\epsilon)d\epsilon}\\
    & = f(x_1 - \theta)
\end{align*}

This is again by noting that $\mu(\theta) = c$ , $g(y-\theta) = c'$, and doing a change of variable to $\epsilon = x_1 - \hat\theta$. Thus $\mu(\theta| x_1, y)$ is a proper posterior distribution because $f$ is proper distribution. Now using this as a prior, we get that $\mu(\theta| \vec x ,y)$ is a proper posterior.

\paragraph{Sender's posterior mean.} We observe that, given a diffuse prior and symmetric anecdote distributions $F,G$, $\mu(\theta| \vec x , y) = \mu(-\theta| - \vec x, -y)$. With this, it is easy to see that, for $\theta = 0$, $\E_{\vec x, y}[\theta_S(\vec x, y)|\theta = 0] = 0$. Moreover, we show below that $\theta_S(\vec x, y) = \theta_S(\vec x + t, y + t) - t$, and hence $\E_{\vec x, y}[\theta_S(\vec x, y)|\theta] = \theta$. Thus, the sender's posterior mean is an unbiased estimator of $\theta$.

We will now show that $\mu(\theta| \vec x, y ) = \mu(\theta + t| \vec x+t,y+t)$, and this would imply,
\[\theta_S(\vec x, y) = \int_\theta \theta\cdot \mu(\theta| \vec x , y) d\theta = \int_\theta \theta\cdot \mu(\theta + t| \vec x +t, y +t) d\theta = \theta_S(\vec x+t, y+t) - t. \]

We have $\mu(\theta| \vec x, y ) = \mu(\theta + t| \vec x+t,y+t)$ because,
\begin{align*}
    \mu(\theta|\vec x ,y) &= \frac{\prod_i f(x_i - \theta) \cdot g(y - \theta) \mu(\theta)}{\int_{\hat{\theta}}\prod_i f(x_i - \hat\theta) \cdot g(y - \hat\theta) \mu(\hat\theta)d\hat\theta} \\
    &=\frac{\prod_i f(x_i +t - \theta -t) \cdot g(y +t - \theta-t) \mu(\theta + t)}{\int_{\hat{\theta}}\prod_i f(x_i + t - \hat\theta - t) \cdot g(y - \hat\theta) \mu(\hat\theta)d\hat\theta} \\
    &= \mu(\theta+t|\vec x+t ,y+t)
\end{align*}

\paragraph{Receiver's posterior distribution.} We show that the receiver's posterior distribution given a translation invariant $\pi$ and $x=\pi(\vec x ,y)$ is a proper distribution. Let $h_{\pi}(x|\theta)$ be the pdf of the signal sent given $\pi$ and $\theta$. Observe that, by definition of translation invariant, $\pi(\vec x - \theta, y - \theta) = \pi (\vec x, y) - \theta$. Therefore, $h_{\pi}(x | \hat\theta) = h_{\pi}(x - \hat\theta|0)$. Note that, $h_{\pi}(\cdot|0)$ only depends on $\pi, F,$ and $G$.

\begin{align*}
    \mu(\theta| \pi,x) & = \frac{h_{\pi}(x|\theta)\mu(\theta)}{\int_{\hat\theta} h_{\pi}(x|\hat\theta)\mu(\hat\theta)d\hat\theta}\\
    &=\frac{h_{\pi}(x - \theta| 0)}{\int_{\hat\theta} h_{\pi}(x - \hat\theta|0)d\hat\theta}\\
    &= \frac{h_{\pi}(x - \theta|\theta = 0)}{\int_{\epsilon} h_{\pi}(\epsilon| 0)d\epsilon} = h_{\pi}(x - \theta| 0)
\end{align*}

\paragraph{Receiver's posterior mean.} Given a translation invariant $\pi$, for any $\theta$, recall that the bias $\beta(\pi) = \int_x (x - \theta)h_\pi(x|\theta) = \int_z z h_\pi(z|0)$. Hence we get that the posterior mean of the sender $\theta_R(\pi,x)$ is

\begin{align*}
  \int_{\hat\theta} \hat\theta \mu(\hat\theta|\pi,x)d\hat\theta &=  \int_{\hat\theta} \hat\theta h_{\pi}(x - \hat\theta| 0)d\hat\theta  \\
  &= \int_{z} (x - z) h_{\pi}(z| 0)dz \\
  &= x - \int_{z} z h_{\pi}(z| 0)dz \\
  &= x - \beta(\pi)
\end{align*}

Thus for any $\theta$, translation invariant $\pi$, we get $\E_{x}[\theta_R(\pi,x)|\theta]= \E_x[x - \beta(\pi)|\theta] = \theta$ (follows directly from the definition of $\beta(\pi)$). Therefore, the receiver's posterior mean $x - \beta(\pi)$ is an unbiased estimator of $\theta$.

\subsection{Ex Ante Payoffs}

When describing payoffs and equilibria in our communication game, it is also important to understand the players' beliefs about payoffs \textit{before} any messages are received.  Concerns about the well-definedness of ex ante payoffs in signaling games with diffuse priors were also noted recently by \citet{ambrus2021defining}.  They provide sufficient conditions under which such a game's ex ante payoffs are well-defined for a class of strategies that include translation invariant strategies.  In this section we describe their conditions and results and show that they are satisfied by our model under some modifications that are without loss of generality for translation invariant strategies.  This will justify our use of expectations over $\theta$ when analyzing payoffs and beliefs, despite $\theta$ being drawn from an improper prior.

We now describe the setting of~\citet{ambrus2021defining} in the context of our model.  There is a state of the world $\theta$ and $N$ players.  Each player receives a collection of signals, where each signal $s_{ij}$ observed by player $i$ is drawn independently as $\theta + \hat{s}_{ij}$ where $\hat{s}_{ij} \sim F_{ij}$ admits positive density.\footnote{The model of~\citet{ambrus2021defining} assumes at most one signal per player, but their analysis holds without change if one player can receive multiple independent signals.}  Each player $i$ then chooses an action $a_i$ from a set $A_i \subseteq \reals$ of allowable actions that can depend on the set of signals received.  Each player $i$ then receives a payoff $u_i$ that depends on $\theta$ and the profile $\vec a$ of actions chosen. It is assumed that payoffs are translation invariant, in the sense that $u_i(\theta+\delta, \vec a+\delta) = u_i(\theta, \vec a)$ for all $\theta \in \reals$, $\vec a \in \reals^n$, and $\delta \in \reals$.  A strategy $\mu_i$ is a (possibly randomized) mapping from received signals to an action.  We write $u_i(\theta, \vec \mu)$ for the expected utility obtained under strategy profile $\vec \mu$ where the expectation is over the realization of signals and any randomness in $\vec \mu$.

\begin{condition}[Compactness and Translation Invariance]
\label{cond:compact}
For each player, given any set of signals received, the set of potential actions is a compact set.  Moreover, whenever all signals are offset by a real number $\delta \neq 0$, then the set of potential actions is likewise offset by $\delta$.
\end{condition}

\begin{condition}[Bounded Interim Payoffs]
\label{cond:bounded}
For each player $i$ there is some finite constant $C_i$ such that, for any realization of $\theta$, the expectation (over signal realizations) of $\max_{\vec a}|u_i(\theta, \vec{a})$ is at most $C_i$, where the maximum is over action profiles that are allowable given the realized signals.
\end{condition}

\begin{condition}[Irreducibility]
\label{cond:irreducible}
For each player $i$ there is no pair of distinct strategies $\mu_i$, $\mu'_i$ such that the expected payoffs of $\mu_i$ and $\mu'_i$ are identical for all $\theta$ and all translation invariant strategies of the other players.
\end{condition}

Roughly speaking,~\citet{ambrus2021defining} show that if Conditions~\ref{cond:compact},~\ref{cond:bounded}, and~\ref{cond:irreducible} are satisfied, then ex ante expected payoffs under any profile of translation invariant strategies is well-defined and equal to the expected payoffs conditional on any fixed value of $\theta$.  To state their result more formally, we first require some definitions.

\begin{definition}[Diffusing Sequence]
Let $\lambda$ denote the Lebesgue measure.  We say a sequence $(P_m)_{m \in \mathbb{N}}$ of Borel probability measures on $\reals$ is \emph{diffusing} if for any set $W$ with $\lambda(W) \in (0,\infty)$ and any $\eta > 0$, there exists $M \in \mathbb{N}$ such that for all $m > M$,
\begin{itemize}
    \item $P_m(W) > 0$ and
    \item for all measurable $Y \subseteq W$, $\left|\frac{P_m(Y)}{P_m(W)} - \frac{\lambda(Y)}{\lambda(W)}\right| < \eta$.
\end{itemize}
\end{definition}

Roughly speaking, a diffusing sequence is an infinite sequence of probability distributions that, in the limit, tends toward an improper distribution that behaves as a diffuse prior.  For example, the sequence of distributions $(N(0,m))$ (i.e., $P_m$ is a mean-zero Gaussian with variance $m$) is a diffusing sequence.

\begin{definition}[Admissibility]
A class of strategy profiles is said to be \emph{admissible} if for any profile $\vec \mu = (\mu_1, \dotsc, \mu_N)$ of strategies in the class, there exists a vector $u^* \in \reals^N$ such that for any diffusing sequence $(P_m)$ and every player $i$, $\lim_{m \to \infty} \E_{\theta \sim P_m}[u_i(\theta, \vec \mu)] = u_i^*$.
\end{definition}

We are now ready to state the result of~\citet{ambrus2021defining}.

\begin{theorem}[\citep{ambrus2021defining}]
\label{thm:admissible}
Suppose the game satisfies Conditions~\ref{cond:compact},~\ref{cond:bounded}, and~\ref{cond:irreducible}.  Then the class of profiles of translation invariant strategies is admissible.\footnote{In fact, they show a more general result: a weaker condition than translation invariance, in which the strategies tend toward a translation invariant limit strategy in the limit as $\theta$ tends to $\pm \infty$, is necessary and sufficient for the ex ante payoffs to be well-defined and determined by payoffs in the limit as $\theta$ tends to $\pm \infty$.  We will only use the sufficient condition of full translation invariance.}
\end{theorem}

We now show how to apply Theorem~\ref{thm:admissible} to our signaling game.  As the theorem applies to translation invariant strategies, we will focus attention on translation invariant strategies in our game. Recall that the sender's strategy is a communication scheme $\pi$, and a translation invariant strategy of the receiver is a translation action rule which is characterized by an offset $\sigma \in \reals$.  We can therefore map our game into the setting of~\cite{ambrus2021defining} as follows: the sender receives signals $(\vec x, y)$ and selects from the set $\{x_1, \dotsc, x_n\}$.  The receiver is not given any signals and chooses an offset $\sigma$.  If the sender chooses $x$ and the receiver chooses offset $\sigma$, they both receive payoffs as defined in our original game when the receiver takes action $x + \sigma$.  Note that this formulation is a simultaneous-move game and that payoffs are indeed translation invariant with respect to $\theta$ and the received signals.

We would like to apply Theorem~\ref{thm:admissible} to this reformulation of our signaling game.  However, we note that Conditions~\ref{cond:compact} and~\ref{cond:bounded} do not directly apply in our model since the action space of the receiver in unbounded.  Indeed, the receiver in our game can choose any offset $\sigma \in \reals$ (an unbounded set).  Moreover, the resulting cost to the sender and receiver can be arbitrarily high for extreme choices of $\sigma$.  However, as we now show, it is without loss to restrict attention to a bounded space of receiver actions.

\begin{lemma}
\label{lem:bounded}
Given $F$, $G$, and $n$, there exists $C > 0$ such that any choice of offset $\sigma > C$ is strictly dominated by $\sigma = C$, and any choice of offset $\sigma < -C$ is strictly dominated by $\sigma = -C$.
\end{lemma}
\begin{proof}
Over the space of translation invariant communication schemes for the sender, the one that induces the (pointwise) maximal posterior over $\theta$ for any given message $x$ is the one that always returns the minimum anecdote $\min\{x_1, \dotsc, x_n\}$.  Note that $\E[\min\{x_1, \dotsc, x_n\}] \geq \E[\theta - \sum_i |x_i - \theta|] \geq \theta - nc$ where $c = \E_{\epsilon \sim F}[|\epsilon|]$ is a constant depending on $F$. In particular, this means that $E[D_{\pi,x}] \leq x + nc$ for \emph{every} choice of translation invariant $\pi$. Thus, for any realization of $x$, a receiver action of $a > x + nc + M_R$ is strictly dominated by action $a = x + nc + M_R$. A symmetric argument, for the maximum scheme, shows that a receiver action of $a < x - nc + M_R$ is strictly dominated by $a = x + nc + M_R$.  Taking $C = nc + |M_R|$ therefore satisfies the conditions of the lemma.
\end{proof}

Lemma~\ref{lem:bounded} shows that when we restrict the sender to translation invariant communication schemes, it is without loss of generality to restrict the receiver to choosing actions that lie in a compact set centered on the received signal.  As the sender is also mechanically restricted to choosing (a distribution over) one of the received anecdotes as an action, which is compact given the anecdote realization and invariant to translations, we conclude that Condition~\ref{cond:compact} is satisfied under this restriction. Likewise, under the restriction (without loss) that the receiver chooses an action from this bounded set, the sender and receiver's expected utilities are likewise bounded uniformly over the choice of $\theta$ (with the extremal value occurring when the sender chooses the maximum anecdote, and the receiver always takes the maximum allowable action $x + C$ given the sender's message).  Thus Condition~\ref{cond:bounded} is satisfied as well.  

Finally, we note that Condition~\ref{cond:irreducible} is satisfied as well.  Indeed, any two distinct strategies (i.e., communication schemes) of the sender, say $\pi$ and $\pi'$ with $\pi(\vec x) \neq \pi'(\vec x)$, would result in different outcomes (and hence payoffs, if we fix $\theta$) on anecdote realization $\vec x$ when the receiver chooses the identity action rule $\alpha(x) = x$ (i.e., offset $\sigma = 0$). On the other hand, any two distinct choices of offset $\sigma \neq \sigma'$ would result in different payoffs to the sender and receiver for any given choice $x$ of the sender.

We conclude that Conditions~\ref{cond:compact},~\ref{cond:bounded}, and~\ref{cond:irreducible} all hold after we eliminate strictly dominated strategies from the space of translation invariant strategies.  We can therefore conclude from Theorem~\ref{thm:admissible} that any profile of undominated translation invariant strategies $(\alpha,\pi)$ is \emph{admissible}. They induce expected payoffs equal to the limit of payoffs (for those strategies) under any diffusing sequence of priors over $\theta$.  The notion of expected payoffs of given strategies with respect to $\theta$ is therefore well-defined, equal to the expected payoffs under \emph{any} realization of $\theta$, and consistent with our posterior belief calculations.

%% file: Sections/apppendix_model.tex
\section{Proofs for Section \ref{sec:model}}\label{app:model}

\begin{proof}[Proof of Theorem~\ref{thm:sender.br}]
Recall that the sender wishes to maximize $u_S(a,\theta)=-(a-(\theta+M_S))^2$.  Since $a = \pi(\vec{x},y)+\sigma(\alpha)$ from the definition of shift, the sender's goal is to choose $\pi$ so that $\pi(\vec{x},y)$ maximizes 
\[ -\E_{\theta}[(\pi(\vec{x},y)+\sigma(\alpha)-(\theta+M_S))^2\ |\ (\vec{x},y)].\]  
This is the expectation of a quadratic loss. Using bias-variance decomposition and the fact that $\theta_S(\vec{x},y) = \E_{\theta}[\theta\ |\ (\vec{x},y)]$ is an unbiased estimator of $\theta$ and the variance of $\theta_S(\vec{x},y)$ is a constant that is independent of $\pi(\cdot)$, this goal is achieved by choosing $\pi(\vec{x},y)$ to maximize \[-(\pi(\vec{x},y)+\sigma(\alpha)-(\theta_S(\vec{x},y)+M_S))^2 
\]
for each $\vec{x}$ and $y$.  For any realization of $\vec{x}$ and $y$, this expression is maximized by setting $\pi(\vec{x},y)$ as close as possible to $\theta_S(\vec{x},y) + M_S - \sigma(\alpha)$.  Since the only constraint on $\pi(\vec{x},y)$ is that it be chosen from the profile of anecdotes $\vec{x}$, the result follows.
\end{proof}

\begin{proof}[Proof of Theorem~\ref{thm:receiver.br}]
Let $x=\pi(\vec{x})$ and let the belief distribution be $B(x) = D_{\pi,x}$. Given any translation invariant $\pi$, the posterior distribution $D_{\pi, x}$ is a proper distribution (see Claim~\ref{claim:diffuse-prior} in Appendix~\ref{app:diffuse-prior} for more details).  Since bias equals the expected difference between the communicated anecdote and $\theta$, we have that $x-\beta(\pi)$ is an unbiased estimate of $\theta$.
Since we work with diffuse priors, the receiver's posterior mean about the state of the world is simply equal to the value of the unbiased estimator (formalized in Claim~\ref{claim:diffuse-prior} in Appendix~\ref{app:diffuse-prior}).
The receiver wishes to maximize 
\begin{align*}
\E_{\theta\sim D_{\pi,x}}[u_R(\alpha(x),\theta)] 
&=\E_{\theta\sim D_{\pi,x}}[-(\alpha(x)-(\theta+M_R))^2]\\
&=-(\alpha(x)-(x-\beta(\pi)+M_R))^2 - \E_{\theta\sim D_{\pi, x}}[(\theta - x + \beta(\pi))^2]
\end{align*}
where the second equality is the bias-variance decomposition and follows because $x- \beta(\pi)$ is an unbiased estimator of $\theta$.
This maximum is achieved for $\alpha(x)=x-\beta(\pi)+M_R$, i.e., a translation of $M_R-\beta(\pi)$ as claimed.
\end{proof}

\begin{proof}[Proof of Proposition~\ref{prop:sender_summary}]
The proof follows by direct manipulations of the sender's expected utility function.
\begin{align*}
\E_{\theta, \vec{x}, y}[u_S(\alpha(\pi_S(\vec{x},y)), \theta)]
&=
\E_{\theta,\vec{x},y}[((\alpha(\pi_S(\vec{x},y))-(\theta+M_S))^2]\\
&=
\E_{\vec{x},y}[((\alpha(\pi_S(\vec{x},y))-(\theta_S(\vec{x},y)+M_S))^2]\\
&=\E_{\vec{x},y}[(\pi_S(\vec{x},y)+\sigma(\alpha)-(\theta_S(\vec{x},y)+M_S))^2]\\
&=\E_{\vec{x},y}[(\pi_S(\vec{x},y)+ \sigma(\alpha)-(\theta_S(\vec{x},y)+M_S)+\beta(\pi_S)-\beta(\pi_S))^2]\\
&=\E_{\vec{x},y}[(\pi_S(\vec{x},y)-(\theta_S(\vec{x},y)+\beta(\pi_S))+\sigma(\alpha)- (M_S - \beta(\pi_S))^2]
\end{align*}
Let $w=\pi_S(\vec{x},y)-(\theta_S(\vec{x},y)+\beta(\pi_S))$ and $z=\sigma(\alpha)- (M_S - \beta(\pi_S))$ so that the above expectation is $\E[(w+z)^2]=\E[w^2+z^2+2wz]$.  Note
\begin{align*}
\E[wz]
&=z\E_{\vec{x},y}[(\pi_S(\vec{x},y)-(\theta_S(\vec{x},y)+\beta(\pi_S)))]=z\E_{\theta,\vec{x},y}[(\pi_S(\vec{x},y)-(\theta+\beta(\pi_S)))]
\end{align*}
where the second inequality follows because $\theta_S(\vec{x},y)$ is a valid posterior mean for $\theta$.  But the right-hand side is zero by definition of bias.  Therefore $\E[wz] = 0$ and hence the claim follows.
\end{proof}

\begin{proof}[Proof of Proposition~\ref{prop:receiver-loss}]
From Theorem~\ref{thm:receiver.br} we see that $\alpha(x) = x + M_R - \beta(\pi)$, where the sender sends signal $x = \pi(\vec{x},y)$ and $\beta(\pi)$ is the bias of the communication scheme $\pi$. Thus, we have receiver's loss  (for any fixed $\theta$) equals

\begin{align*}
\E_{\vec{x},y}[(\pi(\vec{x},y) + M_R - \beta(\pi) - \theta-M_R)^2] & = \E_{\vec{x},y}[(\pi(\vec{x},y) - \theta  -\beta(\pi))^2] \\
    &= \E_{\vec{x},y}[(\pi(\vec{x},y) - \E_{\vec{x},y}[\pi(\vec{x},y)])^2]
\end{align*} 
since by definition of bias $\E_{\vec{x}\sim F_{\theta},y \sim G_{\theta}}[\pi(\vec{x},y)] = \beta(\pi) + \theta$.
\end{proof}

%% file: Sections/appendix_persuasion.tex
\section{Proofs for Section \ref{sec:persuasion}}\label{app:persuasion}

\begin{proof}[Proof of Proposition~\ref{prop:sender_summary}]
The proof follows by direct manipulations of the sender's expected utility function.

\begin{align*}
\E_{\theta, \vec{x}, y}[u_S(\alpha(\pi_S(\vec{x},y)), \theta)]
&=
\E_{\theta,\vec{x},y}[((\alpha(\pi_S(\vec{x},y))-(\theta+M_S))^2]\\
&=
\E_{\vec{x},y}[((\alpha(\pi_S(\vec{x},y))-(\theta_S(\vec{x},y)+M_S))^2]\\
&=\E_{\vec{x},y}[(\pi_S(\vec{x},y)+\sigma(\alpha)-(\theta_S(\vec{x},y)+M_S))^2]\\
&=\E_{\vec{x},y}[(\pi_S(\vec{x},y)+ \sigma(\alpha)-(\theta_S(\vec{x},y)+M_S)+\beta(\pi_S)-\beta(\pi_S))^2]\\
&=\E_{\vec{x},y}[(\pi_S(\vec{x},y)-(\theta_S(\vec{x},y)+\beta(\pi_S))+\sigma(\alpha)- (M_S - \beta(\pi_S))^2]
\end{align*}

\noindent Let $w=\pi_S(\vec{x},y)-(\theta_S(\vec{x},y)+\beta(\pi_S))$ and $z=\sigma(\alpha)- (M_S - \beta(\pi_S))$ so that the above expectation is $\E[(w+z)^2]=\E[w^2+z^2+2wz]$.  Note
\begin{align*}
\E[wz]
&=z\E_{\vec{x},y}[(\pi_S(\vec{x},y)-(\theta_S(\vec{x},y)+\beta(\pi_S)))]=z\E_{\theta,\vec{x},y}[(\pi_S(\vec{x},y)-(\theta+\beta(\pi_S)))]
\end{align*}
where the second inequality follows because $\theta_S(\vec{x},y)$ is a valid posterior mean for $\theta$.  But the right-hand side is zero by definition of bias.  Therefore $\E[wz] = 0$ and hence the claim follows.
\end{proof}

\begin{proof}[Proof of Theorem~\ref{thm:most_informative}]

We prove the Theorem through contradiction. Assume, that $\pi^*$ is a most informative communication scheme with bias $\beta^*$ and utility $U(\pi^*)$ which is not a targeting scheme satisfying $r=\beta^*$.

Consider the targeting scheme $\pi^T$ with offset $r=\beta^*$ and associated bias $\beta^T$. We know that $\beta_T \neq \beta^*$ - otherwise, we would have found a targeting with bias $\beta^*$ which minimizes the Euclidean distance to $\beta^*+\theta_S(\vec{x},y)$ and therefore $U(\pi^T)>U(\pi^*)$ because $\pi^* \neq \pi^T$. Therefore, $\pi^*$ would not be most informative which is a contradiction.

We now calculate the sender's utility when using targeting scheme $\pi^T$:
\begin{eqnarray}
        U(\pi^T)&=&-\E\left[\left(\pi^T(\vec{x},y)-\beta^T-\theta_S(\vec{x},y)\right)^2\right]\nonumber \\
        &=& -\E\left[\left(\pi^T(\vec{x},y)-\beta^*-\theta_S(\vec{x},y)+\beta^*-\beta^T\right)^2\right]\nonumber \\
        &=& -\E\left[\left(\pi^T(\vec{x},y)-\beta^*-\theta_S(\vec{x},y)\right)^2\right]-2(\beta^*-\beta^T)\underbrace{\E\left[\hat{\pi}(\vec{x},y)-\beta^*-\theta_S(\vec{x},y)\right]}_{\beta^T-\beta^*}- (\beta^*-\beta^T)^2\nonumber \\
        &=& -\E\underbrace{\left[\left(\pi^T(\vec{x},y)-\beta^*-\theta_S(\vec{x},y)\right)^2\right]}_{<\left(\pi^*(\vec{x},y)-\beta^*-\theta_S(\vec{x},y)\right)^2}+(\beta^*-\beta^T)^2 > U(\pi^*)
\end{eqnarray}
But this implies that $\pi^*$ is not a most informative communication scheme which is a contradiction.
\end{proof}

\begin{proof}[Proof of Lemma~\ref{lem:symmetric}]
First, it is easy to see that the balanced-offset condition is satisfied at $r^*=0$: the closest anecdote is equally likely to be on the left or right because of symmetry.

Now consider any $r<0$ and any realization $\vec{x}$ of $n$ anecdotes. We construct a coupled realization $\vec{\tilde{x}}$ as follows: for any $x_i<r$ we keep the anecdote and $\tilde{x_i}=x_i$. For any $r\leq x_i \leq -r$ we map the anecdote to to $\tilde{x_i}=r$ and for any $x_i>-r$ we map the anecdote to $\tilde{x_i}=x_i-2r$.

We can see that the coupled distribution is symmetric around $r$ and therefore $\tilde{\beta}(r)=r$. However, because of the coupling all the anecdotes to the right of $r$ are closer to the offset than for the coupled original realization. Therefore, $\tilde{\beta}(r)<\beta(r)$ and hence $r<\beta(r)$. Hence, the targeting scheme with offset $r$ cannot be most informative.
\end{proof}

\begin{proof}[Proof of Lemma~\ref{lem:most_informative_limit}]
We prove the claim by contradiction. If the lemma does not hold then there is some $\epsilon^*>0$ and a sequence $(n_k)$ and associated most informative offsets $(r^*_{n_k})$ that are more than $\epsilon^*$ away from any global maximum of the anecdote density.

First, we observe that there has to be some $\delta>0$ such that $f(r^*_{n_k})<f_{\max}-\delta$ for all $k$ where $f_{max}$ is the global maximum reached by the density function. Suppose not, there would be a subsequence $(n_{k'})$ of $(n_k)$ such that $f(r^*_{n_{k'}})\rightarrow f_{max}$. Moreover, there has to be a subsequence that lies in some bounded interval - otherwise because of uniform continuity the integral of the density function could not be $1$. But this implies by Bolzano-Weierstrass that there is convergent subsequence $(n_{k''})$ of $(n_{k'})$. By continuity, $f(\lim_{k''\rightarrow \infty} r^*_{n_{k''}})=f_{max}$ and therefore $r^*_{n_{k''}}$ eventually is $\epsilon^*$-close to a global maximum which is a contradiction.

We can therefore assume that there exists some $\delta>0$ such that $f(r^*_{n_k})<f_{\max}-\delta$ for all $k$. Because of uniform continuity we can find, for any $\delta'>0$ an $\epsilon'>0$ such that $|f(x)-f(r^*_{n_k})|<\delta'$ for any $x-r^*_{n_k}<\epsilon'$. 

The information loss, or equivalently the variance of the closest anecdote to the offset $r^*_{n_k}$, can then be expressed as:
\[
\var(\pi_{r^*_{n_k}}(\vec{x})) = \E[(\pi_{r^*_{n_k}}(\vec x) - (\theta + r^*_{n_k}))^2], 
\]
the expected squared distance of the closest anecdote to the point $\theta + r^*_{n_k}$, because most informative communication schemes satisfy the balanced-offset condition $\beta(r^*_{n_k}) = r^*_{n_k}$.
Hence, we have, 
\begin{eqnarray}
    \var(\pi_{r^*_{n_k}}(\vec{x}))&=& \int_0^{\infty} nz^2 \left[f(r^*_{n_k}+z)+f(r^*_{n_k}-z)\right] \left[1-F(r^*_{n_k}+z)+F(r^*_{n_k}-z)\right]^{n-1} dz \nonumber \\
    &\geq& \int_0^{\epsilon'} nz^2 2(f(r^*_{n_k})-\delta') \left[1-2(f(r^*_{n_k})+\delta')z\right]^{n-1} dz \nonumber \\
    &=& \frac{f(r^*_{n_k})-\delta'}{f(r^*_{n_k})+\delta'} \frac{1}{2n^2 (f(r^*_{n_k})\dmedit{+}\delta')^2} + o(1/n^2) \nonumber \\
    &>& \frac{f(r^*_{n_k})-\delta'}{f(r^*_{n_k})+\delta'} \frac{1}{2n^2 (f_{\max}-\delta\dmedit{+}\delta')^2} + o(1/n^2),
\end{eqnarray}
where the first inequality follows from the uniform continuity of $f$ and let $\epsilon'< 1/2(f(r^*_{n_k})+\delta')$, the next equality follows simply because the $o(1/n^2)$ terms includes subtractions that are exponentially small in $n$, and finally the last inequality follows by our assumption that $r^*_{n_k}$ is bounded away from the global maxima.

We can compare this loss to the information loss when choosing a targeting scheme with offset $\tilde{r}$ such that $f(\tilde{r})=f_{\max}$:
\begin{eqnarray}
    \var(\pi_{\tilde{r}}(\vec{x})) &=& \E[(\pi_{\tilde r}(\vec x) - (\theta + \tilde r))^2] - (\beta(\tilde r) - \tilde r)^2 \nonumber \\
    &\leq& \int_0^{\infty} nz^2 \left[f(\tilde{r}+z)+f(\tilde{r}-z)\right] \left[1-F(\tilde{r}+z)+F(\tilde{r}-z)\right]^{n-1} dz \nonumber \\
    &\leq& \int_0^{\epsilon'} nz^2 2(f(\tilde{r})+\delta') \left[1-2(f(\tilde{r})-\delta')z\right]^{n-1} dz \nonumber \\
    &+& \int_{\epsilon'}^{\infty} nz^2 \left[f(\tilde{r}+z)+f(\tilde{r}-z)\right] \left[1-F(\tilde{r}+z)+F(\tilde{r}-z)\right]^{n-1} dz \nonumber \\
    &\leq& \int_0^{\epsilon'} nz^2 2(f(\tilde{r})+\delta') \left[1-2(f(\tilde{r})-\delta')z\right]^{n-1} dz \nonumber \\
    &+& \int_{\epsilon'}^{\infty} nz^2 \left[f(\tilde{r}+z)+f(\tilde{r}-z)\right] \left[1-2\epsilon'(f(\tilde r) - \delta')\right]^{n-1} dz \nonumber \\
    &=& \frac{f_{\max}+\delta'}{f_{\max}-\delta'} \frac{1}{2n^2 (f_{\max} \dmedit{-} \delta')^2} + o(1/n^2)
\end{eqnarray}

For any $r^*_{n_k}$ we can choose $\delta'$ small enough such that $\var(\pi_{r^*_{n_k}}(\vec{x}))>\var(\pi_{\tilde{r}}(\vec{x}))$ for sufficiently large $n_k$. Therefore, $r^*_{n_k}$ is not a most informative offset for some $k$ as it is dominated by a targeting scheme with offset $\tilde{r}$ and lower information loss. This is a contradiction.

\end{proof}

\begin{proof}[Proof of Proposition~\ref{prop:receiver-loss}]
From Theorem~\ref{thm:receiver.br} we see that $\alpha(x) = x + M_R - \beta(\pi)$, where the sender sends signal $x = \pi(\vec{x},y)$ and $\beta(\pi)$ is the bias of the communication scheme $\pi$. Thus, we have receiver's loss  (for any fixed $\theta$) equals

\begin{align*}
\E_{\vec{x},y}[(\pi(\vec{x},y) + M_R - \beta(\pi) - \theta-M_R)^2] & = \E_{\vec{x},y}[(\pi(\vec{x},y) - \theta  -\beta(\pi))^2] \\
    &= \E_{\vec{x},y}[(\pi(\vec{x},y) - \E_{\vec{x},y}[\pi(\vec{x},y)])^2]
\end{align*} 
since by definition of bias $\E_{\vec{x}\sim F_{\theta},y \sim G_{\theta}}[\pi(\vec{x},y)] = \beta(\pi) + \theta$.

\end{proof}

%% file: Sections/appendix_equilibrium.tex
\section{Proofs for Section \ref{sec:equilibrium}}\label{app:equilibrium}

\begin{proof}[Proof of Theorem~\ref{thm:bne}]
Suppose $(\pi,\alpha)$ is a translation invariant PBE.
Then $\alpha$ must be a translation, with shift $\sigma(\alpha)$. 
Then by Theorem~\ref{thm:sender.br} we know that $\pi$ is a targeting scheme with offset $r = M_S - \sigma(\alpha)$.
Moreover, by Theorem~\ref{thm:receiver.br} we must have that $\sigma(\alpha) = M_R - \beta(\pi)$.  Since 
by definition $\beta(r) = \beta(\pi)$, we conclude that 
\[ r = M_S - (M_R - \beta(r)) = \beta(r) + (M_S - M_R) \]
as required.

The other direction follows immediately from Theorems~\ref{thm:receiver.br} and~\ref{thm:sender.br}, because $\pi$ and $\alpha$ are best responses to each other. Note that, in Theorem~\ref{thm:receiver.br}, $\alpha(x)$ maximizes the receiver utility given the belief distribution $D_{\pi,x}$. Hence, we get that $(\pi,\alpha)$ with belief distribution $B(x) = D_{\pi,x}$ is a translation invariant PBE.
\end{proof}

\begin{proof}[Proof of Theorem~\ref{thm:bne.exist}]
Fix $n$ and $M_S-M_R$. 
Let's start with finding a condition that pins down the offset $r$ of a PBE from Theorem~\ref{thm:bne}. 
Given a communication scheme $\pi_r$ with offset $r$, let $z = \pi_r(\vec x,y) - (\theta_S(\vec x, y) + r)$ denote the distance between the target $\theta_S(\vec{x},y) + r$ and the closest anecdote (out of $n$ total anecdotes), where $z$ is positive if the closest anecdote is larger and negative if the closest anecdote is smaller.  Write $h(z; r)$ for the density of $z$ given $r$, over all randomness in $(\theta, \vec{x}, y)$.

We can now calculate the expected bias $\beta(r)$ of a targeting communication scheme with offset $r$:
\begin{equation}\label{eq:offset.bias}
    \beta(r) = r + \int_{-\infty}^{\infty} z h(z; r) dz
\end{equation}

Write $H(r) = \int_{-\infty}^{\infty} z h(z; r) dz$.  Theorem~\ref{thm:bne} now implies that to show that a PBE exists, it suffices to show that there exists a value of $r$ such that
\begin{equation}
\label{eq:bne.condition}
    H(r) = -(M_S - M_R).
\end{equation}

We will show that \eqref{eq:bne.condition} holds, i.e., $H(r)$ is an onto function in its codomain $(-\infty, +\infty)$, in two steps. In the first step, we establish that $H(r)$ tends to $+\infty$ and $-\infty$, respectively as $r\rightarrow -\infty$ and $r\rightarrow +\infty$. Had $H(r)$ been a continuous function, this step would suffice to prove  \eqref{eq:bne.condition}. The second step handles discontinuous $H(r)$ by establishing that at any point that $H(r)$ is discontinuous, the left limit of $H(r)$ is smaller than the right limit. Together these steps show that $H(r)$ can take any value.

For the first step, we will show that $H(r) \to \infty$ as $r \to -\infty$ and $H(r) \to -\infty$ as $r \to \infty$. 
We know $\E_{\epsilon \sim F}[|\epsilon|]$ is bounded by assumption; say $\E[|x-\theta|] < c_0$.  Then 
\begin{align*}
    \E_{\theta,\vec{x},y}[\max_i x_i - \theta_S(\vec{x},y)]
    &= \E_{\theta,\vec{x},y}[\max_i x_i - \theta + \theta - \theta_S(\vec{x},y)] \\
    & = \E_{\theta,\vec{x},y}[\max_i x_i - \theta] + E_{\theta,\vec{x},y}[\theta - \theta_S(\vec{x},y)]\\
    & = \E_{\theta,\vec{x}}[\max_i x_i - \theta] + 0\\
    & \leq \sum_i E_{\theta,x_i}[|x_i - \theta|]\\
    & \leq n c_0
\end{align*}
where the second equality is linearity of expectation and the third equality follows because $\theta_S(\vec{x},y)$ is a valid posterior mean.

Now choose any $Z > 0$ and suppose $r \geq n c_0 + Z$.  Then \begin{align*}
\int_{-\infty}^{\infty} z h(z; r)dz 
&= \E_{\theta,\vec{x},y}[ (\argmin_{x_i \in \vec{x}}{|x_i - r {-\theta_S(\vec{x},y)}|}) - r {-\theta_S(\vec{x},y)}]\\
&\leq \E_{\theta,\vec{x},y}[ \max_i x_i - r {-\theta_S(\vec{x},y)}]\\
&\leq -Z.
\end{align*}
So for any $Z > 0$, we have that $H(r) \leq -Z$ for all sufficiently large $r$, and hence $H(r) \to -\infty$ as $r \to \infty$.  A symmetric argument\footnote{{By taking $r = -nc_0 - Z$ and observing $\E[\min_i - \theta ] \ge -nc_0$.}} shows that $H(r) \to \infty$ as $r \to -\infty$.

In the second step we show that, roughly speaking, if $H(r)$ is discontinuous at some $r$ then the one-sided limits still exist, and the limit from above will be strictly greater than the limit from below.  To see why, suppose $H$ is not continuous at $r_0$.  For each possible realization of $(\theta,\vec{x},y)$, either $z$ is continuous at $r_0$ or it is not.  If not, this means that $\theta_S(\vec{x},y) + r_0$ is precisely halfway between two anecdotes in $\vec{x}$, say with absolute distance $d > 0$ to each, in which case the limit of $z$ from below is $-d$ (distance to the anecdote to the left) and the limit of $z$ from above is $d$ (distance to the anecdote to the right).  Integrating over all realizations, we conclude that the one-sided limits of $H$ exist, $\lim_{r \to r_0^-}H(r) < \lim_{r \to r_0^+}H(r)$, and {moreover $\lim_{r \to r_0^-}H(r) \leq H(r_0) \leq \lim_{r \to r_0^+}H(r)$.}

Now we are ready to prove \eqref{eq:bne.condition}.
Since $H(r) \to -\infty$ as $r \to \infty$, there must exist some finite $r_1$ such that $H(r_1) < -(M_S - M_R)$. Choose $r_2 \leq r_1$ to be the infimum over all $r'$ such that $H(r) \leq -(M_S - M_R)$ for all $r \in (r', r_1]$.  That is, $(r_2, r_1]$ is a maximal (on the left) interval on which $H(r) \leq -(M_S - M_R)$.  Note that $r_2$ must be finite, since $H(r) \to \infty$ as $r \to -\infty$.  

Suppose for contradiction that $H(r_2) \neq -(M_S - M_R)$.  It must then be that $H$ is discontinuous at $r_2$, as otherwise there is an open ball around $r_2$ on which $H$ is either less than or greater than $-(M_S - M_R)$, but either way this contradicts the definition of $r_2$.  

From the definition of $r_2$ we have that $\lim_{r \to r_2^+}H(r) \leq -(M_S-M_R)$.  So since $H$ is discontinuous at $r_2$, we know (from our analysis of the directionality of discontinuities of $H$) that $\lim_{r \to r_2^-}H(r) < -(M_S-M_R)$.  But this then means that there exists some $\epsilon > 0$ such that $H(r) < -(M_S-M_R)$ for all $r \in (r_2-\epsilon, r_2)$, contradicting our choice of $r_2$.

We conclude that $H(r_2) = -(M_S - M_R)$, so $r_2$ is the desired value of $r$ proving \eqref{eq:bne.condition}.  
\end{proof}

\begin{proof}[Proof Proposition~\ref{prop:r_increases_with _Delta}]
Suppose $\Delta = M_R - M_S > 0$ (the other case is analogous).  Recall that $r(\Delta) = \max_r \{r: H(r) = \Delta\} < 0$ is the maximum offset of any equilibrium targeting scheme.
Assume towards contradiction that there exists $\Delta' > \Delta$ such that $r(\Delta') > r(\Delta)$. Recall that, by definition, $r(\Delta') < \underline{r}$. We can therefore let $r^* = \inf_{r \in [r'(\Delta),\underline{r}]} \{ r: H(r) \leq \Delta \}$.  This infimum is well-defined since $H(\underline{r}) = 0 < \Delta$.  

We claim that $H(r^*) = \Delta$.  This is because, from the proof of Theorem~\ref{thm:bne.exist}, we know that one-sided limits of $H$ exist at $r^*$ and $\lim_{r'\to {r^*}^- }H(r)\le H(r^*) \le \lim_{r'\to {r^*}^+ }H(r)$.  But from the definition of $r^*$ we have that (a) $\Delta \leq \lim_{r'\to {r^*}^+ }H(r)$, and (b) $H(r') > \Delta$ for all $r' \in [r(\Delta'), r^*)$. So the only way to satisfy $\lim_{r'\to {r^*}^- }H(r)\le \lim_{r'\to {r^*}^+ }H(r)$ is to have $\lim_{r'\to {r^*}^- }H(r) = \lim_{r'\to {r^*}^+ }H(r) = \Delta$ and hence $H(r^*) = \Delta$.

We have now reached the desired contradiction, since $H(r^*) = \Delta$ and $r^* > r'(\Delta) > r(\Delta)$, contradicting the definition of $r(\Delta)$.
%
\end{proof}

\begin{proof}[Proof of Corollary~\ref{corollary:minmax}] 
Consider $\Delta=M_R-M_S>0$ (the other case is analogous). We then have $r<\beta(r)<\underline{r}$ and $\beta(r)-r=\Delta$. This implies $r<\underline{r}-\Delta$. Hence, $r \rightarrow -\infty$ as $\Delta \rightarrow \infty$.
\end{proof}

\begin{proof}[Proof of Proposition~\ref{prop:utility_decreases_with _Delta}]
Suppose $\Delta<0$ such that $r(\Delta)\geq \overline{r}$. Let $\pi_r(\vec{x}) = \argmin_{x_i \in \vec{x}} |x_i - (\theta + r)|$ be the closest anecdote to the offset $\theta+r$ and let $z = \pi_r(\vec x) - (\theta + r)$ denote the corresponding distance. The bias $\beta(r) \in (0,r)$ then equals $\beta(r) = \E[\pi_r(\vec{x})] - \theta$.  The variance of $\pi_r(\vec{x})$ is given by
\begin{align*}
    \var[\pi_r(\vec{x})] &= \E[\left(\pi_r(\vec{x}) - (\theta + \beta(r))\right)^2]  \\
    &= \E[\left(z - (\beta(r) - r)\right)^2] \\
    &= \E[z^2] - (\beta(r) - r)^2 \qquad (\text{Since } \E[z] = (\beta(r) -r))\\
    &= n\int_{0}^{\infty} z^2 (f(z + r) + f(r - z) )(1 - F(r+ z) + F(r - z))^{n-1} dz - (\beta(r) -r)^2 
\end{align*}

Let $P(r,z) = 1 - F(r + z) + F(r - z)$ be the probability that an anecdote does not lie in the interval $(\theta + r-z,\theta + r+z)$. Hence we can re-write the above equality as,
\begin{align*}
    \var[\pi_r(\vec{x})] & = -\int_0^{\infty} z^2 \left(\frac{\partial P(r,z)^n}{ \partial z}\right) dz - (\beta(r) -r)^2  \\
    & = -z^2H(r,z)^n\vert_{z=0}^\infty + \int_0^{\infty} 2zP(r,z)^n dz - (\beta(r) -r)^2 \\
    & = \int_0^{\infty} 2zP(r,z)^n dz - (\beta(r) -r)^2 \\
\end{align*}
The last equality follows because $z^2P(r,z)^n=0$ at $z=0$ and  we have $\lim_{z\to \infty} z^2\cdot P(r,z) =0$ for $n\ge 2$. To see why, note that, {for any $z > 2|r|$, observe that if $|x-\theta|\leq z/2$ then $|x - ( \theta+ r)| \leq z$}. By Chebyshev's inequality we have
\[ P(r,z) \leq \Pr_{x \sim F}[|x-\theta| \geq z/2] \leq 4\sigma^2/z^2, \] where $\theta$ is the mean of the anecdote distribution and $\sigma^2$ is its (finite) variance.  Hence for $n \geq 2$ we have $\lim_{z \to \infty} z^2 P(r,z)^n = 0$.  We therefore conclude $z^2 P(r,z)^n |^{\infty}_{z=0} = 0$.

This implies that,
\begin{align*}
    \frac{\partial \var[\pi_r(\vec{x})] }{\partial r} &= \int_0^\infty2nz\frac{\partial P(r,z)}{\partial r} P(r,z)^{n-1} dz - 2(\beta(r) -r)(\beta'(r) - 1) \\
    & = -2\E[z] - 2(\beta(r) -r)(\beta'(r) - 1) \\
    &= 2(r - \beta(r) ) + 2(r - \beta(r))(\beta'(r) - 1) \\
    & \ge 0 \qquad (\text{Since } \beta'(r) \ge 0)
\end{align*}
where the equalities follows because the expected value of $z$ is $\beta(r) - r$ and the last inequality follows because the bias increases with $r$. Hence, for $r \geq \overline{r}$, the variance increases with $r$. A similar argument follows for $r<\underline{r}$.

\end{proof}

%% file: Sections/appendix_experts.tex
\section{Proofs for Section \ref{sec:experts}}\label{app:experts}

\begin{proof}[Proof of Proposition~\ref{prop:polarization_n}]

We prove this result by contradiction. Assume that the result does not hold and hence there is some $\Delta\neq 0$, some distribution $F$, and constants $R_L>r_{\min}$ and $R_M<r_{\max}$ and a subsequence $(n_k)$ and associated equilibrium offsets $r^*(n_k)$ such that $R_L\leq r^*(n_k) \leq R_M$. 

Because of continuity there is some $f_{\min},f_{\max}>0$  such that $f_{\min}\leq f(x) \leq f_{\max}$ for any $x\in [R_L,R_M]$. We can then write for any $r\in [R_L,R_M]$:
\begin{eqnarray}
\beta_n(r)-r  &=& \int_0^{\infty} nz \left[f(r+z)-f(r-z)\right] \left[1-F(r+z)+F(r-z)\right]^{n-1} dz 
\end{eqnarray}
Fix some $0<\epsilon'<\frac{1}{2} f_{\min}$ and we get: 
\begin{eqnarray}
 -\int_0^{\epsilon'} nz f_{\max} \left[1-2zf_{\min}\right]^{n-1} dz +o(1/n^2) &\leq& \beta_n(r)-r \nonumber \\
 &\leq& \int_0^{\epsilon'} nz f_{\max} \left[1-2zf_{\min}\right]^{n-1} dz +o(1/n^2)
 \end{eqnarray}
 This allows us to bound $\beta(r)-r$ as follows:
 \begin{equation}
     |\beta_n(r)-r| \leq \frac{f_{\max}}{2f_{\min}} \frac{1}{n}+o(1/n)  
 \end{equation}
Therefore, for sufficiently large $n_k$ we have $\beta_{n_k}(r^*_{n_k})-r^*_{n_k} \neq \Delta$ and therefore $r^*_{n_k}$ cannot be an equilibrium offset.
\end{proof}

\begin{proof}[Proof of Theorem~\ref{thm:heavy_light_tails} (heavy tails)]
The sent anecdote is $y$ away from offset $r>0$ and has density $h(y)$: 
\begin{equation}
    h(y) = nf(r+y)\left[1-\left(F(r+|y|)-F(r-|y|)\right)\right]^{n-1}
\end{equation}
We prove the claim for heavy tails through contradiction: assume that $\var_n(y)\not\rightarrow \infty$. This implies that there exists some $M>0$ and a subsequence $(n_k)$ with associated offsets $r_{n_k}$ such that $\var_{n_k}(y)<M$. 

We next introduce two symmetric and nested intervals around the offset as well as a formula for the probability that $y$ falls into such an interval as well as conditional expectation of $y$ in this case.

\textbf{Growing interval $I_k$.} We define $I_k=[r_{n_k}-b_{k},r_{n_k}+b_{k}]$ where $b_k=r_{n_k}^{\tau}$ for some $0<\tau<1$ (where $\tau$ will be specified later). This is a symmetric interval around the offset $r_{n_k}$ which gradually increases in width with $r_{n_k}$. Importantly, it doesn't increase linearly with $r_{n_k}$ but only at rate $r_{n_k}^{\tau}$ - this ensures that the width of the interval becomes small relative to the offset.

\textbf{Nested interval $\hat{I}_k$.}  We next use Chebyshev's inequality which tells us that for any fixed $\epsilon>0$:
\begin{equation}
    Prob\left(|y+\Delta| \geq \sqrt{\frac{M}{\epsilon}}\right) \leq \epsilon
\end{equation}
This implies that we can ensure that with probability at least $\epsilon$ the closest anecdote is in the interval $\hat{I}_{n_k}=\left[r_{n_k}-\sqrt{\frac{M}{\epsilon}}-|\Delta|,r_{n_k}+\sqrt{\frac{M}{\epsilon}}+|\Delta|\right]$ for all $n_k$. Note that for sufficiently large $k$ we have $\hat{I}_k \subset I_k$.

\textbf{Formula for $P(y \in I)$.} Consider any symmetric interval $I=[r_{n_k}-b,r_{n_k}+b]$ around the offset (such as $I_k$ or $\hat{I}_k$). 

Then the probability $P(y\in I)$ can be expressed as:
\begin{eqnarray}\label{eqn:cdf}
P(y \in I) &=& \int_{-b}^0 n f(r+y) \left[1-\left(F(r+|y|)-F(r-|y|)\right)\right]^{n-1} dy \nonumber \\ 
&+& \int_{0}^{b} n f(r+y) \left[1-\left(F(r+|y|)-F(r-|y|)\right)\right]^{n-1} dy \nonumber \\ 
&=& \int_{0}^{b} n\left[f(r+y)+f(r-y)\right] \left[1-\left(F(r+y)-F(r-y)\right)\right]^{n-1} dy \nonumber \\
&=& 1-\left[1-\left(F(r+b)-F(r-b)\right)\right]^n
\end{eqnarray}

\textbf{Formula for $E(y|y \in I)$.} We can use an analogous derivation to show:
\begin{eqnarray}\label{eqn:expectation}
E(y|y \in I) &=& \int_{0}^{b} y n\left[f(r+y)-f(r-y)\right] \left[1-\left(F(r+y)-F(r-y)\right)\right]^{n-1} dy
\end{eqnarray}

The rest of the proof for heavy tails proceeds in three steps:
\begin{itemize}
    \item[\textbf{Step 1: }] We show that the density $f$ of the anecdote distribution is essentially constant over the interval $I_{k}$. This allows us to simplify equation \ref{eqn:cdf} for any $I=[r_{n_k}-b,r_{n_k}+b] \subset I_k$ as follows: 
    \begin{equation}
       P(y \in I) \approx 1-\left[1-2f(r)b\right]^n
    \end{equation}

   \item[\textbf{Step 2: }] We then show that $P(y\in \hat{I}_k)\geq 1-\epsilon$ implies $P(y \notin I_k)<\epsilon^{D r_{n_k}^{\tau}}$ for some constant $D$.
   
    \item[\textbf{Step 3: }] We can now finally show that $E(y)\rightarrow 0$ which is a contradiction since $E(y)=\Delta$. Hence our initial assumption that $Var_n(y) \not\rightarrow \infty$ was false.
\end{itemize}

\textbf{Step 1.} The density $h(y)$ of the sent anecdote includes the term 
\begin{equation}
    1-\left(F(r+y)-F(r-y)\right)
\end{equation}
for $y>0$. Using the mean-value theorem we get
\begin{equation}\label{eqn:decomp}
    1-\left(F(r+y)-F(r-y\right) = 1-y\left(f(r+\zeta)+f(z-\zeta)\right)
\end{equation}
for some $\zeta \in [0,y]$. Since $f$ has strong heavy tails we know that (for some constant $C>0$):
\begin{eqnarray}
    F(x)&=&1-C\exp(-\int_{\underline{x}}^x u(t) t^{\alpha-1} dt) \nonumber \\
    f(x) &=& C u(x) x^{\alpha-1} \exp(-\int_{\underline{x}}^x u(t) t^{\alpha-1} dt)
\end{eqnarray}
We can now express $f(r+\zeta)$ (and analogously $f(r-\zeta)$) as
\begin{eqnarray}
f(r+\zeta) &=& C u(r+\zeta) (r+\zeta)^{\alpha-1} \exp(-\int_{\underline{x}}^{r+\zeta} u(t) t^{\alpha-1} dt) \nonumber \\
&=& C u(r+\zeta)r^{\alpha-1} (1+\frac{\zeta}{r})^{\alpha-1} \exp(-\int_{\underline{x}}^{r} u(t) t^{\alpha-1} dt) \exp(-\int_{r}^{r+\zeta} u(t) t^{\alpha-1} dt) \nonumber \\
&=& f(r) \frac{u(r+\zeta)}{u(r)} (1+\frac{\zeta}{r})^{\alpha-1} \exp(-\zeta u(r+\zeta')(r+\zeta')^{\alpha-1})
\end{eqnarray}
where $\zeta' \in [0,\zeta]$. Next we use the mean-value theorem again:
\begin{eqnarray}
    \frac{u(r+\zeta)}{u(r)} &=& \exp(\ln(u(r+\zeta))-\ln(u(r)) \nonumber \\
    &=& \exp(\zeta \frac{u'(r+\zeta'')}{u(r+\zeta''))}) \quad \mbox{for some $\zeta'' \in [0,\zeta]$}
\end{eqnarray}
We obtain:
\begin{eqnarray}
f(r+\zeta) &=& f(r) \exp(\zeta \frac{u'(r+\zeta'')}{u(r+\zeta''))}) (1+\frac{\zeta}{r})^{\alpha-1} \exp(-\zeta u(r+\zeta')r^{\alpha-1}(1+\frac{\zeta'}{r})^{\alpha-1})
\end{eqnarray}
We observe:
\begin{eqnarray}
\left|\zeta \frac{u'(r+\zeta'')}{u(r+\zeta''))}\right| &\leq& \left|(r+\zeta'')^\beta \frac{u'(r+\zeta'')}{u(r+\zeta''))}\right| \rightarrow 0 \quad \mbox{for $\tau<\beta$ and suff. large $r$}\nonumber \\
\frac{\zeta}{r} &\rightarrow& 0 \nonumber \\
\frac{\zeta'}{r} &\rightarrow& 0 \nonumber \\
\zeta r^{\alpha-1}&\leq& r^{\tau} r^{\alpha-1} \rightarrow 0 \quad \mbox{for  $0<\tau<1-\alpha$} 
\end{eqnarray}
Therefore, as long as $0<\tau<\min(1-\alpha,\beta)$ we can ensure that for any small $\vartheta>0$ we have for sufficiently large $n_k$: 
\begin{equation}
\frac{f(r+\zeta)}{f(r)} \in [1-\vartheta,1+\vartheta]    
\end{equation}
Similarly,we can derive
\begin{equation}
\frac{f(r-\zeta)}{f(r)} \in [1-\vartheta,1+\vartheta]    
\end{equation}
This allows us to bound equation \ref{eqn:decomp}:
\begin{equation}
    1-2yf(r)(1+\vartheta)< 1-\left(F(r+y)-F(r-y\right) < 1-2yf(r)(1-\vartheta)
\end{equation}
More generally, we can bound formula \ref{eqn:cdf} for $P(y \in I)$ for any nested interval $I=[r_{n_k}-b,r_{n_k}+b]$ in $I_k$ for all $k$ as:
\begin{equation}\label{eqn:bound}
1- \left[1-2bf(r)(1-\vartheta)\right]^n < P(y \in I) <1- \left[1-2bf(r)(1+\vartheta)\right]^n
\end{equation}
Once fixing $\vartheta>0$ this bound holds for all sufficiently large $n_k$ This completes Step 1 of the argument.

\textbf{Step 2.} Recall, that Chebychev's inequality ensured us that the sent anecdote lies in the $\hat{I}_k$ interval with probability of at least $1-\epsilon$:
\begin{eqnarray}
    P(y\in\hat{I}_k) &\geq& 1-\epsilon \nonumber
\end{eqnarray}
By using inequality \ref{eqn:bound} from Step 1 and setting $\hat{b}=\sqrt{\frac{M}{\epsilon}}+|\Delta|$ we obtain:
\begin{eqnarray}\label{bound:epsilon}
    1- \left[1-2\hat{b}f(r_{n_k})(1+\vartheta)\right]^n &>& 1-\epsilon \nonumber \\
    \left[1-2\hat{b}f(r_{n_k})(1+\vartheta)\right]^n &<& \epsilon
\end{eqnarray}
We also can use inequality \ref{eqn:bound} to bound $P(y \notin I_k)$:
\begin{eqnarray}\label{inequ:notinI}
    P(y\notin I_k) &<& \left[1-2r_{n_k}^{\tau} f(r_{n_k})(1-\vartheta)\right]^n
\end{eqnarray}
We now introduce a helper lemma.
\begin{lemma}\label{lemma:expbound}
Assume $A>1$ and $x$ such that $Ax<1$. Then the following holds:
\begin{equation}
    (1-Ax)^n \leq (1-x)^{An}
\end{equation}
\end{lemma}
\begin{description}
    \item[Proof:] The claim is equivalent to:
\begin{equation}
    \ln(1-Ax) \leq A\ln(1-x)
\end{equation}
Consider the function $g(x)=A\ln(1-x)-\ln(1-Ax)$. Note that $g(0)=0$. We also obtain:
\begin{equation}
    g'(x)=-\frac{A}{1-x}+\frac{A}{1-Ax}=\frac{A(A-1)x}{(1-x)(1-Ax)}
\end{equation}
Hence, we have $g'(x)<0$ for $x<0$ and $g'(x)>0$ for $x>0$ which implies $g(x)\geq 0$ and hence proves the claim.
\end{description}
We can use lemma \ref{lemma:expbound} to further bound inequality \ref{inequ:notinI}:
\begin{eqnarray}\label{bound:step2}
    P(y\notin I_k) &<& \left[1-\frac{r_{n_k}^{\tau}(1-\vartheta)}{\hat{b}(1+\vartheta)} 2\hat{b}f(r_{n_k})(1+\vartheta)\right]^n \nonumber \\
    &\leq& \left[\left[1- 2\hat{b}f(r_{n_k})(1+\vartheta)\right]^n\right] ^{\frac{r_{n_k}^{\tau}(1-\vartheta)}{\hat{b}(1+\vartheta)}} \nonumber \\
    &<& \epsilon^{D r_{n_k}^{\tau}}  
\end{eqnarray}
where $D=\frac{1-\vartheta}{\hat{b}(1+\vartheta)}$. This completes Step 2 of the proof.

\textbf{Step 3.} We can decompose $E(y)$ as follows:
\begin{eqnarray}
    E(y) &=& P(y<r_{n_k}-b_{k}) E(y|y<r_{n_k}-b_{k})+P(y\in I_k)E(y|y\in I_k)+P(y>r_{n_k}+b_{k}) E(y|y>r_{n_k}+b_{k}) \nonumber \\
    &\geq& P(y\notin I_k) E(y|y<r_{n_k}-b_k)+P(y\in I_k)E(y|y\in I_k)
\end{eqnarray}
Denote the expectation of a single anecdote conditional on it having a value $x<0$ with $\mu_0<0$. The expectation of the max anecdote conditional on having a realization less than $r_{n_k}-b_k$ has to be at least $\mu_0$ (since $r>0$) and hence the expected distance from $r$ can be at most $r-\mu_0$. We can use the bound \ref{bound:step2} from Step 2 to obtain:
\begin{eqnarray}
    E(y) &\geq& \underbrace{(r_{n_k}-\mu_0)\epsilon^{D r_{n_k}^{\tau}}}_{B_k} +P(y\in I_k)E(y|y\in I_k)
\end{eqnarray}
It is easy to see that $B_k\rightarrow 0$ as $r_{n_k}\rightarrow \infty$. Therefore, we only have to show that $E(y|y\in I_k)\rightarrow 0$ in order to get a contradiction that $E(y)\neq \Delta$.

We can write $E(y|y\in I_k)$ as:
\begin{eqnarray}
    && E(y|y\in I_k) = \int_0^{b_k} y n\left[f(r_{n_k}+y)-f(r_{n_k}-y)\right] \left[1-\left(F(r_{n_k}+y)-F(r_{n_k}-y)\right)\right]^{n-1} dy \nonumber \\
    &=& \int_0^{b_k} \frac{f(r_{n_k}+y)-f(r_{n_k}-y)}{f(r_{n_k}+y)+f(r_{n_k}-y)} \underbrace{y n\left[f(r_{n_k}+y)+f(r_{n_k}-y)\right] \left[1-\left(F(r_{n_k}+y)-F(r_{n_k}-y)\right)\right]^{n-1}}_{w(y)} dy \nonumber
\end{eqnarray}
We know from Step 1:
\begin{equation}
    \frac{-2\vartheta}{2-2\vartheta} \leq  \frac{f(r_{n_k}+y)-f(r_{n_k}-y)}{f(r_{n_k}+y)+f(r_{n_k}-y)} \leq \frac{2\vartheta}{2-2\vartheta}
\end{equation}
Moreover, we can simplify:
\begin{eqnarray}
    \int_0^{b_k} w(y) dy 
    &=& -\int_0^{b_k} y  \left[\left(1-\left(F(r_{n_k}+y)-F(r_{n_k}-y)\right)\right)^n\right]' dy \nonumber \\
    &=& -b_k P(y \notin I_k)+\int_0^{b_k}  \left(1-\left(F(r_{n_k}+y)-F(r_{n_k}-y)\right)\right)^n dy \nonumber \\
    &=& \underbrace{-b_k P(y \notin I_k)}_{\rightarrow 0} + \underbrace{\int_0^{\hat{b}\frac{1+\vartheta}{1-\vartheta}}  \left(1-\left(F(r_{n_k}+y)-F(r_{n_k}-y)\right)\right)^n dy}_{\mbox{between $0$ and $\hat{b}\frac{1+\vartheta}{1-\vartheta}$}}+\nonumber \\
    &+& \underbrace{\int_{\hat{b}\frac{1+\vartheta}{1-\vartheta}}^{b_k}  \left(1-\left(F(r_{n_k}+y)-F(r_{n_k}-y)\right)\right)^n dy}_{W_k}
\end{eqnarray}
For the last term we again use inequality \ref{eqn:bound}:
\begin{eqnarray}
\left(1-\left(F(r_{n_k}+y)-F(r_{n_k}-y)\right)\right)^n &<&\left[1-2bf(r_{n_k})(1-\vartheta)\right]^n \nonumber \\
 W_k &<& \int_{\hat{b}\frac{1+\vartheta}{1-\vartheta}}^{b_k} \left[1-2yf(r_{n_k})(1-\vartheta)\right]^n dy \nonumber \\
 &=& \int_{\hat{b}\frac{1+\vartheta}{1-\vartheta}}^{b_k}  \left[1-\frac{y(1-\vartheta)}{\hat{b}(1+\vartheta)} 2\hat{b}f(r_{n_k})(1+\vartheta)\right]^n dy \nonumber \\
 &\leq& \int_{\hat{b}\frac{1+\vartheta}{1-\vartheta}}^{b_k}  \left[\left[1- 2\hat{b}f(r_{n_k})(1+\vartheta)\right]^n\right]^{\frac{y(1-\vartheta)}{\hat{b}(1+\vartheta)}} dy \quad \mbox{(lemma \ref{lemma:expbound})} \nonumber \\
 &<& \int_{\hat{b}\frac{1+\vartheta}{1-\vartheta}}^{b_k} \epsilon^{\frac{y(1-\vartheta)}{\hat{b}(1+\vartheta)}} dy \quad \mbox{(inequality \ref{bound:epsilon})} \nonumber \\
 &<& -\frac{\hat{b}(1+\vartheta)}{\ln(\epsilon)(1-\vartheta)} \int_{1}^{\infty} \exp(-z) dz \nonumber \\
 &=& -\frac{\hat{b}(1+\vartheta)}{\ln(\epsilon)(1-\vartheta)} 
\end{eqnarray}
Putting everything together we know that:
\begin{equation}
\int_0^{b_k} w(y) dy < \left(1-\frac{1}{\ln(\epsilon)}\right) \frac{\hat{b}(1+\vartheta)}{1-\vartheta}    
\end{equation}
Therefore we get a bound on the conditional expectation:
\begin{equation}
\frac{-2\vartheta}{2-2\vartheta} \left(1-\frac{1}{\ln(\epsilon)}\right) \frac{\hat{b}(1+\vartheta)}{1-\vartheta}   
\leq E(y|y\in I_k)  \leq \frac{2\vartheta}{2-2\vartheta} \left(1-\frac{1}{\ln(\epsilon)}\right) \frac{\hat{b}(1+\vartheta)}{1-\vartheta}
\end{equation}
Since we can choose $\vartheta$ as small as we want (for sufficiently large $(n_k)$) we can deduce that $E(y|y\in I_k) \rightarrow 0$. This completes the argument for Step 3.
\end{proof}

\begin{proof}[Proof of Theorem~\ref{thm:heavy_light_tails} (light tails)]

Recall by Proposition~\ref{prop:utility_decreases_with _Delta} that, for any $n$ the variance of the anecdote sent in the targeting scheme increases with $|r|$. In particular, the variance of the sent anecdote is at most the variance of anecdote $\pi_{\infty}$, the maximum of the anecdotes available to the sender.  It is therefore sufficient to argue that the variance of the maximum anecdote converges to $0$ as $n \to \infty$.

Let us denote the density of the maximum distribution with $h_n(x)$ and the CDF with $H_n(x)=F(x)^n$. The median of this distribution is denoted with $\hat{x}_n$ and is defined as:
\begin{equation}
    H_n(\hat{x}_n)=F(\hat{x}_n)^n=\frac{1}{2}
\end{equation}
Next, we look at the value of $H_n(x)$ for deviations $y$ away from the median:
\begin{equation}
    H_n(\hat{x}_n+y) = \left[1-C\exp\left(-\int_{\underline{x}}^{\hat{x_n}+y} u(t)t^{\alpha-1} dt\right) \right]^n        
\end{equation}
Let's focus on the inner term:
\begin{eqnarray}
    1-C\exp\left(-\int_{\underline{x}}^{\hat{x_n}+y} u(t)t^{\alpha-1} dt\right) &=& 1-C\exp\left(-\int_{\underline{x}}^{\hat{x_n}} u(t)t^{\alpha-1} dt\right) \exp\left(-\int_{\hat{x_n}}^{\hat{x_n}+y} u(t)t^{\alpha-1} dt\right) \nonumber \\
    &=& 1-C\exp\left(-\int_{\underline{x}}^{\hat{x_n}} u(t)t^{\alpha-1} dt\right) \exp\left(-y u(\zeta) \zeta^{\alpha-1}\right)
\end{eqnarray}
for some $\zeta \in [\hat{x_n},\hat{x_n}+y]$ (for $y>0$) and $\zeta \in [\hat{x_n}+y,\hat{x_n}]$ (for $y<0$) by the mean value theorem.

Next, we calculate the moment $\int_{-\infty}^{\infty} y^k h_n(\hat{x_n}+y) dy$ for $k\geq 1$. We will be interested in the cases $k=1$ and $k=2$ which we need to calculate the variance of the maximum distribution.
\begin{eqnarray}
    \int_{-\infty}^{\infty} y^k h_n(\hat{x_n}+y) dy &=& \int_{-\infty}^{0} y^k h_n(\hat{x_n}+y) dy + \int_{0}^{\infty} y^k h_n(\hat{x_n}+y) dy \nonumber \\
    &=& \int_{-\infty}^{0} y^k H'_n(\hat{x_n}+y) dy -\int_{0}^{\infty} y^k (1-H_n(\hat{x_n}+y))' dy \nonumber \\
    &=& \left[y^k H_n(\hat{x_n}+y)\right]_{-\infty}^{0} - \int_{-\infty}^0 ky^{k-1} H_n(\hat{x_n}+y) dy - \left[y^k (1-H_n(\hat{x_n}+y))\right]_0^{\infty} \nonumber \\
    &+& \int_0^{\infty} ky^{k-1}(1-H_n(\hat{x_n}+y)) dy \nonumber \\
    &=& \underbrace{\int_0^{\infty} ky^{k-1}(1-H_n(\hat{x_n}+y)) dy}_{P_n} - \underbrace{\int_{-\infty}^0 ky^{k-1} H_n(\hat{x_n}+y)}_{M_n} dy
\end{eqnarray}
Now we use lemma \ref{lemma:expbound} for $y<0$:
\begin{eqnarray}
    H_n(\tilde{x}_n+y) &=&\left[1-C\exp\left(-\int_{\underline{x}}^{\hat{x_n}} u(t)t^{\alpha-1} dt\right) \exp\left(-y u(\zeta) \zeta^{\alpha-1}\right)\right]^n \nonumber \\
    &\leq&  \left\{\left[1-C\exp\left(-\int_{\underline{x}}^{\hat{x_n}} u(t)t^{\alpha-1} dt\right)\right]^{n}\right\}^{\exp\left(-y u(\zeta) \zeta^{\alpha-1}\right)} \nonumber \\
    &=& \left(\frac{1}{2}\right)^{\exp\left(-y u(\zeta) \zeta^{\alpha-1}\right)}
\end{eqnarray}
Similarly, we can use the analogous version \ref{lemma:expbound} for $y>0$ (which implies $A<1$):
\begin{eqnarray}
    H_n(\tilde{x}_n+y) &\geq& \left(\frac{1}{2}\right)^{\exp\left(-y u(\zeta) \zeta^{\alpha-1}\right)}
\end{eqnarray}
This implies:
\begin{eqnarray}
0 &\leq& P_n \leq \int_0^{\infty} ky^{k-1}\left(1- \left(\frac{1}{2}\right)^{\exp\left(-y u(\zeta) \zeta^{\alpha-1}\right)}\right) dy
 \nonumber \\
0 &\leq& |M_n| \leq \int_{-\infty}^0 |ky^{k-1}| \left(\frac{1}{2}\right)^{\exp\left(-y u(\zeta) \zeta^{\alpha-1}\right)} dy
\end{eqnarray}
It is easy to see that for $\alpha>1$ the right bounds converge to $0$ and therefore $P_n\rightarrow 0$ and $M_n \rightarrow 0$. From this it follows:
\begin{equation}
\lim_{n\rightarrow \infty} E(y^k) = 0
\end{equation}
Therefore, we have $E(y)\rightarrow 0$ and $Var(y) \rightarrow 0$.
\end{proof}

%% file: Sections/appendix_commitment.tex
\section{Proofs for Section~\ref{sec:commitment}}\label{app:commitment}

\begin{proof}[Proof of Theorem~\ref{thm:stackelberg}]
For any committed translation invariant communication scheme $\pi_S$, the receiver's response $\alpha_{\pi_S}$ is a translation with shift $M_R - \beta(\pi_S)$.
Applying  Proposition~\ref{prop:sender_summary} with $\sigma(\alpha_{\pi_S}) = M_R - \beta(\pi_S)$, we have that the sender's expected loss from any given translation invariant communication scheme $\pi_S$ is
\[ \E[(\pi_S(\vec{x},y) - (\theta_S(\vec{x},y)+\beta(\pi_S)))^2] + (M_R - M_S)^2. \]
Therefore, the sender's optimization problem corresponds to choosing $\pi_S$ that minimizes the information loss $\E[(\pi_S(\vec{x},y) - (\theta_S(\vec{x},y)+\beta(\pi_S)))^2]$, which is precisely the variance of the communicated anecdote (by the definition of bias). 
By Theorem~\ref{thm:most_informative} we have that the translation invariant scheme that minimizes the information loss must be a targeting scheme with $r^* = \beta(r^*)$. This immediately implies that the receiver's shift is $\sigma(\alpha_{\pi_S}) = M_R - r^*$. 
\end{proof}

\begin{proof}[Proof of Proposition~\ref{prop:2anecdotes}]
    We first observe that when $n=2$, the sender has no foresight (i.e., $y = \bot$) and the anecdote distribution is symmetric at $0$, the posterior mean of the sender $\theta_S(x_1,x_2, y) = (x_1+x_2)/2$. Therefore, both anecdotes have the same absolute distance to the mean. 
    
    The information loss of the sender under any translation invariant scheme $\pi$ is
    \begin{align*}
        \E[(\pi(\vec x) - \theta_S(\vec x) - \beta(\pi))^2] & =  \E[(\pi(\vec x) - \theta_S(\vec x))^2] - \beta(\pi)^2\\
        &= \E[(x_1-x_2)^2/4] - \beta(\pi)^2,
    \end{align*}
    because no matter whether $\pi$ choses $x_1$ or $x_2$ the distance to the average is the same and doesn't depend on $\pi$. The unbiased communication scheme clearly has more information loss than any biased communication scheme. In fact, the information loss is minimized by the most biased communication scheme, which means both the maximum scheme and the minimum scheme are most informative (by symmetry).
\end{proof}

\begin{proof}[Proof of Proposition~\ref{prop:3uniform}]
Suppose the anecdote distribution is uniform $[0,1]$. Given three signals $x_1\le x_2 \le x_3$, the sender's posterior is a uniform distribution  on $[x_1 - \frac{1}{2},x_3 - \frac{1}{2}]$.
\paragraph{Optimal Unbiased Scheme.}
Consider the optimal unbiased scheme, call it $\pi_0$.  This communication scheme sends the closest signal to $\theta_S(\vec x)$.  Since $\theta_S$ is the midpoint of the interval $[x_1,x_3]$, and since $x_2$ falls in that interval, the optimal unbiased scheme always sends signal $x_2$.

Let's calculate the mean squared error of signal $x_2$ relative to $\theta$.  The CDF of $x_2$ is given by
\[ H(w) = \Pr[x_2 < w] = w^3 + 3w^2(1-w) \]
for $w \in [0,1]$, since the first term is the probability that all three samples are less than $w$, and the second term is the probability that two of the three samples are less than $w$.  Now write $d = |\theta - x_2| = |1/2 - x_2|$, noting that $d$ is a random variable.  Then $1$ minus the CDF of $d$ is given by
\[ \tilde{H}(z) = \Pr[ d > z ] = \Pr_{w \sim H}[ w < (1/2 - z) ] + \Pr_{w \sim H}[ w > (1/2 + z) ] = 2 \Pr_{w \sim H}[ w < (1/2 - z) ] = 2H(1/2 - z) \]
for $z < 1/2$, and $\tilde{H}(z) = 0$ for $z \geq 1/2$.  Here we used that $\Pr_{w \sim H}[ w > (1/2 + z) ] = \Pr_{w \sim H}[ w < (1/2 - z) ]$ by symmetry.

The total loss of communication scheme $\pi_0$ is therefore
\begin{align*}
\E[d^2] = & \int_0^\infty \Pr[ d^2 > z ] dz \\
& \int_0^\infty \Pr[ d > \sqrt{z} ] dz \\
& \int_0^\infty 2 H(1/2 - \sqrt{z}) dz \\
& = 1/20
\end{align*}
where the final equality is via numerical calculation.

\paragraph{A Better Biased Communication Scheme.}
We will now build a communication scheme with strictly less loss than $\pi_0$.  Write $\pi_{r}$ for the targeting scheme with offset $r$, which by definition returns whichever of the three points is closest to $\theta_S(\vec x) + r$.  We will eventually choose $r = 1/5$, but for now we'll proceed with general $r$.

Which point does $\pi_{r}$ return?  Write $x^*$ for the random variable representing the point that $\pi_{r}$ returns.  Recall that $\theta_S(\vec x) = (x_1 + x_3)/2$, so $\theta_S(\vec x) + r$ is always closer to $x_3$ than $x_1$.  The distance to point $x_3$ is $|\theta_S(\vec x) + r - x_3| = (x_3 - x_1)/2 + r$, and the distance to point $x_2$ is $|\theta_S(\vec x) + r - x_2| = (x_3 + x_1)/2 + r - x_2$.  So the point $x_2$ will be closest precisely if $x_2 > x_1 + 2r$.  To summarize: $x^* = x_2$ if $x_2 > x_1 + 2r$, otherwise $x^* = x_3$.

As before, let's work out the CDF for $x^*$.  What is the probability that $x^* < w$ for some fixed value of $w \in [0,1]$?  If all three points are less than $w$ (which happens with probability $w^3$) then $x^*$ certainly is.  On the other hand, if $x_2 > w$, then certainly $x^* > w$ as well.  If $x_2 < w$ and $x_3 > w$ (which happens with probability $3w^2(1-w)$), then $x^* < w$ only if $x^* = x_2$, which occurs if and only if $x_2 > x_1 + 2r$.  The conditional probability of that last event is equivalent to the probability that two random variables, each drawn uniformly from $[0,w]$, are at least distance $2r$ apart from each other.  So, we can write the CDF as
\begin{align*}
H[w] & = \Pr[x^* < w] \\
&= w^3 + 3w^2(1-w) \Pr[|x_1 - x_2| > 2r \ |\ x_2 < w] \\
&= w^3 + 3w^2(1-w) \cdot 2 \cdot \int_0^{w-2r} \frac{1}{w} \cdot \frac{w-(x+2r)}{w} dx.
\end{align*}
To justify the last equality, consider drawing one point uniformly from $[0,w]$, so with uniform density $\frac{1}{w}$.  What is the probability that a second drawn point is at least $2r$ larger?  If the first point (call it $x$) is greater than $w - 2r$ the probability is $0$.  Otherwise it is $\frac{w-(x+2r)}{w}$.  Integrating over $x$ gives the probability of this event.  We then double that probability to account for the possibility that the first point drawn is the larger one.

Now write $d = |x^* - r - \theta| = |x^* - (1/2 + r)|$.  This will be the distance between the receiver's action and $\theta$ (where recall we fixed $\theta = 1/2$), if the receiver shifts the received signal $x^*$ by $r$.  Note that this may not be the optimal action of the receiver, but the optimal action performs at least as well as $\E[d^2]$.

Now $1$ minus the CDF of $d$ is given by
\begin{align*}
\tilde{H}(z) = \Pr[d > z] =
\begin{cases}
	H(1/2 + r - z) + 1 - H(1/2 + r + z) \quad & \text{if $0 < z < \frac{1}{2}-r$,}\\
	H(1/2 + r - z)	\quad & \text{if $\frac{1}{2}-r < z < \frac{1}{2}+r$,}\\
	0 \quad & \text{if $z > \frac{1}{2}+r$.}\\
\end{cases}
\end{align*}
Note that unlike the case of $\pi_0$, the fact that $r > 0$ breaks symmetry in the calculation of $\tilde{H}$.  But the reasoning is the same: $d > z$ precisely if either $x^*$ is greater than $1/2 + r + z$ or $x^*$ is less than $1/2 + r - z$.

Finally, as before, the total loss of the communication scheme $\pi_r$ is
\begin{align*}
\E[d^2] = & \int_0^\infty \Pr[ d^2 > z ] dz \\
& \int_0^\infty \Pr[ d > \sqrt{z} ] dz \\
& \int_0^\infty \tilde{H}(\sqrt{z}) dz
\end{align*}

For $r = 1/5$, this integral evaluates to approximately $0.036$, which is less than $1/20$.

\paragraph{Intuition and Discussion.}
Why is $\pi_r$ better than $\pi_0$?  In this case, $\theta_S(\vec x) = (x_1 + x_3)/2$, so $\theta_S(\vec x)$ is highly correlated with $x_1$ and $x_3$ and much less correlated with $x_2$.  This fact is specific to the uniform distribution.  By selecting the point closest to $\theta_S(\vec x) + 1/5$, we are trading off probability of returning $x_2$ with probability of returning $x_3$.  Because of the improved correlation with $x_3$, the location of $x_3$ is more highly concentrated, given $\theta_S(\vec x)$, than the location of $x_2$.  So by targeting an ``expected'' location of $x_3$ relative to $\theta_S(\vec x)$ (in this case, $\theta_S(\vec x) + 1/5$), we can reduce the variance of the distance to the closest point.
\end{proof}

%% file: Sections/online_appendix.tex
\MakeTitle{
        Online Appendix: Communicating with Anecdotes
    }{
        Nika  Haghtalab\\ UC Berkeley \and Nicole Immorlica \\Microsoft Research \and Brendan Lucier \\Microsoft Research\and Markus Mobius \\Microsoft Research \and Divyarthi Mohan \\Tel Aviv University
    }{
        April 2024
    }
\setcounter{page}{1}

\setcounter{section}{0}
\renewcommand{\thesection}{S-\arabic{section}}%

\section{Single-peaked and Symmetric Anecdote Distributions}

In this online appendix we examine the intuition that the information loss minimizing communication scheme will target close to the global maximum density by considering the symmetric and strictly single-peaked anecdote distribution around $\theta$. Indeed, under foresight, by Lemma~\ref{lem:symmetric} we have that the sender optimal targeting scheme is the mean scheme.

In general, this intuition suggests that the sender should want to send signals that are close to the posterior mean $\theta_S$, and hence she will want to select a communication scheme with low bias. Theorem~\ref{thm:unbiased_optimality} below formalizes this insight for the \emph{no-foresight} setting by showing that unbiased communication is optimal for the sender as long as (a) the single-peaked distribution is ``well-behaved'' and (b) the sender has access to sufficiently many signals. 

We start by defining a well-behaved anecdote distribution. 
\wellbehaved*

For example, the normal distribution $F \sim N(0,1)$ and the Laplace distribution with density $f(\epsilon)=\frac{1}{2}\exp(-|\epsilon|)$ are well-behaved.  Recall that $f(\epsilon)$ is the density of the signal distribution at $\theta + \epsilon$, an offset of $\epsilon$ from the true state of the world.

\optimalCommitment*

This result confirms our intuition that the optimal communication scheme sends the signal closest to the posterior mean. Intuitively, the sender would like to send \textit{precisely} the posterior mean to the receiver. However, since she can only send a signal, she has to contend with the second best which is to send the signal closest to the posterior mean.
When we interpret our model as a model of memory where the current self communicates with her future self by storing a single anecdote in memory, we can think of the anecdote closest to the posterior mean as the ``most representative anecdote''.

\input{Sections/appendix_commitment_no_foresight}

%% file: Sections/appendix_commitment_no_foresight.tex
\section{Overview of Theorem \ref{thm:unbiased_optimality}}
\label{app:commitment_noforesight}

For the proof of Theorem \ref{thm:unbiased_optimality} we bound the losses from the biased and unbiased communication schemes and show that the unbiased communication scheme dominates. 

\begin{proposition}\label{prop:unbiased-biased-bounds}
Given any well-behaved anecdote distribution $F$, the unbiased targeting communication scheme with $r=\beta(r)=0$ which selects the closest signal to the sender's posterior mean $\theta_S(\vec{x},y)$ has {signaling} loss:
\begin{equation}
\frac{1}{2n^2 f(0)^2} + o\left(\frac{1}{n^2}\right) 
\end{equation}
In contrast, the biased communication scheme with bias $\delta$ has {signaling loss}:
\begin{equation}
\frac{1}{2n^2 f(\delta(\pi))^2} + o\left(\frac{1}{n^2}\right)
\end{equation}
\end{proposition}
These two bounds together imply that the unbiased communication scheme is asymptotically optimal, and the optimal communication scheme is asymptotically unbiased.


Let $X_\delta = \min_i |x_i - \theta_S(\vec{x},y) - \delta|$ denote the absolute distance of the closest signal to the shift of the posterior mean, $\delta + \theta_S(\vec{x},y)$. We observe that the signaling loss, $\E_{\theta,\vec{x},y}[(\pi(\vec{x},y) - \theta_S(\vec{x},y) - \beta(\pi))]$, of any translation invariant communication scheme with bias $\beta(\pi) = \delta$ is at least $\E_{\theta,\vec{x},y}[X_\delta^2]$. 

\paragraph{Optimal unbiased communication scheme.} Since the bias of the communication scheme that sends signal closest to the posterior mean is itself $0$, this is the optimal amongst all unbiased communication schemes.

When the sender does not have foresight, the posterior mean $\theta_S(\vec x , y)$ depends on the realized signals $\vec x$, and this introduces correlation between the signal realizations and the value of $\theta_S(\vec x ,y) + \delta$.  We therefore cannot model $X_\delta$ using independent draws from the signal distribution.  Indeed, as we seen in Section~\ref{subsec:commitment:discussion}, these correlations can significantly impact $\E[X_\delta^2]$ when the number of signals is small.

Our approach is to argue that as $n$ grows large, the impact of these correlations grows small.  Small enough, in fact, that the correlation between $\theta_S(\vec x,y) + \delta$ and the signal closest to that point becomes small enough that it is dominated by the statistical noise that would anyway be present if signals were drawn independently of $\theta_S(\vec x, y)$.  We argue this in three steps.

\begin{itemize}
    \item [\textbf{Step 1:}] We argue that it suffices to focus on cases where $\theta_S(\vec x, y)$ falls within a narrow interval.  Let $I=[-n^{-\frac{1}{2}+\varepsilon},n^{-\frac{1}{2}+\varepsilon}]$ for some $\varepsilon >0$.  Using the law of large numbers, we argue that $\theta_S(\vec x,y) \in I$ with all but exponentially small probability (in $n$).  The contribution to $\E[X_0^2]$ from events where $\theta_S(\vec x,y) \not\in I$ is therefore negligible and can be safely ignored.  This allows us to assume that $\theta_S(\vec x,y) \in I$.
    
    \item [\textbf{Step 2:}] To reduce the impact of correlation we won't focus on the exact value of $\theta_S(\vec x,y)$, but rather an interval in which it falls.  To this end we partition $I$ into subintervals of width $n^{-b}$, where $b$ is chosen so that any given interval is unlikely to contain a signal.  One such subinterval contains the posterior mean $\theta_S(\vec x,y)$; call that subinterval $C$.  We then consider longer subintervals $L$ and $R$ to the left and right of $C$, respectively, of width $n^{-a}$ chosen large enough that we expect many signals to appear in each\footnote{For $\delta\neq 0$, we consider $L$ and $R$ to the left and right of $C + \delta$.}.  See Figure~\ref{fig:proof_intuition}. 

We bound the impact of correlation by showing that if we condition on the number of signals that appear in $L$ and $R$, then the actual arrangement of signals within those subintervals (keeping all other signals fixed) has only negligible effect on the posterior mean.  Specifically, given any arrangement of the signals within $L$ and $R$, the probability that the posterior mean falls within $C$ remains large.  (See Corollary~\ref{cor:posterior-change} for more details.) 

This implies that there is negligible correlation between the joint density function of a fixed number of $k$ signals in $L \cup R$ and the event that $\theta_S(\vec x,y) \in C$. 
We formally show this in Lemmas~\ref{lem:posterior-correlation-upper} and~\ref{lem:posterior-correlation-lower}~\footnote{For these lemmas we assume that the density function $\theta_S$ is sufficiently ``nice'' in $C$. Refer to Section~\ref{app:rare-events} for details about this assumption, and why we can make this assumption without loss of generality.}. 

\item[\textbf{Step 3:}] The analysis in Step 2 is conditional on the number of signals $k$ that fall in $L \cup R$. We now show a concentration result on the distribution of $k$: with high probability, the number of signals that lie in $L \cup R$ is close to the expected number of signals in the interval $L \cup R$ without any correlation to the event $\theta_S(\vec x, y) \in C$.  See Section~\ref{app:observable-asymptotic} for the proof, and Lemma~\ref{lem:few-near-signals} for the proof that it suffices to consider only this high-probability event.

Given this concentration result, we can focus on bounding the expected value of $X_\delta^2$, the squared distance of the signal closest to $\theta_S(\vec x , y) + \delta$, given the numbers of signals in $L$ and $R$.  From the analysis in Step 2, we can view these signals as (approximately) independently distributed within $L$ and $R$.  We can therefore bound the expected squared distance between interval $C + \delta$ and the closest signal to interval $C + \delta$ by performing an explicit calculation for independent signals.  We still do not know the value of $\theta_S(\vec x, y)$ within interval $C$ (and we have not bounded the impact of correlation on that value), but $C$ is sufficiently narrow that this uncertainty has limited impact on $\E[X_\delta^2]$.  We conclude that the impact of correlation on $\E[X_\delta^2]$ is absorbed in lower-order terms.  This gives us the required results of Proposition~\ref{prop:unbiased-biased-bounds}.
\end{itemize}

\begin{figure}
\caption{Intuition for proof of theorem \ref{thm:unbiased_optimality}}
\begin{center}
\begin{tikzpicture}[line cap=round,line join=round]
\fill[fill=gray!30] (6,0) -- (6,5) -- (8,5) -- (8,0);
\draw[->,color=black] (0,0) -- (14,0);
\draw[color=black] (14,0) node [anchor=south west] {$\theta$};
\draw[dotted] (6,0) -- (6,5);
\draw[dotted] (8,0) -- (8,5);
\draw[dashed] (1,0) -- (1,5);
\draw[dashed] (13,0) -- (13,5);
\draw [thick, red,decorate,decoration={brace,amplitude=5,mirror}](6,0) -- (8,0) node[black,midway,yshift=-0.6cm] {\begin{minipage}{2cm}\begin{center}$C$\\ \mbox{width: $n^{-b}$}\end{center}\end{minipage}};
\draw [thick, green,decorate,decoration={brace,amplitude=5,mirror}](1,0) -- (6,0) node[black,midway,yshift=-0.6cm] {\begin{minipage}{2cm}\begin{center}$L$\\ \mbox{width: $n^{-a}$}\end{center}\end{minipage}};
\draw [thick, green,decorate,decoration={brace,amplitude=5,mirror}](8,0) -- (13,0) node[black,midway,yshift=-0.6cm] {\begin{minipage}{2cm}\begin{center}$R$\\ \mbox{width: $n^{-a}$}\end{center}\end{minipage}};
\draw [thick, red,decorate,decoration={brace,amplitude=5}](4.5,0) -- (6,0) node[black,midway,yshift=+0.6cm] {\begin{minipage}{2cm}\begin{center}$\frac{\Delta}{n}$\end{center}\end{minipage}};

\draw[color=black] (3.5,2.5) node {$\approx n^{1-a}$ signals};
\draw[color=black] (10.5,2.5) node {$\approx n^{1-a}$ signals};
\end{tikzpicture}
\label{fig:proof_intuition}
\end{center}
\footnotesize We assume $a<1$ and $1<b<2a-\frac{1}{2}$ which ensures that the intervals $L$ and $R$ contain many signals but the collective influence of these signals on the posterior mean is $O(n^{\frac{1}{2}-2a})$ and hence smaller than $\frac{1}{n}$. It also implies that if we consider a posterior that is contained in $C$ then a rearrangement of signals in $L$ or $R$ will keep the posterior mean within $C$ with high likelihood \textit{and} the probability of signals drawn from the interval $C$ goes to $0$. For example, $a=\frac{4}{5}$ and $b=\frac{12}{11}$ satisfy these conditions.
\end{figure}

\section{Proof of Proposition~\ref{prop:unbiased-biased-bounds}}\label{app:observable-asymptotic}

We start by proving that the that the signaling loss of the unbiased communication scheme that sends the anecdote closest to the posterior mean $\theta_S(\vec x, y)$ is at most  $\frac{1}{2n^2f(0)} + o(1/n^2)$.  Later in Section~\ref{app:biased-loss} we bound the signaling loss of a biased communication scheme.

Let $I =  [-n^{-\frac{1}{2}+\varepsilon}, n^{-\frac{1}{2}+\varepsilon}] + \theta$. For the remainder of this section, we fix $\theta =0$ for brevity, but everything holds for any fixed $\theta$.
Let $\mathcal{P}$ be a partition of $I$ into intervals of length $n^{-b}$. For any $C\in P$, we define $N(C) = L\cup R$, where $L$ (resp. $R$) is the neighboring interval of length $n^{-a}$ to the left of $C$ (resp. to the right of $C$).

We first consider the ``high probability event'' that the following \emph{desirable} properties hold:

\begin{enumerate}
\item $\theta_S(\vec x, y) \in I$ and let $C\in \mathcal{P}$ be the interval with $\theta_S(\vec x, y)$, 
\item $C$ is not \emph{weak} (see Definition~\ref{def:weak-posterior}), and
\item there are \emph{sufficiently} many signals in $N(C) = L\cup R$. 
\end{enumerate}

In Section~\ref{app:high-prob-event}, we bound the signaling loss contributed by this high probability event. Further, in Section~\ref{app:rare-events} we bound the loss from the ``rare event" that some desirable property does not hold: we bound the loss from the event when $\theta_S \notin I$ (in Lemma~\ref{lem:thetaS-notin-I}), when $C$ is a weak interval (in Lemma~\ref{lem:weak-posterior-density}), or when there are very few signals in $N(C)$ (in Lemma~\ref{lem:few-near-signals}). 

With this we are ready to bound the signaling loss.
Recall that the signaling loss of the unbiased communication scheme $\pi(\cdot)$ that sends anecdote closest to $\theta_S(\vec x,y)$ is $\E[X^2_0] = \E[(\pi(\vec x) - \theta_S(\vec x,y))^2]$. Given $\theta_S \in C$, let $K_0$ be the event that there are \emph{sufficiently} many signals in $N(C)$.  We see that,

\begin{eqnarray}
\E[X^2_0] & = &\E[X^2_0\cdot\one\{\theta_S\notin I\}]   + \sum_{C\in W_I}\E[X^2_0\cdot\one\{\theta_S\in C\}] + \sum_{C\in \mathcal{P}\setminus W_I}\E[X^2_0\cdot\one\{\theta_S\in C\}]\notag \\
& = & {\sum_{C\in \mathcal{P}\setminus W_I}\E[X^2_0\cdot\one\{\theta_S\in C, K_0\}]}  \label{eq:commit-cost-unbiased}\\
&  & +\underbrace{\E[X^2_0\cdot\one\{\theta_S\notin I\}] +   \sum_{C\in W_I}\E[X^2_0\cdot\one\{\theta_S\in C\}]  +\sum_{C\in \mathcal{P}\setminus W_I}\E[X^2_0\cdot\one\{\theta_S\in C,\overline{K_0}\}] }_{\mbox{\text{rare events}} }
\label{eq:rare-events}\\
& \le & \frac{1}{2n^2f(0)^2} + o(1/n^2) \notag
\end{eqnarray}

This is because by Lemma~\ref{lem:bound-commit-cost-term}  the term in Eq.~(\ref{eq:commit-cost-unbiased}) is $\sum_{C\in \mathcal{P}\setminus W_I}\E[X^2_0\cdot\one\{\theta_S\in C, K_0\}]  \le \frac{1}{2n^2f(0)^2} + o(1/n^2) $, and from Lemmas~\ref{lem:thetaS-notin-I},~\ref{lem:weak-posterior-density}, and~\ref{lem:few-near-signals} we see that all terms in Eq.~(\ref{eq:rare-events}) contribute at most $o(1/n^2)$. Thus, giving us the required bound on the signaling loss.

\subsection{Contribution of the High Probability Event}\label{app:high-prob-event}

In this section we explain what the desirable properties are, why they are useful, and bound the signaling loss contributed by the event that these properties hold.

\paragraph{Property (1):} $\theta_S(\vec x,y) \in I =  [-n^{-\frac{1}{2}+\epsilon},n^{-\frac{1}{2}+\epsilon} ]$.

\paragraph{Property (2):} Let $C \in \mathcal{P}$ be the interval with $\theta_S(\vec x, y)$. We need $C$ to be not \emph{weak}. 

We start with the definition of a \emph{weak interval}.

\begin{definition}\label{def:weak-posterior}
   Let $\tau(\cdot)$ be the pdf of the posterior mean $\theta_S(\vec{x},y)$. We say that an interval $C \in \mathcal{P}$ is \emph{weak} if $\tau(\overline{\theta}) \le c_2 n^{1+1/22}e^{-n^{1/22\alpha}} $ for all $\overline{\theta} \in C$. Let $W_I \subset \mathcal{P}$ be the set of all such intervals $C$.
\end{definition}

By Claim~\ref{claim:weak-interval-probability} we will see that the probability that $\theta_S(\vec x,y) \in C$ for some \emph{weak interval} $C$ is negligible $O(n^{-4log n +1})$. Moreover, if $C$ is not \emph{weak}, that is, $\tau(\overline{\theta}) \ge c' n^{1/22} n^{-4\log n + 1}$ then we get that $\tau(\overline{\theta}') = \tau(\overline{\theta})(1+ O(n^{-\frac{1}{22}}))$  for all $\overline{\theta}, \overline{\theta}'\in C$ by Claim~\ref{claim:uniform-posterior}.

\paragraph{Property (3):} Next we show that there are \emph{sufficiently} many signals in $L$ and $R$. 
We start by proving the following claim that $f(x) = f(0) (1 + O(1/\sqrt{n}))$ for all $x \in I = [-n^{-\frac{1}{2}+\epsilon},n^{-\frac{1}{2}+\epsilon} ]$.

\begin{claim}\label{claim:uniform-density}
Given any well-behaved distribution with pdf $f$,  for all $x\in I$, we have $f(x)= f(0) (1 + O(1/\sqrt{n}))$.
\end{claim}

\begin{proof}
Without loss of generality, we assume that $x > 0$, since $f(-x) = f(x)$. By mean value theorem we see that $f(x)= f(0) + x f'(\tilde{x})$ for some $\tilde{x} \in [0, x]$.
By our assumption on $g'$ we get $|g(x)| \le c x^m$ for some constants $c>0$ and $m\ge 0$. This implies $|f'(x)| \le c x^m f(x)$ for all $x > 0$. Since $f$ is non-increasing in $(0,\infty)$ we see that $f(x)\le f(\tilde{x}) \le f(0)$. Moreover, $f'(\tilde{x})\le 0$, so we have $f'(\tilde{x}) \ge -c x^mf(x)$. By mean value theorem we have,

\begin{align*}
    f(x) &= f(0) + x f'(\tilde{x}) \\
    & \ge f(0) - x c \tilde{x}^m f(\tilde{x}) \qquad (\text{Since }f'(\tilde{x})/f(\tilde{x})\ge -c\tilde{x}^{m} )\\
    & \ge f(0) ( 1 - x c \tilde{x}^m) \qquad ( \text{Since } f(0)\ge f(\tilde{x}))\\
    & \ge f(0) (1 - c x^{m+1}) \qquad ( \text{Since } \tilde{x} \le x)
\end{align*}

Therefore, for all $x\in I$ we have $f(x) \ge  f(0)(1 - c(n^{(-\frac{1}{2} + \epsilon)(m+1)}) )$. Note that $m \ge 0$, hence we get $f(x)= f(0)(1 - O(n^{-1/2 + \epsilon}))$.
\end{proof}

Using the above claim that $f(x)$ is approximately $f(0)$ for $x\in I$, we bound the number of signals in a subset $A \subset I$.

\begin{claim}\label{claim:k-signals-chernoff}
Given any interval $A\subset I$ of length $\ell$, the expected number signals in $A$ is ${n\ell}f(0)(1 - O(n^{-1}))$. Let $Y(A)$ be the number of signals in $A$. For any $0<\varepsilon<1$, we have $$\Pr[ Y(A) \le (1-\varepsilon)\E[Y(A)] ] \le \exp\left(-\frac{\varepsilon^2 f(0)n\ell}{2}\right).$$
\end{claim}

\begin{proof}

Let $Y_i = 1$ if $x_i \in A$ and $0$ otherwise. So we have, $\sum_{i=1}^n Y_i = Y(A)$. By Claim~\ref{claim:uniform-density} we have $f(x) = f(0)(1 + O(1/\sqrt{n}))$ for all $x\in I$. Therefore, we have $\Pr[x_i \in A] = \int_A f(x) dx =f(0)(1 + O(1/\sqrt{n}))\int_A dx = f(0)(1 - O(1/\sqrt{n}))\ell $. Note that, $Y_i$ are i.i.d.~random variables, and $\E[Y(A)] =f(0)(1 - O(1/\sqrt{n}))n\ell$. By using Chernoff bound we get
$$\Pr[ Y(A) \le (1-\varepsilon)\E[Y(A)] ] \le \exp\left(-\frac{\varepsilon^2 f(0)n\ell(1 -O(n^{-1/2}))}{2}\right).$$
\end{proof}

We partition the interval $I$ into intervals $J$ of length $n^{-a}/M$ for {$M = n^{1/22}$}. Let $\mathcal{J}$ denote the partition. Note that, the size of $\mathcal{J}$  is $n^{-1/2 + \epsilon + a} M$.

\begin{lemma}\label{lem:many-signals-J}
Let $k_m =  f(0)n^{1-a}/M $ For each $J \in \mathcal{J}$. Let $Y(J)$ be the number signals in interval $J$ (of length $n^{-a}/M$). $\Pr[ |Y(J) - k_m|\ge \varepsilon_1 k_m] \le \exp\left(-\frac{\varepsilon_1^2 f(0)n^{1-a}}{3M}\right).$

Moreover, the probability that there is a $J \in \mathcal{J}$ with $|Y(J) - k_m| \ge n^{-1/20} k_m$ is at most $O(\exp^{-n^{1/22}} )$
\end{lemma}

\begin{proof}
By directly invoking Claim~\ref{claim:k-signals-chernoff} on $J$ of length $n^{-a}/M$ we get $\Pr[ |Y(J) - k_m|\ge \varepsilon_1 k_m] \le \exp\left(-\frac{\varepsilon_1^2 f(0)n^{1-a}}{3M}\right).$
Note that, the size of $\mathcal{J}$  is $n^{-1/2 + \epsilon + a} M$. Therefore, by union bound, we get $\Pr[\exists J \in \mathcal{J} : |Y(J) - k_m|\ge \varepsilon_1 k_m] \le (n^{-1/2 + \epsilon + a} M) \cdot \exp\left(-\frac{\varepsilon_1^2 f(0)n^{1-a}}{3M}\right)$ which is at most $O(\exp^{-n^{1/20}} )$, for $\varepsilon_1 = O(n^{-1/22})$ and $M = n^{1/100}$.
\end{proof}

We will now only focus on the case where \emph{all}  $J \in \mathcal{J}$ has sufficiently many signals, which immediately implies the following Corollary~\ref{cor:k-signals-in-L}. In Lemma~\ref{lem:few-near-signals} we bound the loss of the rare event that this is not the case.

\begin{corollary}\label{cor:k-signals-in-L}
Let $k_0 =  f (0) n^{1 - a}(1-2/M)$, and $\varepsilon = O(n^{-1/22})$. For all $C \in \mathcal{P}$, let $N(C)=L\cup R$ (of size $n^{-a}$). Let $Y(C)$ be the number signals in interval $N(C)$. If all $J \in \mathcal{J}$ have $k_m(1\pm \varepsilon)$ signals, then  $Y(C)$ has  $k_0 (1 \pm \varepsilon)$ signals.

 That is,  $\Pr[ |Y(C) - k_0| \ge \varepsilon k_0]  \le O(\exp^{-n^{1/20}} )$.

\end{corollary}

\begin{proof}
Note that any $N(C)\subset I$ of length $n^{-a}$ contains at least $M-2$ many intervals $J\in \mathcal{J}$. By Lemma~\ref{lem:many-signals-J}, we have that all $J\in \mathcal{J}$ has at least $(1-\varepsilon) f (0) n^{1 - a}$ many signals with high probability. Therefore, $N(C)$ contains at least $k_0$ many signals with probability  $1 - O(\exp^{-n^{1/20}} )$. Note this is regardless of which $C \in \mathcal{P}$ we are considering.
\end{proof}

We will now only consider the event where all the desirable properties hold. For each $C\in \mathcal{P}\setminus W_I$, let $K_0$ denote the event that there are {$k_0(1\pm \varepsilon_1)$} signals in $N(C)$.

\begin{lemma}\label{lem:exponential-upper}
Fix any $C\in \mathcal{P}\setminus W_I$. The distribution of the random variable $X_0 \cdot \one(K_0)$ conditioned on $\theta_S\in C$ is stochastically dominated by the exponential distribution with {$\lambda= 2nf(0)$}. That is, \[\Pr[X_{(0)}\cdot\one(K_0)>d|\theta_S\in C] < \exp[- \lambda d].\]
\end{lemma}

\begin{proof}
For $A\subseteq [n]$ let $K_A$ denote the event that $x_i \in N(C)$ iff $i\in A$. We also use $x_A$ to denote $\{x_i\}_{i\in A}$. For all $d>n^{-a}+n^{-b}$ we have $\Pr[X_{(0)}\cdot\one(K_0)>d|\theta_S\in C] =0$. For all $d<n^{-a}/2+n^{-b}$, let $B_d$ be the interval of length $2d$ centered around $\theta_S$.

We will use the following results/facts:

\begin{enumerate}
    \item $\Pr[ x_A \notin B_d | K_A, \theta_S \in C] =\Pr[x_A \notin B_d | K_A ]\cdot \frac{\Pr[\theta_S \in C | x_A \notin B_d ,K_A]}{\Pr[\theta_S\in C| K_A] }$ (by Bayes rule)
    \item Since $f$ is near uniform in $N(C)$ we have $\Pr[x_A \notin B_d | K_A ] \approx (1-\frac{d}{n^{-a}})^{|A|}$ (by Claim~\ref{claim:uniform-density})
    \item Since redrawing $x_A \in N(C)$ doesn't change $\theta_S$ much, we have $\Pr[\theta_S \in C | x_A \notin B_d ,K_A] \le \Pr[\theta_S \in C' | K_A]$ where $C' = C \pm |A|n^{-1-a}$ (by Lemma~\ref{lem:posterior-correlation-upper})
    \item $\Pr[K_A |\theta_S \in C]\cdot \frac{\Pr[\theta_S \in C' | K_A]}{\Pr[\theta_S\in C| K_A] } = \frac{\Pr[\theta_S\in C']\Pr[K_A|\theta_S\in C']}{\Pr[\theta_S \in C]}$ (by Bayes rule)
    \item Since $C$ is not weak and hence $\tau$ is near uniform in $C$ we have, $\frac{\Pr[\theta_S\in C']}{\Pr[\theta_S \in C]} = (1 + o(1/\sqrt n))$ (by Claim~\ref{claim:uniform-posterior})
\end{enumerate}

Thus, we get,
\begin{align*}
    \Pr[X_0 > d \text{ and } K_0 | \theta_S \in C] &= \sum_{A : |A|\approx k_0}\Pr[X_0 > d \text{ and } K_A| \theta_S \in C] \\
    & = \sum_{A : |A| \approx k_0} \Pr[K_A| \theta_S \in C] \Pr[ x_i \notin B_d \forall i\in A| K_A, \theta_S \in C] \\
    & = \sum_{A : |A| \approx k_0} \Pr[K_A| \theta_S \in C] \Pr[x_i \notin B_d \forall i\notin A | K_A ]\cdot \frac{\Pr[\theta_S \in C | x_A \notin B_d ,K_A]}{\Pr[\theta_S\in C| K_A] }\\
    & \le (1-d/n^{-a})^{k_0}  \sum_{A : |A| \approx k_0} \Pr[K_A| \theta_S \in C] \cdot \frac{\Pr[\theta_S \in C | x_A \notin B_d ,K_A]}{\Pr[\theta_S\in C| K_A] }\\
    &\le (1-d/n^{-a})^{k_0} \sum_{A : |A| \approx k_0} \frac{\Pr[\theta_S\in C']\Pr[K_A|\theta_S\in C']}{\Pr[\theta_S \in C]} \\
    &\le (1-d/n^{-a}))^{k_0} (1+o(1/\sqrt n))
\end{align*}

{Recall that $k_0=(1-\varepsilon_1)f(0)n^{1-a}(1-2/M)$. Let $\varepsilon_1 = O(n^{-1/22})$ and $M = n^{1/22}$. Since we can bound $1-x \le e^{-x}$ we get,}
\begin{align*}
\Pr[X_{(0)}\cdot\one(K_0)>d|\theta_S\in C] 
&\le \left(1 + O(n^{-1/2})\right) \cdot \exp\{ -2nf(0)(1-O(n^{-1/22}))(1 - O(n^{-1/4}))\cdot d\}
\end{align*}

\end{proof}

We finally bound the cost of the event with all the desirable properties.

    \begin{lemma}\label{lem:bound-commit-cost-term}
Fix any $C\in \mathcal{P}\setminus W_I$. Then we have, $\E[X^2_0\cdot \one(K_0)|\theta_S\in C] \le \frac{1}{2n^2f(0)^2} + o(1/n^2)$.
    \end{lemma}

\begin{proof}

Let $Z(\lambda)$ be the random variable with  exponential distribution. We observe that $\E[X_0^2 \cdot \one(K_0)|\theta_S\in C] \le \E[Z(\lambda)^2](1 + o(n^{-1/2}))$ for $\lambda = 2n f(0)(1 - O(n^{-1/22}))(1-O(n^{-1/4}))$, because of the stochastic dominance proved above in Lemma~\ref{lem:exponential-upper}. Moreover, $\E[Z(\lambda)^2] = \frac{2}{\lambda^2}$. Hence, we get $\E[X_0^2 \cdot \one(K_0)|\theta_S\in C] \le \frac{1}{2n^2f(0)^2}(1 - O(n^{-1/22}))(1-O(n^{-1/4})) \le \frac{1}{2n^2f(0)^2} + o(1/n^2)$.

\end{proof}

In Section~\ref{app:rare-events} we bound the loss due to the rare events of $\theta_S\notin I$, $C\in W_I$, and $\overline{K_0}$, that is, the number of signals in $N(C)$ is not in $(1\pm\varepsilon_1)n^{1-a}f(0)(M-2)/M$. We show that these contribute up to $o(1/n^2)$ loss.

\subsection{Contribution of Rare Events}\label{app:rare-events}

In this section we bound the loss from the rare events from Eq.~(\ref{eq:rare-events}).

\begin{lemma}\label{lem:thetaS-notin-I}
$\E[X^2_0\cdot\one\{\theta_S\notin I\}] \le O\left(\exp(-\frac{n^{2\epsilon}A}{2})\right)$ for some constant $A >0$.
\end{lemma}

\begin{proof}
Recall that $\theta_S(\vec{x})$ is the MMSE estimator and $\theta_S-\theta^* \rightarrow^{d} \mathcal{N}(0,C_{\mathcal{I}}/n)$. Hence the probability that $\theta_S \notin I$ is at most $\exp(-n^{2\epsilon}A)$ for some $A >0$. Let $P = \exp(-n^{2\epsilon}A)$. To bound $\E[X_0^2\cdot\one\{\theta_S\not\in I\}]$ we see that,
    \begin{align*}
        \E[X_0^2\cdot\one\{\theta_S\not\in I\}] &  = \int_0^\infty \Pr[X_0^2 > y \wedge \theta_S\not\in I]dy \\
        & = \int_0^{1/P^{1/2}} \Pr[\theta_S\not\in I]\Pr[X_0^2 > y|\theta_S\not\in I] dy +\int_{1/P^{1/2}}^\infty \Pr[X_0^2 > y \wedge \theta_S\not\in I]dy \\
        &\le P^{1/2} + 2\int_{1/P^{1/2}}^\infty \Pr[ X_i > \theta_S + \sqrt{y} \wedge \theta_S >0] dy \qquad (\text{For any arbitrary choice of  } i)\\
        &\le P^{1/2} + 2\int_{1/P^{1/2}}^\infty \Pr[X_i > \sqrt{y}] dy \\
        &\le P^{1/2} + 2\int_{1/P^{1/2}}^\infty (e^{-\sqrt{y}}) dy \qquad {(\text{Since }\int_{x}^\infty f(z) dz < e^{-x} \text{ for all } x > Q)}\\
        &\le P^{1/2} + 4(1/P^{1/4} +1)\exp\left(-1/P^{1/4}\right) \\
        &=O\left(\exp(-\frac{n^{2\epsilon}A}{2})\right)
     \end{align*}

Recall that, $P = \exp(-n^{2\epsilon}A$. So we have $P^{1/2}  = O\left(\exp(-\frac{n^{2\epsilon}A}{2})\right)$.  Since $xe^{-x}$ is $O(e^{-x})$ for $x$ sufficiently large, the term $4(1/P^{1/4} +1)\exp\left(-1/P^{1/4}\right) =  O\left(\exp\left( - \exp( \frac{n^{2\epsilon}A}{4})\right)\right)$.

\end{proof}

For a well-behaved distribution we have $1 - F(x) \le c_3e^{-x}$ for all $x > Q$.

\begin{claim}\label{claim:all-tail-signals}
Let $T_Q$ be the event that all $|x_i| > Q$. Then $\E[X^2_0 \one(T_Q \wedge \theta_S\in I)] \le o(1/n^2)$.
\end{claim}

\begin{proof}
Since $\theta_S \in I = [-n^{-1/2 + \epsilon}, n^{-1/2 + \epsilon}]$, and all signals $|x_i| > Q$ are outside $I$, we have that $X_0 = \min_i |x_i - \theta_S| \le |x_i| + n^{-1/2 + \epsilon}$ for all $x_i$. Let $t(x) =  (|x| + n^{-1/2 + \epsilon})^2$. Hence, we have,

\begin{align*}
    \E[X^2_0 \one(T_Q\wedge \theta_S\in I)] & \le \E[(|x_1| +n^{-1/2 + \epsilon})^2 \one(T_Q \wedge \theta_S\in I)]  \qquad (\text{For an arbitrary choice of } i = 1)\\
    & =  \E[t(x_1) \one(T_Q \wedge\theta_S\in I)] \\
    & = \int_{x_1 : |x_1| > Q}\cdots\int_{\vec{x}_{-1}} t(x_1) f(x_1)\prod_{i\neq 1} f(x_i) \one(|x_i| > Q)\cdot \one (\theta_S(\vec{x}) \in I) d\vec{x} \\
    & \le \int_{x_1 : |x_1| > Q}\cdots\int_{\vec{x}_{-1}} t(x_1) f(x_1)\prod_{i\neq 1} f(x_i) \one(|x_i| > Q) d\vec{x} \\
    &\le \int_{x_1 : |x_1| > Q} t(x_1) f(x_1) d x_1 \left( 2\exp(-Q) \right)^{n-1} \qquad (\text{ By tail bound Assumption of } f) \\
    & \le \exp\left(-\frac{Q(n-1)}{2}\right)2\int_{Q}^\infty (x + n^{-1/2 + \epsilon})^2 f(x) d x \\
    & \le \exp\left(-\frac{Q(n-1)}{2}\right)\cdot  Q \cdot O(1)\qquad (\text{ By tail bound Assumption of } f)
\end{align*}

\end{proof}

Recall that $W_I\subset \mathcal{P}$ is the set of all intervals $C$ such that $\tau(\overline{\theta}) < c' n^{1/22} n^{-4\log n + 1} $ for all $\overline{\theta}\in C$.

\begin{lemma}\label{lem:weak-posterior-density}
 Then $\sum_{C\in W_I}\E[X^2_0 \one\{\theta_S \in C\}] \le O\left(Q^2n^{-4\log n + 1} \right) + o(1/n^2)$.
\end{lemma}

\begin{proof}
Let $\overline{T}_Q$ be the event that there is some $|x_i| \le Q$. Since $\theta_S \in I$ and there is some $|x_i| \le Q$, we have that $X_0 \le Q+n^{-\frac{1}{2}+\epsilon}$. Thus, \[\sum_{C\in W_I}\E[X^2_0 \one\{\theta_S \in C\}\cdot\one\{\overline{T}_Q \}] \le (Q+n^{-\frac{1}{2}+\epsilon})^2\Pr[\theta_S \in W_{I}] \le (Q+n^{-\frac{1}{2}+\epsilon})^2O\left( n^{-4\log n + 1}\right)
\]
where the last inequality follows from Claim~\ref{claim:weak-interval-probability} (proved in Section~\ref{app:commitment-lemmas}) that $\Pr[\theta_S \in W_{I}] \le O\left( n^{-4\log n + 1}\right)$.

Moreover, by Claim~\ref{claim:all-tail-signals} proved above, we have $\sum_{C\in W_I}\E[X^2_0 \one\{\theta_S \in C\}\cdot\one\{ T_Q \}] \le \E[X^2_0 \one(T_Q \wedge \theta_S\in I)] \le o(1/n^2)$. Thus, proving the lemma.

\end{proof}

\begin{lemma} \label{lem:few-near-signals}
   {Let  $k_0 = f(0)n^{1-a}$}, $\varepsilon_1 = O(n^{-1/22})$. Let $\overline{K_0}$ be the event such that $Y(C) \neq k_0 (1 \pm \varepsilon_1)$. Then,
    $\sum_{C\in \mathcal{P}\setminus W_I} \E[X_0^2\cdot\one\{\overline{K_0} \text{ and } \theta_S\in C\} ] \le (Q+n^{-1/2})^2O\left(\exp\left(-n^{1/20} f(0)\right)\right)+ O(\exp^{-Qn}) \le o(1/n^2)$.
    \end{lemma}

    \begin{proof}
    Consider the case where $N(C)$ doesn't have $k_0(1\pm \varepsilon_1)$ signals. Let $A_0$ denote that event.
 By Claim~\ref{claim:all-tail-signals}, when all $|x_i| > Q$ and $\theta_S \in I$, we bound the expected $X^2_0$ by $O(\exp^{-Qn})$.  If there is even a single $|x_i| \le Q$ (denoted by the event $\overline{T}_Q$), then we can bound $X_0^2$ by $(Q + n^{-1/2 +\epsilon})^2$ because $\theta_S \in I = [-n^{-1/2+\epsilon}, n^{-1/2 + \epsilon}]$. By corollary~\ref{cor:k-signals-in-L}, we get that $\Pr[ A_0 \cap \overline{T}_Q \cap \theta_S \in C] \le \Pr[ A_0] \le \exp\left(-\frac{n^{-1/11} f(0)n^{1-a}}{3M}\right)$.

 \begin{align*}
 \E[X_0^2\cdot\one\{\overline{K_0} \text{ and } \theta_S\in C\} ] & \le  \E[X_0^2\cdot\one\{A_0\text{ and } \theta_S\in C\}\cdot\one\{{T}_Q\} ] + \E[X_0^2\cdot\one\{A_0\text{ and } \theta_S\in C\}\cdot\one\{\overline{T}_Q\} ] \\
 &\le  O(\exp^{-Qn}) + (Q + n^{-1/2 +\epsilon})^2\exp\left(-\frac{n^{-1/11} f(0)n^{1-a}}{3M}\right) \\
 &\le o(1/n^2)
 \end{align*}
    \end{proof}

\subsection{Loss of Biased Communication Schemes}\label{app:biased-loss}
In this section we 
that for all sufficiently large $n$ and all $\delta$, $\E[X^2_\delta] \ge \frac{1}{2n^2f(\delta)^2} - o(1/n^2)$. We focus only on $|\delta| \le 2(\log n)^2$ {and $\delta$ such that $f(\delta)\ge n^{-1/100}$}~\footnote{When $f(\delta) \le n^{-1/100}$ we see that $\E[X^2_\delta] \ge \Omega(\frac{n^{1/100}}{n^{2}}) >>\E[X^2_0]$. }.
 
 Recall $\mathcal{P}$ be a partition of $I$ into intervals of length $n^{-b}$. For any $\delta$, we denote $C_\delta = C + \delta$ for any $C \in P$ and $N(C_\delta) = L_\delta\cup R_\delta$, where $L_\delta$ (resp. $R_\delta$) is the neighboring interval of length $n^{-a}$ to the left of $C_\delta$ (resp. to the right of $C_\delta$).

Similar to the unbiased loss we consider the high probability event where all the following desirable properties hold:

\begin{enumerate}
\item $\theta_S(\vec x, y) \in I$ and let $C\in \mathcal{P}$ be the interval with $\theta_S(\vec x, y)$,
\item $C$ is not \emph{weak} (see Definition~\ref{def:weak-posterior}), and
\item there are \emph{sufficiently few} signals in $N(C_\delta) = L_\delta\cup R_\delta$. 
\end{enumerate}

We will show that with high probability all the desirable properties hold.

With this we are ready to bound the signaling loss.
Recall that the signaling loss of a biased communication scheme $\pi$ with $\delta(\pi) = \delta$ is $L(\pi, \delta(\pi)) \ge L(\pi_\delta, \delta) = \alpha^2\E[X^2_\delta]$. Given $\theta_S \in C$, let $K_\delta$ be the event that there are \emph{sufficiently few} signals in $N(C_\delta)$ and there are no signals in $C_\delta$.  We see that,

\begin{eqnarray}
\E[X^2_\delta] & \ge &  {\sum_{C\in \mathcal{P}\setminus W_I}\E[X^2_\delta\cdot\one\{\theta_S\in C, K_\delta\}]}  \notag \\
& \ge & {\sum_{C\in \mathcal{P}\setminus W_I} \Pr[\{\theta_S\in C, K_\delta\}]\cdot \underbrace{\E[X^2_\delta~|\theta_S\in C, K_\delta]}_{\mbox{conditional expectation}} }\label{eq:commit-cost-biased}\\
& \ge & \frac{1}{2n^2f(\delta)^2} - o(1/n^2) \notag
\end{eqnarray}

This is because by Lemma~\ref{lem:commit-cost-delta}  the conditional expectation term in Eq.~(\ref{eq:commit-cost-biased}) is $\E[X^2_\delta~|\theta_S\in C, K]  \ge \frac{1}{2n^2f(\delta)^2}(1 - o(1))$. We note that $ \{ \vec{x} : \forall J_\delta \in \mathcal{J}_\delta \text{ with sufficiently few signals}\}\subset K_\delta$ for all $C \in \mathcal{P}$,\ and by Lemma~\ref{lem:many-signals-J-delta} we see that $ \Pr[\overline{K_\delta}] \le \Pr[\exists J_\delta \in \mathcal{J}_\delta \text{ with too many signals}] \le O(\exp\{f(\delta)n^{1/20}\})$. By Lemma~\ref{lem:weak-posterior-density} we bound the probability of $C$ is weak. Thus, we have, 
\begin{align*}
\sum_{C\in \mathcal{P}\setminus W_I} \Pr[\{\theta_S\in C, K_\delta\}] &\ge \Pr[\forall J_\delta \in \mathcal{J}_\delta \text{ with sufficiently few signals, and } \theta_S \notin W_I] \\
&\ge 1 - O(\exp\{f(\delta)n^{1/20}\}) - O(\exp\{f(\delta)n^{1/55}\})O(n^{4\log n -1}).
\end{align*}

We start by showing a more general version of Claim~\ref{claim:uniform-density}.
\begin{claim}\label{claim:uniform-density-delta}
Given any well-behaved distribution with pdf $f$,  for all $x\in I + \delta$, we have $f(x)= f(\delta) (1 + O(1/n))$.
\end{claim}

\begin{proof}
By mean value theorem we see that $f(x)= f(\delta) + (x-\delta) f'(\tilde{x})$ for some $\tilde{x} \in [\delta, x]$.
By our assumption on $g'$ we get $|g(x)| \le c x^m$ for some constant $c>0$. This implies $|f'(x)| \le c |x|^m f(x)$ for all $x$.   Since $|\delta| < 2( \log n)^2$ and $x \in I+\delta$, we have $|x|\le 2 (\log n)^2+ n^{-1/2 + \epsilon}$.  By mean value theorem we have,
\begin{align*}
|f(x) - f(\delta)| & = |(x-\delta)f'(\tilde{x})| \\
& \le |x - \delta| c |\tilde{x}|^m f(\tilde{x}) \qquad (\text{Since }|f'(\tilde{x})| \le c|\tilde{x}|^{m}f(\tilde{x}) )\\
& \le c (n^{-1/2 + \epsilon}) (2(\log n)^2 +  n^{-1/2 + \epsilon} )^m f(\tilde{x}) \\
& \le c' (n^{-1/2 + \epsilon}) n^{1/4}f(\tilde{x}) \qquad (\text{ Since } (\log n)^{2m} = o(n^{1/4}))\\
& \le f(\delta) + c' n^{-1/4 + \epsilon}f(\tilde{x})
\end{align*}

 Without loss of generality, we assume that $x,\delta > 0$, since $f$ is symmetric. Thus, we get $f(\tilde{x}_i)$ is between $f(x)$ and $f(\delta)$, because $f$ is single-peaked.

 Suppose $f(x)\ge f(\tilde{x}) \ge f(\delta)$, then $f(x) \le f (\delta) + c' n^{-1/4 + \epsilon}f(\tilde{x}) \le f (\delta) + c' n^{-1/4 + \epsilon}f(x)$. So we get $f(x)(1 - c' n^{-1/4 + \epsilon}) \le f(\delta)$. Thus, $f(x) \le f(\delta)\frac{1}{1 - c' n^{-1/4 + \epsilon}}  \le  f(\delta) ( 1 + c''n^{-1/4 + \epsilon} )$ for some constant $c'' >0$.

Similarly, if $f(x)\le f(\tilde{x}) \le f(\delta)$, then
$f(x) \ge f(\delta) - c' n^{-1/4 + \epsilon}f(\tilde{x}) \ge  f(\delta) - c' n^{-1/4 + \epsilon}f({\delta})$. Thus, we get, $f(x) \ge f(\delta)( 1 -O( n^{-1/4 + \epsilon} ) )$.

Therefore, for $\delta < 2(\log n)^2$ and all $x\in I+\delta$ we have $f(x)= f(\delta)(1 - O(n^{-\frac{1}{4} + \epsilon}))$.
\end{proof}

\begin{claim}
Given any interval $A\subset I+\delta$ of length $\ell$, the expected number signals in $A$ is ${n\ell}f(\delta)(1 - O(n^{-1/4}))$. Let $Y(A)$ be the number of signals in $A$. For any $0<\varepsilon_1<1$, we have $$\Pr[ |Y(A) - \E[Y(A)]| \ge \varepsilon_1)\E[Y(A)] ] \le \exp\left(-\frac{\varepsilon_1^2 f(\delta)n\ell}{3}\right).$$
\end{claim}

\begin{proof}

Let $Y_i = 1$ if $x_i \in A$ and $0$ otherwise. So we have, $\sum_{i=1}^n Y_i = Y(A)$. By Claim~\ref{claim:uniform-density-delta} we have $f(x) = f(\delta)(1 + O(n^{-1/4+\epsilon}))$ for all $x\in I+\delta$. Therefore, we have $\Pr[x_i \in A] = \int_A f(x) dx =f(\delta)(1 + O(n^{-1/4+\epsilon}))\int_A dx = f(\delta)(1 + O(n^{-1/4+\epsilon}))\ell $. Note that, $Y_i$ are i.i.d.~random variables, and $\E[Y(A)] =f(\delta)(1 + O(n^{-1/4+\epsilon}))n\ell$. By using Chernoff bound we get
$$\Pr[ Y(A) \ge (1+\varepsilon)\E[Y(A)] ] \le \exp\left(-\frac{\varepsilon_1^2 f(\delta)n\ell(1 +O(n^{-1/4+\epsilon}))}{3}\right).$$
\end{proof}

We again partition $I+\delta$ into intervals $J_\delta$ of size $n^{-a}/M$. Exactly following Lemma~\ref{lem:many-signals-J} we see that \emph{all} $J_\delta\subset I+\delta$ have $(1\pm \varepsilon)f(\delta)n^{1-a}(1 + O(n^{-1/4+\epsilon})$ many signals in $L_\delta$ and $R_\delta$.

\begin{lemma}\label{lem:many-signals-J-delta} 
Let $k_m = f(\delta)n^{1-a}/M$. The probability that there is a $J_\delta \in \mathcal{J}_\delta$ with more than $(1+\varepsilon)k_m$ signals (or less than $(1-\varepsilon)k_m$ is $O(\exp\{\varepsilon^2f(\delta)n^{1-a}/3M\})$.
\end{lemma}
 
\begin{corollary}\label{cor:k-signals-in-L-delta}
Let $k_{\delta} =  f (\delta) n^{1 - a}$, let $k^*= (1 + \varepsilon_1)k_{\delta}$, and $k'=(1-\varepsilon_1)k_{\delta}(1-2/M)$. Let $Y(N_\delta)$ be the number signals in interval $N(C_\delta)$ (of length $n^{-a}$). $\Pr[ ( Y(N_\delta)  \notin [k', k^*] ] \le \exp\left(-\frac{\varepsilon_1^2 f(\delta)n^{1-a}}{3M}\right).$
\end{corollary}

  Let $K_\delta$ denote the event that there are at most $(1+\varepsilon_1)k_\delta$ and at least $(1-\varepsilon_1)k_\delta(1-2/M)$ many signals in $N(C_\delta)$, and there are no signals in $C_\delta$. 

\begin{lemma}{\label{lem:exponential-lower}}
The distribution of the random variable $X_\delta$ conditioned on $K_\delta,\theta_S\in C$  stochastically dominates (up to a factor of $(1 - o(1))$) the exponential distribution with {$\lambda= 2nf(\delta)(1 + O(n^{-1/10}))$}. That is, for $d < n^{-9/10}$, $\Pr[X_{\delta} > d \wedge K_\delta~|\theta_S\in C] \ge \exp[- \lambda d]  \frac{ \Pr[\theta_S\in C\setminus E \text{ and } K_\delta ]}{\Pr[\theta_S \in C ] }$.
\end{lemma}

\begin{proof}
Let $k^* = (1 + \varepsilon_1)k_\delta$, and $k' = (1-\varepsilon_1)k_\delta(1-2/M)$. For all $A\subset [n]$ such that $k'\le |A| \le k^*$, define $K_A$ to be event where $x_i \in N(C_\delta)$ iff $i\in A$, and there are no signals in $C_\delta$ .

For all $d<n^{-a}+n^{-b}$, let $B_d$ denote the interval of length $2d$ centered around $\theta_S + \delta$.

We will use the following results/facts:

\begin{enumerate}
    \item $\Pr[ x_A \notin B_d | K_A, \theta_S \in C] =\Pr[x_A \notin B_d | K_A ]\cdot \frac{\Pr[\theta_S \in C | x_A \notin B_d ,K_A]}{\Pr[\theta_S\in C| K_A] }$ (by Bayes rule)
    \item Since $f$ is near uniform in $N(C_\delta)$ we have $\Pr[x_A \notin B_d | K_A ] \approx (1-\frac{d}{n^{-a}})^{|A|}$ (by Claim~\ref{claim:uniform-density-delta})
    \item Since redrawing $x_A \in N(C)$ doesn't change $\theta_S$ much, we have $\Pr[\theta_S \in C | x_A \notin B_d ,K_A] \ge \Pr[\theta_S \in C\setminus E | K_A]$ (by Lemma~\ref{lem:posterior-correlation-upper})
    \item $\Pr[K_A |\theta_S \in C]\cdot \frac{\Pr[\theta_S \in C\setminus E | K_A]}{\Pr[\theta_S\in C| K_A] } = \frac{\Pr[\theta_S\in C\setminus E \wedge K_A]}{\Pr[\theta_S \in C]}$ (by Bayes rule)
\end{enumerate}

Thus, we get,
\begin{align*}
    \Pr[X_\delta > d \cdot\one{(K_\delta)} | \theta_S \in C] &= \sum_{A : |A|\in (k',k^*)}\Pr[X_0 > d \text{ and } K_A| \theta_S \in C] \\
    & = \sum_{A : |A|\in (k',k^*)} \Pr[K_A| \theta_S \in C] \Pr[ x_i \notin B_d \forall i\in A| K_A, \theta_S \in C] \\
    & = \sum_{A : |A| \in (k',k^*)} \Pr[K_A| \theta_S \in C] \Pr[x_i \notin B_d \forall i\notin A | K_A ]\cdot \frac{\Pr[\theta_S \in C | x_A \notin B_d ,K_A]}{\Pr[\theta_S\in C| K_A] }\\
    & \ge \left(1-\frac{d}{n^{-a}}(1 + O(n^{-1/4}))\right)^{k^*}  \sum_{A : |A| \in (k',k^*)} \Pr[K_A| \theta_S \in C] \cdot \frac{\Pr[\theta_S \in C \setminus E| K_A]}{\Pr[\theta_S\in C| K_A] }\\
    &\ge \left(1-\frac{d}{n^{-a}}(1 + O(n^{-1/4}))\right)^{k^*}  \sum_{A : |A| \in (k',k^*)} \frac{\Pr[\theta_S\in C\setminus E \wedge K_A]}{\Pr[\theta_S \in C]}\\
    &\ge \left(1-\frac{d}{n^{-a}}(1 + O(n^{-1/4}))\right)^{k^*}  \frac{\Pr[\theta_S\in C\setminus E \wedge K_\delta]}{\Pr[\theta_S \in C]}
\end{align*}

We bound $ \left( 1 - \frac{2d}{n^{-a}}\right)^{f(\delta)(1+\varepsilon_1)n^{1-a}}$ by observing that $(1 - x) \ge e^{-x-{x^2}} $ for $x < 1/2$. We will choose of $\varepsilon_1 = O(n^{-1/20})$ and consider $d \le n^{- a -1/20} = n^{-17/20}$, this gives us $\varepsilon_3 = O(n^{-1/20})$.
\begin{align*}
\Pr[X_{(\delta)}\cdot\one(K_\delta)>d|\theta_S \in C]
& \ge \left(1 - \frac{2d}{n^{-a}} (1 + O(1/n^{1/4}))\right)^{k^*}  \frac{ \Pr[\theta_S\in C\setminus E \text{ and } K_\delta ]}{\Pr[\theta_S \in C ] } \\
& \ge \frac{ \Pr[\theta_S\in C\setminus E \wedge K_\delta ]}{\Pr[\theta_S \in C ] }\exp(-2dn(1 + \varepsilon_2)f(\delta))\exp(-{(2d)^2n^{1 + a}f(\delta)(1+\varepsilon_2)})\\
& \ge  \frac{ \Pr[\theta_S\in C\setminus E \wedge K_\delta ]}{\Pr[\theta_S \in C ] } \exp(-2dn(1 + \varepsilon_3)f(\delta)) \qquad (\text{For } 2d < n^{-a}\varepsilon_3)
\end{align*}

\end{proof}

Finally, for each $C \in \mathcal{P}\setminus W_I$, we bound the loss $\E[X_\delta^2 \cdot\one\{{K}_\delta , \theta_S\in C\} ] $.

  \begin{lemma} \label{lem:commit-cost-delta}
   Let ${K}_\delta$ be the event such that $(1 - \varepsilon_1)(1-2/M) k_\delta \le Y(L_\delta) \le (1+\varepsilon_1)k_\delta$, and no signals in $C_\delta$. Then for all $C \in \mathcal{P}\setminus W_I$,
   
   \begin{enumerate}
   \item for $f(\delta) >n^{-1/100}$, we have $\E[X_\delta^2 \cdot\one\{{K}_\delta , \theta_S\in C\} ] \ge  \frac{1}{2n^2f(\delta)^2} (1 - o(1)) { \Pr[\theta_S\in C\setminus E \text{ and } K_\delta ]}$,
   \item for $f(\delta) \le n^{-1/100}$, we have $\E[X_\delta^2] \ge c_5 \frac{n^{1/100}}{n^{2}}$.
   \end{enumerate}
   
    \end{lemma}

    \begin{proof}

Let $Z(\lambda)$ be the random variable with  exponential distribution. We  observe that $\E[X_\delta^2 \cdot \one(K_\delta)|\theta_S\in C] \ge \E[Z(\lambda)^2\cdot \one\{d < {(n^{-9/10})^2}\}]\cdot \frac{ \Pr[\theta_S\in C\setminus E \text{ and } K_\delta ]}{\Pr[\theta_S \in C ] } $ for $\lambda = 2n f(\delta)(1+\varepsilon_3)$, because of the stochastic dominance proved above in Lemma~\ref{lem:exponential-lower}. Moreover, $\E[Z(\lambda)^2\cdot \one\{d < {({n^{-17/20}})^2}\}]= \frac{2}{\lambda^2}(1 - O(\lambda n^{-17/20}\exp^{-\lambda n^{-17/20}}))$. Hence, we get $\E[X_\delta^2 \cdot \one(K_\delta)|\theta_S\in C] \ge \frac{1}{2n^2f(\delta)^2} (1 - o(1))\frac{ \Pr[\theta_S\in C\setminus E \text{ and } K_\delta ]}{\Pr[\theta_S \in C ] } $, for $f(\delta) \ge c_4n^{-1/100}$ for a constant $c_4 >0$.

{We finish by noting that, for sufficiently large $n$, when $f(\delta) = O(\frac{1}{n^{1/100}})$, we get $\E[X_\delta^2] \ge c_5(\frac{n^{1/100}}{n^2}) >> \frac{1}{2n^2f(0)^2} \ge \E[X_0^2] $ for a constant $c_5 > 0$.}

\end{proof}

\begin{proof}[Proof of Proposition~\ref{prop:unbiased-biased-bounds} (b)]
By lemma~\ref{lem:commit-cost-delta},  for a sufficiently large $n$ and any $\delta < 2(\log n)^2$ and $f(\delta) > n^{-1/100}$, we have $\E[X_\delta^2] \ge\sum_{C\in \mathcal{P}\setminus W_I}\Pr[\theta_S\in C]\E[X_\delta^2|\theta_S \in C] \ge \sum_{C\in \mathcal{P}\setminus W_I} \frac{1}{2n^2f(\delta)^2}(1 - o(1)){ \Pr[\theta_S\in C\setminus E \text{ and } K_\delta ]}$ .

Given $E\subset C$ as the union of first and last $kn^{-1-a+\epsilon}$ length interval, we define $2E \subset C$ to be the union of the first and last $2kn^{-1-a+\epsilon}$ intervals.
Since $K_\delta$ is the event that there are $(1 \pm \varepsilon)k_\delta$ many signals in $N(C_\delta)$ and there are no signals in $C_\delta$, that $\Pr[\theta_S\in C\setminus E] \ge \Pr[\theta_S \in C\setminus 2E| Y(N(C_\delta)\cup C_\delta) \in (1 \pm \varepsilon)k_\delta ]$, this is because if $\theta_S \in C\setminus 2E$ and $\vec{x}_{[k]} \in N(C_\delta)\cup C_\delta$ then rearranging the signals $\vec{x}_{[k]}$ by moving the signals in $C_\delta$ into $N(C_\delta)$ changes $\theta_S$ by at most $kn^{-1-a} f(\delta)$. Moreover, we have $\Pr[K_\delta] = \Pr[Y(N(C_\delta)\cup C_\delta) \in (1 \pm \varepsilon)k_\delta]\cdot \Pr[Y(C_\delta) = 0~|Y(N(C_\delta)\cup C_\delta) \in (1 \pm \varepsilon)k_\delta]$. The probability that there are no signals in $C_\delta$ (of length $n^{-12/11}$) is at least $(1 - O(n^{-1/11}))$. Thus, we get,

\[
\E[X_\delta^2] \ge \sum_{C\in \mathcal{P}\setminus W_I} \frac{1}{2n^2f(\delta)^2}(1 - o(1)){ \Pr[\theta_S\in C\setminus 2E \text{ and } Y(N(C_\delta)\cup C_\delta) \in (1 \pm \varepsilon)k_\delta ]} (1 - O(n^{-1/11}))
\]

Let $\mathcal{G}_\delta$ denote the event that all $J_\delta \in \mathcal{J}_\delta $ has $(1 \pm \varepsilon) k_\delta/M$ signals. Thus, we get,

\begin{align*}
\E[X_\delta^2] &\ge \sum_{C\in \mathcal{P}\setminus W_I} \frac{1}{2n^2f(\delta)^2}(1 - o(1)){ \Pr[\theta_S\in C\setminus 2E \text{ and } \mathcal{G}_\delta]} (1 - O(n^{-1/11}))\\
& = (1 - O(n^{-1/11}))\Pr[ \mathcal{G}_\delta \text{ and } \exists C \in \mathcal{P}\setminus W_I \text{ s.t.~} \theta_S \in C\setminus 2E]\\
& \ge (1 - O(n^{-1/11}))\left(1 - \Pr[\overline{\mathcal{G}}_\delta] - \Pr[ \overline{\exists C \in \mathcal{P}\setminus W_I \text{ s.t.~} \theta_S \in C\setminus 2E}]\right) \\
& \ge  (1 - O(n^{-1/11}))\left( - \Pr[\overline{\mathcal{G}}_\delta] + \Pr[\exists C \in \mathcal{P}\setminus W_I \text{ s.t.~} \theta_S \in C\setminus 2E ] \right)
\end{align*} 

Recall that, $\Pr[\theta_S \in I] \ge (1 - \exp(-n^{2\epsilon}A)$, and by Claim~\ref{claim:weak-interval-probability} (proved in Section~\ref{app:commitment-lemmas}) we have $\Pr[\theta_S\in W_I] \le O(n^{-4log n +1})$. By Lemma~\ref{lem:many-signals-J-delta} we have $\Pr[\overline{\mathcal{G}}_\delta] \le O(\exp\{\varepsilon^2f(\delta)n^{1-a}/3M\})$. Hence, we get $\E[X_\delta^2] \ge \frac{1}{2n^2f(\delta)^2}(1 - o(1))$.
\end{proof}

\subsection{Helpful lemmas to bound the correlation between \texorpdfstring{$\theta_S$}{thetaS} and the closest signal}\label{app:commitment-lemmas}
In this section we will introduce some helpful lemmas for proving Proposition~\ref{prop:unbiased-biased-bounds}.

In Lemma~\ref{lem:posterior-differential}, we characterize the effect of a single signal $x$ on the posterior mean $\theta_S(\vec{x})$. This lemma directly implies Corollary~\ref{cor:posterior-change}, where we show that if for any signals $\vec{x}$ rearranging at most $k$ signals in $L$ (and $R$) to get $\vec{y}$ guarantees that the new posterior mean is $\theta_S(\vec{y}) \in \theta_S(\vec{x}) \pm O(k n^{-1-a})$.

 \begin{lemma}{\label{lem:posterior-differential}}
For any signals $\vec{x}$ observed by the sender we have,
\[
\Big\lvert\frac{\partial {\theta_S}(\vec{x})}{\partial x_i} \Big\rvert\le c_1 {Var_{\theta\sim D_S(\vec{x})}[\theta] + 2\theta_S(\vec{x})^2},
\]

where $Var_{\theta\sim D_S(\vec{x})}[\theta]$ is the variance of the sender's posterior distribution $D_S(\vec x)$.
\end{lemma}

\begin{proof}

Let $h_S(\theta|\vec{x})$ denote the pdf of the sender's posterior distribution, $h(\theta)$ be the (constant) pdf of the diffuse prior. Note that the pdf of a signal $x$ given that the state of the world is $\theta$ (denoted by $\hat{f}(x|\theta)$) equals $f(x - \theta)$, where $f$ is the pdf of $F$. Recall that when the sender observes $\vec{x}$ they update their posterior in a Bayesian way. Hence, we have,

\begin{align*}
h_S(\theta|\vec{x}) & = \frac{\prod_i\hat{f}(x_i|\theta)h(\theta)}{\int_{\hat{\theta}} \prod_i\hat{f}(x_i|\hat{\theta})h(\hat{\theta})d\hat{\theta}} \\
& = \frac{\prod_i{f}(x_i - \theta)}{\int_{\hat{\theta}}\prod_i f(x_i - \hat{\theta})d\hat{\theta}}
\end{align*}

 and the sender's posterior mean is
\begin{equation}
{\theta_S}(\vec{x}) = \frac{\int_{\theta}\theta\prod_i{f}(x_i - \theta)d\theta}{\int_{\theta}\prod_i f(x_i - \theta)d\theta}
\end{equation}

We want to understand $\frac{\partial {\theta_S}(\vec{x})}{\partial x_i}$ which is the effect of a single signal $x_i$ on the posterior mean.

We can deduce
\begin{eqnarray}
\frac{\partial {\theta_S}(\vec{x})}{\partial x_i} &=& \frac{\int \theta f'(x_i-\theta) \prod_{j \neq i} f(x_j-\theta)d\theta \int \prod_{j} f(x_j-\theta)d\theta}{\left(\int \prod_{j} f(x_j-\theta)d\theta\right)^2} \nonumber \\
&& - \frac{\int \theta\prod_{j} f(x_j-\theta)d\theta \int f'(x_i-\theta) \prod_{j\neq i} f(x_j-\theta)d\theta}{\left(\int \prod_{j} f(x_j-\theta)d\theta\right)^2} \nonumber \\
&=& \frac{\int \theta g(x_i-\theta) \prod_{j} f(x_j-\theta)d\theta \int \prod_{j} f(x_j-\theta)d\theta}{\left(\int \prod_{j} f(x_j-\theta)d\theta\right)^2} \nonumber \\
&& - \frac{\int \theta\prod_{j} f(x_j-\theta)d\theta \int g(x_i-\theta) \prod_{j} f(x_j-\theta)d\theta}{\left(\int \prod_{j} f(x_j-\theta)d\theta\right)^2}
\end{eqnarray}
where $g(y)=\frac{f'(y)}{f(y)}$.

Using the mean value theorem we can write $g(x_i-\theta)=g(x_i)-g'(\tilde{x}_i)\theta$ for $\tilde{x}_i\in [x_i-\theta,x_i]$. We then obtain:
\begin{eqnarray}
\frac{\partial {\theta_S}(\vec{x})}{\partial x_i} &=& \frac{\int \theta \left(g(x_i)-g'(\tilde{x}_i)\theta\right) \prod_{j} f(x_j-\theta)d\theta \int \prod_{j} f(x_j-\theta)d\theta}{\left(\int \prod_{j} f(x_j-\theta)d\theta\right)^2} \nonumber \\
&&- \frac{\int \theta\prod_{j} f(x_j-\theta)d\theta \int \left(g(x_i)-g'(\tilde{x}_i)\theta\right) \prod_{j} f(x_j-\theta)d\theta}{\left(\int \prod_{j} f(x_j-\theta)d\theta\right)^2} \nonumber \\
&=& -\frac{\int \theta g'(\tilde{x}_i) \theta  \prod_{j} f(x_j-\theta)d\theta \int \prod_{j} f(x_j-\theta)d\theta}{\left(\int \prod_{j} f(x_j-\theta)d\theta\right)^2} \nonumber \\
 && +\frac{\int \theta\prod_{j} f(x_j-\theta)d\theta \int g'(\tilde{x}_i)\theta  \prod_{j} f(x_j-\theta)d\theta}{\left(\int \prod_{j} f(x_j-\theta)d\theta\right)^2}
\end{eqnarray}

Next, we note that for a well-behaved distribution, $|g'(\tilde{x}| \le c_1$. Hence, we can bound,

We can then simplify:
\begin{eqnarray}
\Big\lvert\frac{\partial{\theta_S}(\vec x)}{\partial x_i}\Big\rvert &\le &  c_1\left( \frac{\int_{\theta}\theta^2\prod_j{f}(x_j - \theta)d\theta}{\int_{\theta}\prod_j f(x_j - \theta)d\theta} + \left(\frac{\int_{\theta}\theta\prod_j{f}(x_j - \theta)d\theta}{\int_{\theta}\prod_j f(x_j - \theta)d\theta}\right)^2 \right) \\
&= &c_1\left(\E_{\theta\sim D_S(\vec{x})}[\theta^2|\vec x]+\E_{\theta\sim D_S(\vec{x})}[\theta|\vec x]^2\right) \nonumber \\
&= & c_1 \left({Var_{\theta\sim D_S(\vec{x})}[\theta]} + 2\E_{\theta\sim D_S(\vec{x})}[\theta|\vec x]^2 \right)
\end{eqnarray}
\end{proof}

Next. we bound the shift in $\theta_S$ when rearranging $k$ signals in $L$ (or $R$) using Lemma~\ref{lem:posterior-differential}. We note that  $Var_{\theta\sim D_S(\vec{x})}[\theta] = O(1/n)$ and $\theta_S \in I$.

\begin{corollary}{\label{cor:posterior-change}}
 Assume that the posterior mean $\theta_S(\vec{x})$ lies within the interval $C\subset I$. Consider a subset of $k$ signals in a subset $A$ of length $\ell$. Any rearrangement of these signals within $A$ changes the posterior mean by $O(k\ell n^{-1+2 \epsilon})$.
\end{corollary}

\begin{proof}
We prove this by using the mean value theorem on the function $\theta_S : \reals^n \rightarrow \reals$. Given $\vec{x}$, consider a subset of signals $\vec{x}_{[k]} \in A$.  Let $\vec{y}$ be any vector such that $y_i\in A$ for all $i\in [k]$  and $y_i = x_i$ for the rest. By mean value theorem we get for some $\vec{z}$ such that:

\[
\theta_S(\vec{y}) = \theta_S(\vec{x}) + \nabla\theta_S(\vec{z})\cdot(\vec{y} -\vec{x})
\]

Note that $|y_i - x_i| \le \ell$ for all $i\in [k]$ and $y_i - x_i = 0$ otherwise. That is, at most $k$ terms with $|y_i - x_i| \neq 0$. Hence, we get,

\begin{align*}
|\theta_S(\vec{y}) - \theta_S(\vec{x})| &= |\nabla\theta_S(\vec{z})\cdot(\vec{y} -\vec{x})|\\
&\le \sum_{i=1}^{k} c_1 \left({Var_{\theta\sim D_S(\vec{x})}[\theta] + 2\theta_S(\vec{x})^2}\right) |y_i - x_i| \quad \text{(By Lemma~\ref{lem:posterior-differential})}\\
&=  c_1{O(n^{-1+2\epsilon})}\sum_{i=1}^{k}\ell
\end{align*}

Since $x_i,y_i \in A$ we bound $|x_i - y_i| \le \ell$. Further, by bounding $Var_{\theta\sim D_S(\vec{x})}[\theta]$ by $O(1/n)$, and $\theta_S(\vec{x})$ by $n^{-1/2 +\epsilon}$, we get $\theta_S(\vec{y}) = \theta_S(\vec{x}) \pm k\ell O(n^{-1 + 2\epsilon})$ when rearranging at most $k$ signals in each $A$.

\end{proof}

Next, we show that the density of the posterior mean is nice in the interval $C$.

    \begin{claim}\label{claim:uniform-posterior}
Assume that the density function $f$ has exponential (or thinner) tails. Let $\tau(\cdot)$ be the density function of the posterior mean. Then for all $\overline{\theta}\in C $ we have $\tau(\overline{\theta}+\epsilon') = \tau(\overline{\theta}) \left(1+O(n^{-1/22})\right) + O(e^{-n^{1/22\alpha}})$ for all $0 < \epsilon'\le 1/n^b$.
\end{claim}

\begin{proof}
Fix a posterior mean $\overline{\theta} \in  [-n^{-\frac{1}{2}+\epsilon},n^{-\frac{1}{2}+\epsilon}]$ and consider all the signal draws $X(\overline{\theta})=\left\{\vec{x}|\overline{\theta}(\vec{x}) = \overline{\theta}\right\}$ that generate this posterior mean. We know that $\tau(\overline{\theta})=\int_{\vec{x} \in X(\overline{\theta})} \prod_{x_i \in \vec{x}} f(x_i) dx$. Now consider $\overline{\theta}+\epsilon'$. We can couple all the signal realizations in $X(\overline{\theta}+\epsilon')$ and $X(\overline{\theta})$ by considering uniform shifts of the corresponding $\vec{x}$ by $\epsilon'$. This follows from the assumption of a diffuse prior, and we get $\theta_S(\vec x + \epsilon') = \theta_S(\vec x) + \epsilon'$. That is, $\tau(\overline{\theta}+\epsilon')=\int_{\vec{x} \in X(\overline{\theta})} \prod_{x_i \in \vec{x}} f(x_i+\epsilon') dx$.

Next, consider the probability of observing $x$ versus the coupled signal realizations $x+\epsilon'$.

      \begin{eqnarray}
      \prod_{x_i \in \vec{x}} f(x_i+\epsilon') &=& \prod_{x_i \in \vec{x}} \left[f(x_i)+f'(\tilde{x}_i)\epsilon'\right]
      \end{eqnarray}

Recall that, by our assumption on $g'$ we have $|f'(x)| \le c |x|^{m}f(x)$.
\begin{itemize}
    \item Note that, for all $|x_i| < 4(\log n)^2$, we have
\begin{align*}
|f(x_i +\epsilon') - f(x_i)|  &= |\epsilon'f'(\tilde{x}_i)| \\
&\le |\epsilon'c(|\tilde{x}|^{m})|f(\tilde{x}) \\
&\le |\epsilon'c(|4 \log n|^{2m})|f(\tilde{x}) \\
&\le |c'(n^{-b+1/22})|f(\tilde{x}) \qquad (|\epsilon'|\le n^{-12/11} \text{ and } (\log n)^{2m} = o(n^{1/22})
\end{align*}

Note that, wlog we can assume that $\mathrm{sign}(x_i + \epsilon) =\mathrm{sign}(x_i) $ because $f$ is symmetric. Thus, we get $f(\tilde{x}_i)$ is between $f(x_i)$ and $f(x_i + \epsilon')$, because $f$ is single-peaked. If $f(x_i ) \le f(\tilde{x}_i) \le f(x_i+\epsilon')$ we get $0\le f(x_i+\epsilon') -  f(x_i) \le  c'(n^{-b+1/22})|f(\tilde{x}_i) \le c'(n^{-b+\epsilon})|f(x_i + \epsilon')$. Thus,  $f(x_i) \le f(x_i + \epsilon) \le f(x_i)\left(\frac{1}{1 - c'(n^{-b+\epsilon} }\right) \le  f(x_i) (1 + c''n^{-b+\epsilon})$.

Similarly, if $f(x_i ) \ge f(\tilde{x}_i) \ge f(x_i + \epsilon) $, then we get  $f(x_i) \ge f(x_i + \epsilon) \ge f(x_i)\left(\frac{1}{1 + c'(n^{-b+\epsilon} }\right) \ge  f(x_i) (1 - c'n^{-b+\epsilon})$.

    \item  Further,  by our assumption that $f$ has exponential tails we have $\Pr_{\vec{x}}[ \exists x_i \text{ s.t. } |x_i|> 4(\log n)^2] \le n{e^{-4(\log n)^2}}  = ({n^{-4\log n + 1}})$.
\end{itemize}

If $x_i \in [-4(\log n)^2,4(\log n)^2] $ for all $i$, then
{we bound $\prod_{x_i\in \vec{x}}f(x_i +\epsilon') = \prod_{i}f(x_i)(1+ O(n^{-b +\frac{1}{22}}))
= \left(\prod_{i}f(x_i)\right)(1+ O(n^{-b +\frac{1}{22}}))^n = \left(\prod_{i}f(x_i)\right)(1+ O(n^{1-b +\frac{1}{22}}))  =  \left(\prod_{i}f(x_i)\right)(1+ O(n^{-\frac{1}{22}})) $ for $b = 12/11$ }.

Hence, we get
$$
    \tau(\overline{\theta}+\epsilon') = \int_{-n^{1/22\alpha}}^ {n^{1/22\alpha}} \left(\prod f(x_i)\right)(1+ O(n^{-\frac{1}{22}}))\one\{\vec{x}\in X(\hat{\theta})\} d\vec{x} \qquad + O({n^{-4\log n + 1}})
$$

Therefore,  $\tau(\overline{\theta}+\epsilon') = \tau(\overline{\theta})(1+ O(n^{-\frac{1}{22}})) + O({n^{-4\log n + 1}})$ for all $0<\epsilon' <1/n^b$ and $\overline{\theta}\in C$.

\end{proof}

Observe that, if $\tau(\overline{\theta}) \ge c' n^{1/22} n^{-4\log n + 1}$ then we can get $\tau(\overline{\theta})(1+ O(n^{-\frac{1}{22}}))$.

Recall that  $W_I\subset \mathcal{P}$ is the set of all intervals $C$ such that $\tau(\overline{\theta}) < c' n^{1/22} n^{-4\log n + 1} $ for all $\overline{\theta}\in C$ and some constant $c'>0$. We show that the total probability mass of these intervals is $O(n^{ - 4\log n +1})$.
\begin{claim}\label{claim:weak-interval-probability}
$\Pr[\theta_S \in W_I] \le O(n^{ - 4\log n +1})$.
\end{claim}

\begin{proof}
This is simply because there are at most $2n^{b-\frac{1}{2}+\epsilon}$ many intervals in $\mathcal{C}$ (since each interval is of size $n^{-b}$). Therefore,
\[
\Pr[\theta_S \in W_I] \le \sum_{C\in W_I} \int_{ C} \tau({\overline{\theta}}) d\overline{\theta} \le 2n^{b-\frac{1}{2}+\epsilon} \left(n^{-b}\cdot c' n^{1+1/22 - 4\log n}\right) \le O(n^{ - 4\log n + 6/11 + \epsilon}).\]
\end{proof}

Using Corollary~\ref{cor:posterior-change} we bound the correlation between the events $\theta_S\in C$ and any realization of $k$ anecdotes in $L_\delta$ (and $R_\delta$).

\begin{lemma}{\label{lem:posterior-correlation-upper}}
Fix any $C \in \mathcal{P}\setminus W_I$.  For $A\subseteq [n]$, $|A| = k$, let $K_A$ denote the event that $x_i \in L_\delta$ (and $R_\delta$) iff $i\in A$. Let $\vec s \in L_\delta^k$ be a set of $k$ (at most $c'n^{1-a}$) signals in $L_\delta$. Then $\Pr[\theta_S \in C~|K_A, \vec x_A = \vec s] \le \Pr[\theta_S \in C\pm kn^{-1-a}|K_A]$.
\end{lemma}

\begin{proof}
Let $A(\vec z)= \{ \tilde{z}\in \reals^{n-k} : \theta(\vec z\cup \tilde{z})\in C \}$ for any subset of $k$ signals $\vec z \in L_\delta^k$. By Corollary~\ref{cor:posterior-change} we know that by changing $\vec s$ to any $\vec z$ ( in $L_\delta$) for each $\tilde{s}\in A(\vec s)$ $\theta_S(\tilde{s},\vec{z}) = \theta_S(\tilde{s},\vec{s}) + O(kn^{-1-a+\epsilon})$. If $\theta_S(\tilde{s},\vec{s})\in C$ then $\theta_S(\tilde{s},\vec{y})\in C\pm kn^{-1-a+\epsilon}$. 

This implies, $\Pr[\theta_S \in C| K_A, \vec x_A = \vec s] \le\Pr[\theta_S \in C\pm kn^{-1-a+\epsilon}| K_A, \vec x_A = \vec y]$ for all $\vec z\in L_\delta^k$. Hence, we get,
\[
\Pr[\theta_S \in C~|K_A, \vec x_A = \vec s] \le \Pr[\theta_S \in C\pm kn^{-1-a}|K_A]
\]

\end{proof}

Similarly, we have a lower bound on $\Pr[\theta_S \in C | K_A, x_A=\vec s]$.

\begin{lemma}\label{lem:posterior-correlation-lower}
Fix any $C \in \mathcal{P}\setminus W_I$.  For $A\subseteq [n]$, $|A| = k$, let $K_A$ denote the event that $x_i \in L_\delta$ (and $R_\delta$) iff $i\in A$, and no signals in $C_\delta$. Let $\vec s \in L_\delta^k$ be a set of $k$ (at most $c'n^{1-a}$) signals in $L_\delta$.
Then $\Pr[\theta_S \in C | K_A, x_A = \vec s] \ge \Pr[\theta_S \in C\setminus E | K_A] $, where $E\subset C$ as the union of the first and last $kn^{-1-a+\epsilon}$ length sub-interval of $C$.
\end{lemma}

\begin{proof}
Let $A(\vec z)= \{ \tilde{z}\in \reals^{n-k} : \theta(\vec z\cup \tilde{z})\in C \}$ for any subset of $k$ signals $\vec z \in L_\delta^k$. By Corollary~\ref{cor:posterior-change} we know that by changing $\vec s$ to any $\vec z$ ( in $L_\delta$) for each $\tilde{s}\in A(\vec s)$ $\theta_S(\tilde{s},\vec{z}) = \theta_S(\tilde{s},\vec{s}) + O(kn^{-1-a+\epsilon})$. If $\theta_S(\tilde{s},\vec{s})\in C$ then $\theta_S(\tilde{s},\vec{y})\in C\pm kn^{-1-a+\epsilon}$. Thus, $\theta_S(\tilde{s},\vec{z})\in C\setminus E$ then $\theta_S(\tilde{s},\vec{s})\in C$.

This implies, $\Pr[\theta_S \in C| K_A, \vec x_A = \vec s] \ge\Pr[\theta_S \in C\setminus E| K_A, \vec x_A = \vec z]$ for all $\vec z\in L_\delta^k$. Hence, we get,
\[
\Pr[\theta_S \in C~|K_A, \vec x_A = \vec s] \ge \Pr[\theta_S \in C\setminus E|K_A]
\]

\end{proof}